%% file: ITW_full_version_2.tex
\documentclass[peerreview]
{IEEEtran}

\usepackage[square,numbers,sort&compress]{natbib}
\usepackage{hyperref} %
\usepackage%
{hyperref} %
\hypersetup{
    colorlinks,
    linkcolor={blue!80!black},%
    citecolor={green!50!black},
    urlcolor={blue!80!black}
}

\usepackage{amsthm}
\usepackage{amsmath}
\usepackage{amsfonts}
\usepackage{amssymb}
\usepackage{mathrsfs}
\usepackage{amssymb} 
\usepackage{stackrel}
\usepackage{mathtools}
\usepackage{nicefrac}

\usepackage{physics}

\usepackage{verbatim}
\usepackage{enumerate}

\usepackage[T1]{fontenc}
\usepackage{bbm} 
\usepackage{bbding}
\usepackage{keystroke}
\usepackage{pifont}
\usepackage{relsize} % math font size

\usepackage{graphicx}
\graphicspath{{../}}
\makeatletter\def\input@path{{../}}\makeatother
\usepackage{xcolor}
\usepackage{tikz}
\usetikzlibrary{positioning,fit,backgrounds,calc}
\usetikzlibrary{graphs}
\usepackage{tabularx}
\usetikzlibrary{arrows,shapes,positioning}

\usepackage{float}
\usepackage{bm}
\usepackage{graphicx}
\newcommand\mcH{\mathcal{H}}
\newcommand\mcE{\mathcal{E}}
\newcommand\mcN{\mathcal{N}}
\newcommand\mcF{\mathcal{F}}

\newcommand\mcC{\mathcal{C}}

\newcommand\mcX{\mathcal{X}}
\newcommand\mcT{\mathcal{T}}

\newcommand\mcD{\mathcal{D}}

\newcommand\mbfC{\mathbf{C}}

\newcommand\mbE{\mathbb{E}}

\newcommand\mbone {\mathbbm{1}}

\newcommand\id{\text{id}}
\renewcommand\trace {\mathrm{Tr}}

\newcommand\capc{C}
\newcommand\rqad{\mathsf{R}}

\theoremstyle{remark}	\newtheorem{theorem}{Theorem}
\theoremstyle{remark}	\newtheorem{lemma}[theorem]{Lemma}
\theoremstyle{remark}	
\theoremstyle{remark}	\newtheorem{proposition}[theorem]{Proposition}
\theoremstyle{remark} \newtheorem{definition}{Definition}
\theoremstyle{remark} \newtheorem{remark}{Remark}
\theoremstyle{remark} \newtheorem{example}{Example}

\allowdisplaybreaks

\title{Quantum Action-Dependent Channels}

\author{
    \IEEEauthorblockN{Michael Korenberg and Uzi Pereg}
    }

\begin{document}

%----------------------------------------------
% Title
%----------------------------------------------
\maketitle
\setcounter{page}{1} % Start numbering from 1 here
\thispagestyle{plain}

%----------------------------------------------
% Abstract
%----------------------------------------------
\begin{abstract}
We study communication over a quantum action-dependent channel, where the transmitter first performs an action that ``shocks" the channel environment, and subsequently encodes a message into a transmission sent through the channel. This two-stage interaction arises in various settings, including rewriting over defective memory and quantum effects such as measurement-induced state collapse. Our model can be viewed as a quantum generalization of Weissman's classical action-dependent channel (2010). Here, however, Alice cannot have a copy of the environment state due to the no-cloning theorem. Instead, she may share entanglement with this environment. 
We derive achievable rates for reliable message transmission via the quantum action-dependent channel, with either causal or non-causal channel side information (CSI). %For the non-causal setting, we further provide a matching regularized converse.
 As a case study, we analyze memory storage with depolarization and selective rewriting, demonstrating how action-dependent control influences performance.
% We study the quantum action-dependent channel. The model can be viewed as a quantum analog of the classical action-dependent channel model. 
% In this setting, the communication %Gel'fand-Pinsker 
% channel has two inputs: Alice's transmission and the input environment. 
% The action-dependent mechanism enables the transmitter to influence the channel's environment through an action channel.
% Specifically, Alice encodes her message into a quantum action,  which subsequently affects the environment state. For example, a quantum measurement at the encoder can induce a state collapse of the environment.
% In addition, Alice has access to side information. 
% Unlike the classical model, she cannot have a copy of the environment state due to the no-cloning theorem.
% Instead, she shares entanglement with this environment.
% % ,as a result, her measurement can lead to a collapse of the environment state. %the states of the channel are considered as quantum entangled states shared between the channel and the transmitter. 
% % We establish an achievable communication rate for reliable message transmission via the quantum action-dependent channel, thereby extending the classical action-dependent framework to the quantum domain.
% We establish an achievable communication rate for reliable message transmission via the quantum action-dependent channel with either causal or non-causal side information, thereby extending the classical action-dependent framework to the quantum domain.
% We further demonstrate our results for memory with depolarization and selective rewrite.
\end{abstract}


%----------------------------------------------
% Introduction
%----------------------------------------------
\section{Introduction} 
\label{section:Introduction}
% A fundamental problem in information theory is the characterization of reliable communication over channels affected by \emph{random parameters}, often referred to as channel states \cite{Pradhan2020,Zivarifard2022,
% Kim2023,
% Belzig2024,Boche2024,
% Wang2024,Luo2024,Chen2025}. 
% Beginning with Shannon’s seminal work on channels with side information \cite{Shannon1958}, the study of channels with random parameters has revealed the crucial 
% role of side information at the encoder or decoder. Gel’fand and Pinsker \cite{pinsker1980coding} established the capacity for channels with non-causal side information at the 
% encoder. Costa’s ``writing on dirty paper'' result \cite{Costa1983} further extended it to Gaussian channels.
% These works lay the foundation for a rich literature on random-parameter models in both point-to-point and multi-user 
% settings \cite{%HeegardElGamal1983,
% bennatan2006superposition, ramachandran2019joint,Pereg2019,Ahmadipour2024}.
A fundamental problem in information theory is the characterization of %reliable 
communication over channels affected by random parameters \cite{Pradhan2020,Zivarifard2022,
% Kim2023,
Belzig2024,
% Wang2024,Luo2024,
Chen2025,yao2026nonsignaling}. Since Shannon’s seminal work on side information, it has become clear that channel side information (CSI) plays a central role in determining achievable rates \cite{Dupuis2009}. Gel'fand and Pinsker \cite{pinsker1980coding} further characterized the capacity when non-causal CSI is available at the encoder. %, and Costa’s ``writing on dirty paper" result \cite{Costa1983} extended this insight to Gaussian channels. Together, 
These foundational results have given rise to a rich body of work on communication over random-parameter channels~%
% A fundamental problem in information theory is the characterization of reliable communication over channels affected by \emph{random parameters} % , 
% % often referred to as channel states 
% \cite{Pradhan2020,Zivarifard2022,
% Kim2023,
% Belzig2024,
% yao2026nonsignaling,
% % Wang2024,Luo2024,
% Chen2025}. 
% Beginning with Shannon’s seminal work on  side information, % \cite{Shannon1958}, 
% the study of such channels  has revealed the crucial 
% role of CSI %at the encoder or decoder 
% \cite{Dupuis2009}. Gel’fand and Pinsker \cite{pinsker1980coding} established the capacity for 
% %channels with 
% non-causal CSI at the 
% encoder. Costa’s ``writing on dirty paper'' result \cite{Costa1983}  extended it to Gaussian channels.
% These works lay the foundation for a rich literature on random-parameter models %in both point-to-point and multi-user 
% % settings 
\cite{%HeegardElGamal1983,
% bennatan2006superposition, 
% ramachandran2019joint,
Keshet2008}.

In the random-parameter paradigm% 
%\textcolor{blue}{referring to states is very confusing in the quantum context}
, the parameters are typically drawn from nature and cannot be controlled by the communicating parties \cite{Dupuis2009,pinsker1980coding,Keshet2008}.
The parameters influence the channel by altering its transition law.
In the classical setting, the channel is described in terms of a probability function
$%\begin{align*}
  P_{Y|X,S}(\cdot|x, s)
$, where $X$ represents Alice's transmission 
and $Y$ is Bob's observation at the channel output. The random parameter $S$ has a specified distribution and its variation can significantly affect the channel output. For this reason,   %\end{align*} 
% That is, once the state is realized, 
% it acts as a hidden parameter that determines how the input is mapped to the output. 
% This perspective clarifies why state 
side information, i.e., the knowledge of $S$, %(at encoder, decoder, or both) 
has a profound effect on capacity.

Weissman \cite{weissman2010} introduced the \emph{action-dependent channel}. In this model, the encoder first selects an
action sequence, which in turn, generates the channel parameters in a noisy fashion. 
The overall channel input then depends on both the message and the induced parameters. This two-stage procedure captures a 
broad class of practical problems, such as memories with defects \cite{Kim2016}, magnetic recording with rewriting \cite{HeegardElGamal1983}, and other scenarios in which the transmitter can probe or partially
control the channel before communication, as was demonstrated in \cite{kittichokechai2012multi}. For example, a two-stage coding strategy can steer a defective memory to improve reliability: first, the transmitter writes to the memory and immediately tries to read it back to learn about defects. Then, the transmitter rewrites the defective bits based on that information.
Extensions include the probing capacity, which quantifies the information that %channel state 
can be learned through actions 
\cite{asnani2011probing} %, 
and %models with 
action cost constraints \cite{kittichokechai2015coding}.
The framework has also been generalized to multi-user communication, with generalizations to broadcast channels \cite{steinberg2012degraded, ahmadi2012channels} and multiple-access channels \cite{dikstein2014mac}. 
Furthermore, its implications for security have been explored in the context of wiretap channels \cite{dai2013wiretap, dai2020impact} and other secure communication
settings \cite{welling2024transmitter, zivarifard2025covert}, where the action channel acts as a broadcast channel influenced by the transmitter's actions.

The ongoing development of quantum information theory is foundational for engineering next-generation communication and computation systems \cite{Jouguet2013Experimental, Orieux2016Recent, Petit2020Universal,Wang2022,Bassoli2023}. By leveraging the principles of quantum mechanics, this field aims to overcome the limitations of classical technologies. 
Quantum technology also unlocks entirely new phenomena with no classical parallel, such as entanglement, i.e., the strongest resource of quantum correlation,
as well as the no-cloning theorem, which forbids the perfect duplication of  quantum information, motivating a deeper study of fundamental communication limits.
Action dependence appears in quantum communication as well. For example, a quantum measurement by the encoder on the transmission system could result in a state collapse of the channel input environment.

Quantum environment-dependent channels are particularly relevant to scenarios that involve not only the transmission of a message but also parameter estimation, a central task in fields such as quantum metrology. These types of quantum Gel'fand-Pinsker channels have been studied, both with and without entanglement assistance
\cite{Dupuis2009} (see also \cite{pereg2019entanglement}).
% \textcolor{blue}{add recent works by Jafar on CSI, nonsignaling (arXiv:2506.17803, arXiv:2602.11568), entanglement assistance (arXiv:2603.20416) }
Recently, Yao and Jafar studied the role of nonsignaling correlations and quantum entanglement assistance for classical channels with CSIT, showing that nonsignaling assistance can virtually ``teleport'' the transmitter's state knowledge to the receiver in both causal and non-causal settings \cite{yao2025virtual, yao2026nonsignaling}, and that entanglement assistance can strictly improve the capacity and even activate the
zero-error capacity of classical channels with causal state information \cite{yao2026quantum}.
The security implications of side information have also been explored in the quantum setting
through wiretap channels \cite{Anshu2020} and covert communication \cite{zivarifard2024covert}. 
Other variations include scenarios in which the decoder performs parameter estimation \cite{Pereg2022}. 
% The action-dependent framework has not been studied in the quantum literature thus far.

Beyond passive environment dependence, many emerging quantum technologies involve active manipulation of the communication environment by the transmitter.
For example, in joint quantum communication and sensing \cite{Wang2022ITW,liu2024quantum},
adaptive probing signals are used both to estimate properties of the physical medium and to transmit information.
Because probing a quantum system generally disturbs its state through measurement backaction,
the transmitter’s actions directly influence the effective communication channel.
Similarly, quantum error mitigation protocols for noisy intermediate-scale quantum devices
\cite{endo2018practical,cai2023quantum}
use controlled interventions to characterize and partially shape the effective noise process prior to computation or communication.
In such settings, the communication environment is not merely an externally given random process,
but is itself influenced by the transmitter’s actions.
These scenarios naturally motivate a quantum action-dependent channel framework.

%Channel figure
\begin{figure*}[tb]
  \centering
  \resizebox{0.85\textwidth}{!}{% Wrap the TikZ figure
      \input{vertical_fig_QADC.tex} % Include the TikZ code here
  }
  \caption{Coding over a quantum action-dependent channel with non-causal CSI. Here Alice acts as the Action encoder, encoding the 
  message $M$ into an action sequence $G^n$, and the main encoder, encoding the message and side information $S_0^n$ into 
  the channel input $A^n$. The action sequence $G^n$ is fed into the action channel $\mcT^{\otimes n}_{G \to S_0}$,  with the Stinespring dilation $T^{\otimes n}_{G \to SS_0}$ which produces the environment state $S^n$ and side-information $S_0^n$ for Alice. The quantum communication channel $\mcN^{\otimes n}_{SA \to B}$ 
  takes the environment state $S^n$ and input $A^n$, producing the output $B^n$, which is measured by Bob to decode the message.}
  \label{fig:environment-dependent-channel}
\end{figure*}

In this paper, we study the quantum action-dependent channel. Our action dependence model modifies the standard environment-dependent paradigm by allowing the transmitter's actions to ``shock" the %influence a 
quantum environment, which in turn governs the channel transformation. 
% In this paper, we introduce the quantum counterpart of the action-dependent channel, which we term as the action-dependent 
% quantum channel.
Specifically, an encoder (Alice) first encodes a classical message $M$ into an action sequence $G^n$, which is fed into a quantum action channel 
$\mcT^{\otimes n}_{G \to S_0}$, associated with a Stinespring representation isometry $T_{G \to SS_0}$. See Figure~ \ref{fig:environment-dependent-channel}. 
% The action channel produces side information $S_0^n$, for Alice, and an environment system $S^n$.
The action channel produces side information $S_0^n$ for Alice, while its 
Stinespring isometry jointly produces the pure state $\ket{\sigma_{SS_0}^{\otimes n}}$
on the side-information system $S_0^n$ and the inaccessible channel environment $S^n$.
Alice then encodes the message and the side information into her transmission $A^n$, which is sent through the quantum communication channel $\mcN^{\otimes n}_{SA \to B}$.
The receiver (Bob) obtains the output sequence $B^n$, and performs a measurement in order to estimate Alice's message.

Our framework can be viewed as the quantum counterpart of the classical action-dependent channel introduced by Weissman \cite{weissman2010}. However, the generalization is nontrivial due to fundamental quantum principles such as the no-cloning theorem. While a classical channel parameter can be perfectly copied and then sent back to the transmitter, an unknown quantum state cannot. 
% Thereby, side information is modeled through quantum entanglement shared between two distinct systems, $S$ and $S_0$, where $S$ represents the environment affecting the channel, and $S_0$ is the side information available to Alice.
Thereby, side information is modeled as the output of  %the Stinespring dilation  %$T_{G \to SS_0}$ 
%of 
the action channel $\mcT_{G \to S_0}$. The isometric extension (Stinespring dilation) produces a pure bipartite state $\sigma_{SS_0}$ on two distinct systems: $S$, the inaccessible environment that affects the channel, and $S_0$, the side 
information available to Alice.

% \textcolor{blue}{Again, a reviewer would criticize this and say that an extension is not worthwhile.} 
% This carries the same intuition of Weissman's classical model into the quantum domain, replacing classical probe-rewrite operations with quantum states, 
% channels, and entanglement. 
% In our model, an encoder (Alice) first encodes a message $M$ into an action sequence $G^n$, which is fed into a quantum 
% channel $\mcT^{\otimes n}_{G \to S S_0}$. This channel produces side information $S_0^n$ for Alice and an environment state $S^n$. 
% Alice then encodes the message and side information into a sequence $A^n$, which is transmitted through the main quantum channel $\mcN^{\otimes n}_{SA \to B}$.
% The receiver (Bob) obtains the output sequence $B^n$. 
% To analyze this model, we employ one-shot information-theoretic 
% techniques, as developed in recent works on quantum state-dependent channels \cite{Anshu2020, zivarifard2024covert}.
% \textcolor{blue}{Again, there is too much emphasis on the previous contributors. This is not informative. You should describe the actual techniques and give citation where it's needed. The reader also expects an explanation of how your analysis differs from the previous results. In what ways your model is more challenging than the previous ones?}
% We derive an achievable rate for transmitting classical information through such channels, providing a foundational step towards understanding this new class of quantum channels.
We derive achievable rates and regularzied capacity formulas for both the causal and noncausal settings. Our achievability proof is based on quantum one-shot information-theoretic methods.
As opposed to the classical analysis~\cite{weissman2010}, our techniques establish non-asymptotic performance bounds by directly analyzing the error probability for a finite number of channel uses, rather than relying on asymptotic arguments. The introduction of action dependence makes our analysis more challenging than in previous environment-dependent channel models \cite{Anshu2020, Pereg2022}. In previous settings, the channel environment and side information are set by an external source. Here, however, the transmitter's action \textit{induces} the shared entangled state. Consequently, our analysis must account for this additional layer of control. The formula of our capacity bound thus includes optimization over not only the input state, but also a quantum ensemble for the sender's action.

% The rest of the paper is organized as follows. In Section~\ref{section:Notation}, we introduce the notation and definitions used throughout the paper. 
% In Section~\ref{section:ADCoding}, we formally define the quantum action-dependent channel model and then present our main result, an achievable rate for this channel, in Section \ref{section:Main_Result}.
% In Section~\ref{section:examples}, we consider memory with depolarization and selective rewrite, and study the implications of our results on this channel.
% In Section \ref{section:Coding}, we describe the one-shot coding scheme used to prove the achievability result. The detailed proof is provided in Section \ref{section:proof of main theorem},
% with key lemmas proved in Appendices~\ref{appendix:proof_of_lem1}--\ref{appendix:cardinality}. The detailed derivations for the examples are given in Appendix~\ref{appendix:example_derivations}.
% %
% Section~\ref{section:summary} provides summary and discussion on open problems and applications, such as joint quantum communication and sensing and quantum error mitigation.
The rest of the paper is organized as follows. In Section~\ref{section:Notation}, we introduce the notation and definitions used throughout the paper. In Section~\ref{section:ADCoding}, we formally define the quantum action-dependent channel model. In Section~\ref{section:Main_Result}, we present our main result, an achievable rate for this channel. In Section~\ref{section:examples}, we consider memory with depolarization and selective rewrite, and study the implications of our results on this channel. In Section~\ref{section:Coding}, we describe the one-shot coding scheme used to prove the achievability result. Sections \ref{section:proof of main theorem} and \ref{section:proof of causal theorem} 
present a detailed analysis for the non-causal and causal settings, respectively. Section~\ref{section:summary} provides a summary and discussion on open problems and applications, such as joint quantum communication and sensing and quantum error mitigation. Key lemmas are proved in Appendices~\ref{appendix:proof_of_lem1}--\ref{appendix:cardinality}, and detailed derivations for the examples are given in Appendix~\ref{appendix:example_derivations}.

% \ref{appendix:Th1}--
%----------------------------------------------
% Notation and Definitions
%----------------------------------------------

\section{Notation and Basic Definitions}
\label{section:Notation}
\subsection{Quantum States and Channels}
Quantum systems are denoted by uppercase letters 
(e.g., $A, B$) and their corresponding finite-dimensional Hilbert spaces by $\mcH_A, \mcH_B$. 
The corresponding dimension is denoted by $\abs{\mcH_A}$.
The set of density operators on $\mcH_A$ is $\mathscr{D}(\mcH_A)$. 
Quantum states (density operators) are denoted by Greek letters, e.g., $\rho, \sigma$. %, and pure states by $\ket{\psi}, \ket{\phi}$. The identity operator is $\mbI_X$.
A POVM is a set of positive semi-definite operators $\{D_m\}$ that satisfy $\sum_m D_m = \mbone$, where $\mbone$ denotes the identity operator. If the quantum state before the measurement is $\rho$, then the probability of an outcome $m$ is $\Pr(m) = \Tr(D_m \rho)$.
% \begin{align*}
% \sum_m D_m = \mbI.
% \end{align*}
% The probability of obtaining outcome $m$ when measuring a state $\rho$ is given by
% \begin{align*}
% p(\hat{M} = m) = \Tr(D_m \rho).
% \end{align*}

A quantum channel $\mcN_{A\to B}$ is a completely positive trace-preserving (CPTP) map. We write $\id_A$ for the identity channel on system $A$. Here, we consider a quantum channel $\mcN_{SA \to B}$, where $A$, $S$, and $B$ are associated with the transmitter (Alice), the channel environment (``channel state"), and the receiver (Bob).
The channel can be represented through its Stinespring dilation, in terms of an isometry $V_{SA \to BE}$ that couples the output system to an auxiliary environment $E$. 
Namely,
\begin{align}
    \mcN_{SA \to B}(\rho_{SA}) = \Tr_{E} \left[V \rho_{SA} V^{\dag}\right].
\end{align}
for $V\equiv V_{SA\to BE}$ that satisfies
$V^\dagger V=\mbone_{SA}$.
We will see that in the action-dependent model, the channel environment $S$ is affected by the encoding operation through an action map $T_{G\to SS_0}:\mathcal{H}_G\to \mathcal{H}_{S}\otimes \mathcal{H}_{S_0}$, as shown in  Figure~\ref{fig:environment-dependent-channel}, where $T_{G\to SS_0}$ is an isometry.
Here, $S_0$ represents channel side information (CSI) that is available to Alice, while $S$ is the channel's inaccessible environment.  
We provide further details and explanation in Section~\ref{section:ADCoding} below.

We assume that the channel is memoryless. That is, once Alice encodes an action sequence
$G^n\equiv G_1,\ldots,G_n$, her action state $\rho_{G^n}$ undergoes the tensor-product isometry $T_{G\to SS_0}^{\otimes n}$. Similarly, the input sequence $(S^n,A^n)$ is transmitted through the channel $\mathcal{N}_{SA\to B}^{\otimes n}$.
% Here, we consider the case 
%

\subsection{Information Measures}
For a quantum state $\rho \in \mathscr{D}(\mcH)$, the von Neumann entropy is
\begin{align}
    H(\rho) = -\Tr(\rho \log \rho).
\end{align}
For a bipartite state $\rho_{AB}\in \mathscr{D}(\mcH_A\otimes\mcH_B)$, the quantum mutual information is defined as:
\begin{align}
    I(A;B)_{\rho} = H(\rho_A) + H(\rho_B) - H(\rho_{AB}),
\end{align}
and the conditional entropy as $H(A|B)_{\rho} = H(\rho_{AB}) - H(\rho_B)$. Unlike its classical counterpart, the quantum conditional entropy can be negative.

The quantum relative entropy between two states $\rho$ and $\sigma$ in 
$\mathscr{D}(\mathcal{H})$
is defined as $D(\rho \| \sigma) = \Tr\left(\rho \left(\log\rho - \log\sigma\right)\right)$ if 
$\mathrm{supp}(\rho) \subseteq \mathrm{supp}(\sigma)$, and $D(\rho \| \sigma)=+\infty$ otherwise.
The sandwiched R\'enyi divergence \cite{MuellerLennert2013} is defined as 
 %   For $\rho, \sigma \in \tilde{D}(\mcH)$,
 \begin{align}
     \tilde{D}_{\alpha}\left(\rho \| \sigma \right) \coloneqq  \frac{1}{\alpha - 1} \log \Tr[\sigma^{\frac{1-\alpha}{2\alpha}} \rho \sigma^{\frac{1-\alpha}{2\alpha}}]^{\alpha}.
 \end{align}

    % It is monotonic under CPTP maps: $\tilde{D}_{\alpha}\left(\mcN(\rho) \| \mcN(\sigma) \right) \leq D_{\alpha}\left(\rho \| \sigma \right)$.

Other key definitions are given below:
\begin{enumerate}
    \item \textbf{Pinching}  \cite{tomamichel2015quantum}:
      For a Hermitian operator $A = \sum_{i} a_i \Pi_i$ with projectors $\Pi_i$ in its eigenspaces, the pinching map is
      \begin{align}
        \mathcal{E}_A(B) := \sum_{i} \Pi_i B \Pi_i.
      \end{align}
      This operation has the properties of a quantum channel and projects $B$ onto a block-diagonal structure dictated by the eigenspaces of $A$. % (see Figure~\ref{fig:pinching}).
      One of the properties of the pinching map is that the resulting operator, $\mathcal{E}_A(B)$, always commutes with $A$, i.e., $[A, \mcE_A(B)] = 0$.
      Another key property of a pinching map is the pinching inequality. Let $\nu_A$ be the number of distinct nonnegative eigenvalues of $A$, then:
      \begin{align}
          B \leq \nu_A \mcE_A (B). \label{eq:pinching_inequality}
      \end{align}
      % is an operation on operators that projects them onto subspaces where a given operator has definite values. 
      % Formally, let $H$ be a finite-dimensional Hilbert space and $A \in L(\mcH)$ be a Hermitian operator on $\mcH$. 
      % Suppose $A$ has the spectral decomposition
      % \begin{equation*}
      %   A = \sum_{a_i \geq 0} a_i \Pi_i= \sum_{a_i \geq 0} a_i \ketbra{u_i}{u_i}\,
      % \end{equation*} 
      % 
      % where $a_1, a_2, \dots, a_r$ are the distinct, and positive, eigenvalues of $A$, and $\Pi_i` is the orthogonal projection onto the eigenspace of $A` with eigenvalue `$a_i`. 
      % $\{{\ket{u_i}}\}$ are the corresponding orthonormal eigenvectors. 
      % The \emph{pinching map} associated with $A$ is the
      % linear map $\mathcal{E}_A: L(\mcH) \to L(\mcH)$ defined by 
      % \begin{equation}
      % \mathcal{E}_A(B) := \sum_{i} \Pi_i B \Pi_i, \qquad B \in L(\mcH)\,. 
      % \end{equation}
    \item \textbf{Fidelity}: For $\rho, \sigma \in \mathscr{D}(\mathcal{H})$,
    \begin{align}
            F(\rho, \sigma) := \left\|\sqrt{\rho} \sqrt{\sigma}\right\|_1.
    \end{align}
        
    % \begin{align}
    %   F(\rho, \sigma) := \left\|\sqrt{\rho} \sqrt{\sigma}\right\|_1^2.
    % \end{align}

    \item \textbf{Purified Distance}:  For $\rho, \sigma \in \mathscr{D}(\mathcal{H})$,
    \begin{align}
            P(\rho, \sigma) := \sqrt{1-F^2(\rho, \sigma)}.
    \end{align}
\end{enumerate}
These quantities provide a geometric measure of distance in the space of density matrices.
\section{Action-Dependent Coding}
\label{section:ADCoding}
Before presenting our main results, we introduce a code for the transmission of messages via a quantum action-dependent channel, where the encoder selects an action that affects the channel environment.
% \textcolor{red}{In addition, reminder the reader what action dependence means.} 
Specifically, Alice has two roles: she  encodes both the action $G$ and the transmission $A$ through the channel. Her action encoder sends $G$ through % encodes %the message into 
% an action sequence, 
% which is fed into 
an action channel, which is associated with a Stinespring isometry $T_{G \to SS_0}$. This produces
Alice's CSI $S_0$ and the inaccessible channel environment $S$. % state for the main channel and side information back to Alice.
%
% We denote the action-dependent channel by $\mcN_{SA \to B}\circ \mcT_{G \to S S_0}$.
%
\begin{definition}[Action-Dependent Code]
  An $(M,n)$ code for communication over a quantum action-dependent channel, $\mcN_{SA \to B}$ governed by an action map %channel $\mcT_{G \to S_0}$ with Stinespring dilation 
  $T_{G \to SS_0}$, with non-causal CSI at the encoder, consists of:
  \begin{enumerate}
      \item An \emph{encoder} that comprises % that maps each message $m \in \msM = \{1, \dots, M\}$ and the action $\ell  \in \msL = \{1, \cdots, L\}$ to a sequence of quantum states. This involves 
      two stages:
      \begin{itemize}
          \item a classical-quantum action encoder $\mathcal{L}: \{1,\ldots,M\}\to \mathscr{D}(\mcH_G^{\otimes n}) $ %$\mcJ: \msM \to G^n$ 
          % that prepares a pure quantum action state $\ket{\rho_{G^n}^{(m)}} \in \mathscr{D}(\mcH_G^{\otimes n})$, for $m\in \{1,\ldots,M\} =\msM$.
          that maps the message to a pure quantum action state. %prepares a pure quantum action state $\ket{\rho_{G^n}^{(m)}} \in \mcH_G^{\otimes n}$, for $m\in \{1,\ldots,M\} =\msM$.
          \item a transmission encoder $\mcE^{(m)}_{S_{0}^{n} \to A^n}$  that receives the side-information $S_0^n$ and prepares the channel input $A^n$. %from the channel $\mcT^{\otimes n}_{G \to S S_0}$, 
          %prepares the input state $\rho_{S^n A^n}^{(m, \ell)} \in \mathscr{D}(\mcH_S^{\otimes n} \otimes\mcH_A^{\otimes n})$ for the main channel $\mcN^{\otimes n}_{SA \to B}$.
      \end{itemize}
      \item A \emph{decoding measurement}, i.e., a POVM $\{D_m\}_{m=1}^M$ on the output Hilbert space $\mcH_B^{\otimes n}$.
  \end{enumerate}
\end{definition}
The coding scheme works as shown in  Figure~ \ref{fig:environment-dependent-channel}. Alice selects a uniform message $m\in\{1,\ldots,M\}$.
She first prepares the state of her quantum action $G^n$  using her action encoder:
\begin{align}
\rho_{G^n}^{(m)}=\mathcal{L}(m) \,.
\end{align}
Each quantum action $G_i$ is then sent through the action map $T_{G\to SS_0}$, producing
\begin{align}
\rho_{S^n S_0^n}^{(m)}=T^{\otimes n}_{G \to S S_0} {\rho_{G^n}^{(m)}} (T^{\otimes n}_{G \to S S_0})^\dagger .
\end{align}
Given her access to the side information $S_0^n$, Alice applies the transmission encoder to prepare her input $A^n$ to the communication channel:
\begin{align}
\rho_{S^n A^n}^{(m)}=\mathrm{id}_{S^n}\otimes \mcE^{(m)}_{S_{0}^{n} \to A^n}
(\rho_{S^n S_0^n}^{(m)}).
\end{align}
Both the input environment $S^n$ and Alice's transmission $A^n$ are fed into the channel $\mathcal{N}_{SA\to B}^{\otimes n}$,
hence
\begin{align}
\rho_{B^n}^{(m)}=\mcN^{\otimes n}_{S A\to B}
(\rho_{S^n A^n}^{(m)}).
\end{align}
Bob receives $B^n$. He performs the measurement $\{D_m\}$ to obtain an estimate of Alice's message.  

For an $(M,n, \varepsilon)$ code, the average probability of error is bounded by %, $\bar{p}_e$, is required to be no greater than 
$\varepsilon$, i.e.,
      \begin{align}
        \bar{p}_e^{(n)} \coloneqq 1 -\frac{1}{M}\sum_{m=1}^M \Tr\Bigl[D_m\, \rho_{B^n}^{(m)}\Bigr] \leq \varepsilon. \label{eq:avg_error_defenition}
      \end{align}

\begin{definition}[Achievable Rate]
  A communication rate $R$  is said to be achievable %, with respect to the action channel $\mcT_{G \to S S_0}$, 
  if for every $\varepsilon,\delta>0$ and sufficiently large $n$, there exists a $(2^{n(R-\delta)},n, \varepsilon)$ code for the quantum action-dependent channel $\mcN_{SA \to B}\circ T_{G \to S S_0}$ with non-causal CSI.
%  
  % there exists a sequence of  $(M,n, \varepsilon_n)$ codes
  % such that $\lim_{n\to\infty} \varepsilon_n = 0$ and $\liminf_{n\to\infty} \frac{1}{n}\log M \geq R$. 
  The channel capacity $\capc_{\text{n-c}}(\mcN\circ T)$  is defined as the supremum of all achievable rates, where the subscript `n-c' stands for non-causal CSI. 
\end{definition}

\begin{remark}
% \textcolor{blue}{Expalin what happens without side information (Alice does not have access to $S_0$): Effectively $\bar{\mathcal{N}}_{A\to B}=\mathcal{N}_{SA\to B}\circ \mathrm{Tr}_S$}
The setting of an environment-dependent channel $\mathcal{N}_{SA\to B}$ where the encoder cannot influence the channel environment can be viewed as a special case of our model, 
under the constraint that  $\abs{\mathcal{H}_G}=1$. 
% Several works have consider this environment-dependent model
% \textcolor{blue}{add ref to all works that considered CSI: Dupuis, Pereg on entanglement-assisted communication, parameter estimation, quantum channel with classical state masking (PRA), quantum state masking (T-IT), Anshu, ... }.
This environment-dependent model has been studied in several works, including entanglement-assisted communication with CSI~\cite{Dupuis2009, Dupuis2010t, pereg2019entanglement}, parameter estimation~\cite{Pereg2022}, classical~\cite{pereg2022classical} and quantum state masking~\cite{pereg2021quantum}, secure communication over the quantum Gel'fand-Pinsker wiretap channel~\cite{Anshu2020}, and covert communication over a quantum MAC with a helper~\cite{zivarifard2024covert}.
If, in addition, Alice does not have access to the side information $S_0$, she has no knowledge of the channel environment realization, and thus 
cannot adapt her encoding accordingly. In this case, the environment system $S$ is prepared by the action isometry 
$T_{G \to SS_0}$ but remains inaccessible to Alice, and the side-information system $S_0$ is effectively discarded. 
The model then reduces to that of a ``regular" channel $\bar{\mcN}_{A \to B}$ without environment or action dependence,
where $\bar{\mcN}_{A \to B}(\rho) = \mcN_{SA \to B} \circ \mathrm{Tr}_{S S_0} \left(  T \ketbra{0}_G  T^{\dagger}\otimes \rho \right) $. 
% , i.e., the action channel is affectively
% \begin{align}
%     \mcT_{G \to S}(\rho) = \mathrm{Tr}_{S_0} \circ \left[ T \rho  T^{\dagger}\right].
% \end{align}
% The overall channel reduces to
% \begin{align}
%     \bar{\mcN}_{A \to B} = \mcN_{SA \to B} \circ (\mcT_{G \to S} 
%     \otimes \mathrm{id}_A),
% \end{align}
% Alice's transmission $A^n$ then passes through a fixed environment-dependent channel whose environment she cannot influence through feedback, and the capacity reduces to the Holevo information of $\bar{\mcN}_{A \to B}$, with no gain from non-causal side information.
\end{remark}

\begin{remark}
In the action-dependent setting, if Alice does not have access to CSI, then she simply prepares the action state $\rho_G$ that produces the ``best" channel. 
This can be interpreted as environment \emph{assistance}
% \textcolor{blue}{add references -- Winter}.
\cite{karumanchi2016quantum, karumanchi2016classical}.
\end{remark}

%----------------------------------------------
% causal side information at the encoder
%----------------------------------------------
Next we consider causal CSI, following the definition in \cite{Pereg2022}. In the causal case, Alice has access to the past and present CSI systems.
That is, at time $i$, she has $\left({S_{0,j}}\right)_{j\leq i}$ .
\begin{definition}[Causal Action-Dependent Code]
  An $(M,n)$ code for communication over a quantum action-dependent channel $\mcN_{SA \to B}\circ%$, that is governed by an action map %channel $\mcT_{G \to S_0}$ with Stinespring dilation 
  % $
  T_{G \to SS_0}$ with causal CSI, consists of:
  \begin{enumerate}
      \item Action and transmission \emph{encoders}:
      \begin{itemize}
          \item classical-quantum  action encoder $\mathcal{L}: \{1,\ldots,M\}\to \mathscr{D}(\mcH_G^{\otimes n}) $. %that prepares a quantum action state $\ket{\rho_{G^n}^{(m)}} \in \mcH_G^{\otimes n}$, for $m\in \{1,\ldots,2^{nR}\} =\msM$.
          \item a sequence of  transmission encoders $\mcE^{(m,i)}_{S_0^i \to A^i}$, each acts on   the CSI systems $S_0^i$, up to time $i$, to prepare the channel input $A^i$, for $ i \in  \{1,\ldots,n-1\}$. The encoder sequence must satisfy the causality constraint:
          \begin{align}
            \mcE^{(m,i)}_{S_0^i \to A^i} \otimes \id_{{S_0}_{i+1}} = \Tr_{A_{i+1}} \circ \; \mcE^{(m,i+1)}_{S_0^{i+1} \to A^{i+1}}.
            \label{eq:causal-encoding}
          \end{align}
      \end{itemize}
      \item \emph{Decoder}, i.e., a POVM $\{D_m\}_{m=1}^{2^{nR}}$ on the output Hilbert space $\mcH_B^{\otimes n}$.
  \end{enumerate}
\end{definition}

The channel capacity $\capc_{\text{caus}}(\mcN\circ T)$ is defined accordingly, where the subscript `caus' stands for causal CSI. 
The capacity notation is summarized in Table~\ref{Table:Capacity}.

\begin{remark}
The causality constraint \eqref{eq:causal-encoding} is equivalent to requiring that for every message $m$ and every $i\in\{1,\ldots,n\}$, the induced joint state satisfies
  \begin{align}
  \rho^{(m)}_{S^i A^i}
  &=
  \id_{S^i}\otimes
  \mcE^{(m,i)}_{S_0^i \to A_1,\ldots,A_i}
  \!\left(\rho^{(m)}_{S^i S_0^i}\right),
  \label{eq:causal-joint}
  \\[0.5em]
  \Tr_{A_{i+1}^n}\!\left[\rho^{(m)}_{A^n}\right]
  &=
  \mcE^{(m,i)}_{S_0^i \to A_1,\ldots,A_i}
  \!\left(\sigma^{(m)}_{S_0^i}\right),
  \label{eq:causal-marginal}
  \end{align}
  where $\sigma^{(m)}_{S_0^i}$ denotes the marginal of the side-information state $\rho^{(m)}_{S_0^n}$ on the subsystems ${S_0}_1,\ldots,{S_0}_i$.
\end{remark}

% At time $i$, the encoder prepares an input state 
%   $\rho_{A^i Q_{i+1} S_{0,i+1}\cdots  S_{0,n} S^n}^{(m)}=  \mathrm{id}_{A^{i-1}}\otimes  \bar{\mathcal{E}}_{Q_{i} S_{0,i}\to A_i Q_{i+1}}^{(m,i)}\otimes \mathrm{id}_{ S_{0,i+1}\cdots  S_{0,n} S^n}
%   (\rho_{A^{i-1} Q_{i} S_{0,i}\cdots  S_{0,n} S^n}^{(m)})$, where $Q_2$ is a memory system. 

\begin{remark}
\label{remark:causal_operational}
  Operationally, causal encoding can be performed using a memory system $Q_i$ at each time step. 
At time $i$, the encoder applies an encoding map $\bar{\mathcal{E}}_{Q_{i} S_{0,i}\to A_i Q_{i+1}}^{(m,i)}$ on the 
memory system $Q_i$ and the fresh side information $S_{0,i}$, which results in the following
 input state: 
  $\rho_{A^i Q_{i+1} S_{0,i+1}\cdots  S_{0,n} S^n}^{(m)}=  \mathrm{id}_{A^{i-1}}\otimes  \bar{\mathcal{E}}_{Q_{i} S_{0,i}\to A_i Q_{i+1}}^{(m,i)}\otimes \mathrm{id}_{ S_{0,i+1}\cdots  S_{0,n} S^n}
  (\rho_{A^{i-1} Q_{i} S_{0,i}\cdots  S_{0,n} S^n}^{(m)})$, where the input $Q_i$ is the memory produced in the previous step, and
  the output $Q_{i+1}$ is to be used in the next step. 
%
  % At time $i=1$, the encoder prepares an input state $\rho_{A_1 Q_1}^{(m)}=\bar{\mathcal{E}}_{A_1 Q_1}^{(m,1)}(1)$, where $Q_1$ is a memory system. At time $i=2$, the encoder prepares an input state 
  % $\rho_{A_1 A_2 Q_1}^{(m)}=\mathrm{id}_{S_1 A_1}\otimes \bar{\mathcal{E}}_{Q_1 S_{0,1}\to A_2 Q_2}^{(m,2)}
  % (\rho_{S_1,S_{0,1}}^{(m)}\otimes \rho_{ A_1 Q_1}^{(m)}\otimes )$, where $Q_2$ is a memory system.
% The causality constraint \eqref{eq:causal-encoding} allows the encoder to prepare the entire input state so far, 
% as long as the marginals remain consistent, rather than acting on a memory system and outputting only the present symbol. 
% The two formulations are equivalent, as
The equivalence between the two formulations follows from Uhlmann's theorem \cite{tomamichel2015quantum}.
% To see this, consider for example $\sigma_S^{\otimes 2}$. Suppose that at time $i=1$, the encoder acts on $S_1$ and prepares 
% $\rho_{A_1 \widetilde{S}_1 S_2}$, where $\widetilde{S}_1$ is a memory system. Then, at time $i=2$, the encoder 
% acts on $\widetilde{S}_1 S_2$ and prepares $\omega_{A_1 A_2}$. Let $\ket{\rho_{A_1 \widetilde{S}_1 S_2 R_1}}$ 
% and $\ket{\omega_{A_1 A_2 R_2}}$ be the respective purifications. The requirement $\rho_{A_1}=\omega_{A_1}$ implies that 
% there exists an isometry that maps from $\widetilde{S}_1 S_2 R_1$ to $A_2 R_2$. Tracing out the reference systems, we obtain an 
% encoding channel from $\widetilde{S}_1 S_2$ to $A_2$.
% Conversely, using our formulation, the encoder acts at time $i=2$ on $S_1 S_2$ and 
% prepares $\omega_{A_1 A_2}$. Let $\ket{\rho_{A_1 S_2 M_1}}$ and $\ket{\omega_{A_1 A_2 R_2'}}$ be the 
% respective purifications, where the reference $M_1$ may be viewed as the memory system. 
% The requirement $\rho_{A_1}=\omega_{A_1}$ implies that there exists an isometry that maps from $M_1 S_2$ to $A_2 R_2'$. 
% Tracing out the reference system, we obtain an encoding channel from $M_1 S_2$ to $A_2$. The same formulation appeared in 
% earlier work \cite{Pereg2022} as well.
\end{remark}

\begin{table}
\caption{Capacity Notation}
\label{Table:Capacity}
\centering
\begin{tabular}{l|lll}
\noalign{\vspace{4pt}}
&Without CSI& Causal CSI & Non-Causal CSI\\ \hline
\noalign{\vspace{2pt}}
No Actions&$C(\mcN)$&$C_{\text{caus}}(\mcN \circ \sigma)$&$C_{\text{n-c}}(\mcN \circ \sigma)$ \\[4pt]
Quantum Action Dependence&$\capc(\mcN \circ T)$&$\capc_{\text{caus}}(\mcN \circ T)$& $\capc_{\text{n-c}}(\mcN \circ T)$
\end{tabular}

\end{table}

% Using the pinch-based decoder $\{\beta(m,i)\}_{(m,i) \in \mcM \times \mcI}$ constructed above, we now derive a bound on the probability of decoding error.
% We focus on the event where the message $m$ is sent but the decoder does not produce $\hat{m}=m$, i.e. $\hat{m}\neq m$.
% This corresponds to the POVM outcome $\beta(m,i)$ given $m$. The (conditional) error probability is 
% \begin{align}
%   p_e &= \Pr\{\hat{M}\neq m \mid M = m\} 
%   \\&= \Pr\left(\text{decoder has decided }\hat{m}\neq m \mid m\right), \label{eq:one_shot_error} \nonumber
% \end{align}
% The average performance in the one-shot scenario is similar to \eqref{eq:avg_error_defenition}, but with $n=1$. 

% \begin{theorem} [One-shot error probability] \label{Theorem:1}
%     The upper bound of the average performance for $\alpha \in \left(0,\frac{1}{2} \right)$ is given by:
%     \begin{align}
%         \bar{p}_e &\leq   12 \cdot \nu_{1}^{\alpha} 2^{\alpha \left[R + R_S - D_{1-\alpha}\left(\rho_{VUB} \| \rho_{VU} \otimes \rho_B \right)\right]} \\
%         &+ \frac{2}{\alpha} \left( \frac{\nu_2}{2^{R_S}}\right)^{\alpha} 2^{\alpha D_{1+\alpha}\left( \rho_{VS|U} \| \rho_{V|U} \otimes \rho_{S|U}\right)} \nonumber
%     \end{align}
% \end{theorem}
% Where $\nu_2$ is the number of distinct eigenvalues of $\rho_{S|U}$.

%----------------------------------------------
% Related Work
%----------------------------------------------
\section{Related Work}

We briefly review related work on the capacity of quantum channels for the transmission of classical information, both with and without environment dependence. 
We distinguish between three main models: a quantum channel $\mathcal{N}_{A\to B}$ without environment dependence, an environment-dependent channel $\mathcal{N}_{SA\to B}\circ \sigma_{SS_0}$, and an action-dependent channel \mbox{$\mathcal{N}_{SA\to B}\circ T_{G\to SS_0}$}. See Table~\ref{Table:Capacity}.

\subsection{Without Environment Dependence}
\label{subsection:Holevo}
First, we consider a channel that is not affected by an external environment $S$. 
Let $\mcN_{A \to B}$ be a quantum channel without environment dependence (say, $\abs{\mathcal{H}_S}=1$). The \emph{Holevo information} of the channel is defined as \cite{wilde2017}
\begin{align}
  \chi\left(\mcN\right) \coloneqq \max_{p_X(x),\, \ket{\phi_A^x}} I\left(X;B\right)_{\rho}
  \label{eq:holevo_info}
\end{align}
where the maximization is over the auxiliary variable $X\sim p_X$ over an alphabet of size $|\mcX| \leq |\mcH_A|^2$, and the input state ensemble
$\left\{ \ket{\phi_A^x}{\phi_A^x}\right\}$, with respect to the classical-quantum state
$\rho_{XB} \equiv \sum_{x \in \mcX} p_X(x) \ketbra{x}{x} \otimes \mcN_{A\to B}\left(\ketbra{\phi_A^x}{\phi_A^x}\right)$  at the channel output. 
% That is, the Holevo information captures the maximum classical mutual information between the input label $X$ and the output system $B$, optimized over all input ensembles.
The HSW theorem \cite{Holevo1998,SchumacherWestmoreland1997} establishes that $\chi\left(\mcN\right)$ is an achievable 
rate for the transmission of classical information over a quantum channel. However, the classical capacity of a general quantum channel requires a \emph{regularization} 
over multiple channel uses \cite{wilde2017}:
\begin{align}
  C(\mcN) = \lim_{n\to\infty}\frac{1}{n}\,\chi\!\left(\mcN^{\otimes n}\right)
  \label{eq:regularized_capacity}
\end{align}
(see Table~\ref{Table:Capacity}).

\begin{remark}
In the classical setting, Shannon's noisy channel coding theorem \cite{shannon1948mathematical} provides a computable single-letter capacity formula. In the quantum setting, however, the expression in \eqref{eq:regularized_capacity} involves an optimization over entangled inputs across $n$ channel uses, which renders the computation generally intractable. The question of whether the capacity $C\left(\mcN\right)$ equals the single-letter expression $\chi\left(\mcN\right)$  reduces to the \emph{additivity} of the Holevo information, i.e., whether
$\chi\!\left(\mcN_1\otimes \mcN_2\right) = \chi\left(\mcN_1\right)+\chi\left(\mcN_2\right)$
holds for all channels $\mcN_1,\mcN_2$ 
% \textcolor{blue}{add two references: 1) https://doi.org/10.1007/s00220-003-0981-7 2) Section 8.3 in Holevo, A. S. (2019). Quantum systems, channels, information: a mathematical introduction. Walter de Gruyter GmbH \& Co KG.}. 
(see \cite{shor2004equivalence} and \cite[Sec.~8.3]{holevo2019quantum}).
Hastings \cite{Hastings2009} showed that additivity does not hold in general, 
and as a consequence, the classical capacity of a general quantum channel remains an open problem. Nonetheless, additivity has 
been established for specific classes of channels, including unital qubit channels \cite{King2002} %, depolarizing channels \cite{King2003}, 
and entanglement-breaking channels \cite{Shor2002EB}.
Further discussion on the single letter capacity formula can be found in \cite{Pereg2023}.
% \textcolor{red}{add ref to Sec. VI-B, DOI: https://doi.org/10.1103/PhysRevA.108.042616}.
Nonetheless, the HSW Theorem \cite{Holevo1998,SchumacherWestmoreland1997} provides a useful and computable achievable rate, i.e., the lower bound
\begin{align}
  C(\mcN) \geq  \chi\!\left(\mcN\right) .
  \label{eq:lower_capacity_0}
\end{align}
\end{remark}

\subsection{With Environment Dependence}
\label{Subsection:Environment_Dependence}
Next, we consider a quantum channel that is affected by an  environment system, $S$.
Let $\mcN_{SA \to B}\circ\sigma_{SS_0}$ be a quantum channel with environment dependence, but no action dependence, as considered in 
% \cite{Anshu2020} \textcolor{blue}{add all works that consider channels with CSI}. 
~\cite{Dupuis2009, Dupuis2010t, anshu2018building, pereg2019entanglement, Pereg2022, pereg2022classical, pereg2021quantum, Anshu2020, zivarifard2024covert}.
% \textcolor{blue}{add reference}.
The environment system $S$ is inaccessible, yet the encoder may have access to a CSI system $S_0$ that shares entanglement with the environment.

The setting where the encoder cannot influence the channel environment can be viewed as a special case of our model, 
under the constraint that  $\abs{\mathcal{H}_G}=1$.
 Alice is then effectively neutralized of her influence over the channel environment, as the action system $G$ is degenerate.
The channel $\mcN_{SA \to B}$ has two inputs, Alice's transmission $A$ and the channel environment $S$.
The CSI $S_{0}$ and the channel environment $S$ are in a fixed state $\ket{\sigma_{SS_0}}=T|0\rangle$, which is unaffected by Alice's encoding operation.
We denote the capacity in this case by $C_{\text{n-c}}(\mathcal{N}\circ\sigma)$, as in Table~\ref{Table:Capacity}.

Define the 
single-letter expression
\begin{align}
    \mathsf{R}_{\text{n-c}}(\mathcal{N} \circ \sigma)
    \coloneqq
    \max_{\substack{p_V,\, \rho_{AS}^v \\ \mathrm{Tr}_{VA}\,[\rho_{VAS}] = \sigma_{S}}}
    \bigl[ I(V; B)_{\rho} - I(V; S)_{\rho} \bigr],
    \label{eq:R_0}
\end{align}
where $V \sim p_V$ is a classical auxiliary variable, and the mutual information quantities are 
evaluated on the state
\begin{align}
    \rho_{VAS} = \sum_{v} p_V(v)\,\ketbra{v}{v}_V \otimes \rho_{AS}^v.
\end{align}
such that $\Tr_{S_0}[\sigma_{SS_0}] = \Tr_{VA}[\rho_{VSA}]$. 
Intuitively, Alice may send Bob a compressed representation of her CSI. The term $I(V;S)_\rho$ then captures the communication cost of doing so.

\begin{theorem}[Without action dependence, non-causal CSI %Achievable Rate with Environment Dependence 
{\cite[Corollary~1]{Anshu2020}}]
\label{Theorem:Anshu}
    The capacity of the quantum environment-dependent  channel $\mathcal{N}_{SA \to B}$ with non-causal 
    CSI and \emph{without} action dependence satisfies
    \begin{align}
        C_{\text{n-c}}(\mathcal{N} \circ \sigma)
        =
        \lim_{k \to \infty} \frac{1}{k}\,
        \mathsf{R}_{\text{n-c}}\!\left(\mathcal{N}^{\otimes k} \circ \sigma^{\otimes k}\right).
    \end{align}
\end{theorem}

\begin{remark}
% \textcolor{blue}{Show that the classical Gel'fand-Pinsker channel is a special case.}
The classical Gel'fand-Pinsker channel \cite{pinsker1980coding} is a special case of the environment-dependent model.
In the classical model, a classical channel $P_{Y|X,S_0}$ is governed by an i.i.d. sequence of random parameters, according to
$S_0\sim q(s)$. 
This can be described by a classical state,
\begin{align}
\sigma_{S_0}=\sum_s q(s) \ketbra{s}_{S_0}
\end{align}
where $S_0$ is a classical CSI register that stores the channel parameter value. 
In the quantum case, there exists a purification, 
\begin{align}
\ket{\sigma_{SS_0}}=\sum_s \sqrt{q(s)} \ket{s}_{S}\otimes \ket{s}_{S_0}
\end{align}
where $S$ can be viewed as the \emph{inaccessible environment} that complements the information in the CSI system $S_0$.
Note that this state is typically entangled.
Since Alice has $S_0$ and the channel acts on $S$, one may say that Alice shares entanglement with the channel \cite{Dupuis2009}.  %\textcolor{blue}{add ref to Dupuis}.
\end{remark}

% \begin{remark}
% \textcolor{blue}{Expalin what happens when the channel is not affected by the envrionment (CSI is decoupled from the channel):
% $\ket{\sigma_{SS_0}}=\ket{\sigma_{S}}\otimes \ket{\sigma_{S_0}}$. Show that (18) reduces to the Holevo information.} 
% \end{remark}

\begin{remark}
When the channel is not affected by the environment, say $\mathcal{N}_{AS\to B}=\mathcal{N}'_{A\to B}\circ \trace_S$, the model reduces to that of the standard quantum channel $\mathcal{N}'_{A\to B}$ without environment dependence. 
% there are two 
% natural cases to consider.
%
% First, the environment system $S$ may be fixed in a deterministic state 
% $\ket{0}_S$. In 
% this case, the Stinespring isometry $T_{G \to SS_0}$ always produces 
% $\ket{\sigma_{SS_0}} = \ket{0}_S \otimes \ket{\sigma_{S_0}}$ and the channel $\mcN_{SA \to B}$ effectively reduces to
% \begin{align}
%     \hat \mcN_{A \to B} = \mcN_{SA \to B}\bigl(\ketbra{0}_S \otimes \rho_A\,\bigr),
% \end{align}
% acting solely on Alice's input $A$.
%
% Second, the channel may be completely insensitive to the environment, 
% in the sense that
% \begin{align}
%     \tilde{\mcN}_{A \to B} = \mathrm{Tr}_S \circ\, \mcN_{SA \to B},
% \end{align}
% meaning the output $B$ does not depend on $S$ at all, regardless of 
% its state.
%
 In this case, the side information $S_0$ is decoupled from the channel 
transformation, and Alice's access to $S_0$ does not improve performance. The 
penalty term $I(V;S)_\rho$ in \eqref{eq:R_0} vanishes, and the rate 
expression $\mathsf{R}_{\text{n-c}}(\mcN \circ \sigma)$ reduces to the 
Holevo information (cf. \eqref{eq:holevo_info} and \eqref{eq:R_0}). % confirming that non-causal CSI yields no coding advantage when the channel 
% environment $S$ does not affect the channel output.
\end{remark}

The example below demonstrates the power of side information.
% \begin{example}
% \textcolor{blue}{Add an example that shows the power of CSI: Completely depolarizng channel where $s$ controls the Pauli error.} 
% \end{example}
\begin{example}
Consider an environment-dependent qubit depolarizing channel $\mcN_{SA \to B}$ that depends on a classical environment parameter,
\begin{align}
    \mcN_{SA \to B}(\ketbra{s}{s}_S\otimes\rho_A ) 
    = \mathsf{P}_s \, \rho_A \, \mathsf{P}_s,
\end{align}
for $s\in \{0,1,2,3\}$,
where $\mathsf{P}_s$ are the 
Pauli operators: $\mathsf{P}_0 = \mathsf{I}$, $\mathsf{P}_1 = \mathsf{X}$, 
$\mathsf{P}_2 = \mathsf{Y}$, and $\mathsf{P}_3 = \mathsf{Z}$. 
In other words, the noisy channel applies a Pauli error that is controlled by the environment  parameter $s\in \{0,1,2,3\}$.
Suppose the action channel produces the environment 
state
\begin{align*}
\ket{\sigma_{SS_0}}=\frac{1}{2} \sum_{s=0}^3 \ket{s}_S\otimes \ket{s}_{S_0} .
\end{align*}
Note that this yields a symmetric environment state,
\begin{align}
    \sigma_S = \frac{1}{4} \sum_{s=0}^3 \ketbra{s}{s} %=  (1-\varepsilon)\ketbra{0}{0}_S 
    .% + \frac{\varepsilon}{3}\sum_{s=1}^{3} \ketbra{s}{s}_S,
\end{align}
% where $\varepsilon \in [0,1]$ is a given constant. That is, the 
% environment parameter $S$ selects which Pauli operator is applied 
% to the input system, with $q(0) = 1 - \varepsilon$ and 
% $q(1) = q(2) = q(3) = \varepsilon/3$. 
Without CSI, the effective 
channel averages over the environment parameter values, simulating a completely-depolarizing channel, such that the output state is $\frac{\mbone}{2}$, regardless of the channel input. Therefore the capacity is zero, i.e., $C(\mathcal{N})=0$.
% \begin{align}
%     \bar{\mcN}_{A \to B}(\rho_A) 
%     &= \sum_s q(s)\, \mathsf{P}_s \, \rho_A \, \mathsf{P}_s \notag \\
%     &= (1-\varepsilon)\,\rho_A 
%     + \frac{\varepsilon}{3}(\mathsf{X}\rho_A \mathsf{X} 
%     + \mathsf{Y}\rho_A \mathsf{Y} 
%     + \mathsf{Z}\rho_A \mathsf{Z}) \notag \\
%     &= (1-p)\,\rho_A + p\,\frac{\mathsf{I}}{2},
% \end{align}
% where $p = \frac{4\varepsilon}{3}$ is the depolarization parameter, 
% i.e., $\bar{\mcN}_{A \to B}$ is the standard depolarizing channel 
% with parameter $p$.

Given CSI, Alice can perform a measurement on $S_0$ to obtain the environment parameter value $s$. % that  trough the entangled state:
% \begin{align}
%     \ket{\sigma_{SS_0}} = \sum_{s=1}^4 \sqrt{q(s)} \ket{s}_S \otimes \ket{s}_{S_0}
% \end{align}
As a result, she knows which Pauli error will occur before she transmits, and can therefore 
``pre-correct" by applying the same Pauli operator to her input: 
$\mcF_{S_0 \to A}(\rho) 
= \mathsf{P}_s \, \rho \, \mathsf{P}_s$. 
Since the Pauli operators are Hermitian and unitary, % satisfy $\mathsf{P}_s^2 = \mathsf{I}$, 
the channel output is
\begin{align}
    \mcN_{SA \to B}\bigl(\ketbra{s}{s}_S \otimes \mcF_{S_0 \to A}(\rho) \bigr) 
    = \mathsf{P}_s (\mathsf{P}_s \, \rho \, \mathsf{P}_s) 
    \mathsf{P}_s = \rho,
\end{align}
and the effective channel is noiseless, hence $C_{\text{n-c}}(\mathcal{N} \circ \sigma)=1$.
%
% The rate $R = 1$ qubit per 
% channel use is therefore achievable, regardless of the value 
% of~$\varepsilon$. 
This demonstrates the power of  CSI: 
while without side information the capacity is zero, %that of the 
% depolarizing channel with parameter $p = \frac{2\varepsilon}{3}$, 
full knowledge of the environment allows Alice to completely 
cancel out the Pauli error and communicate at the maximal rate.
\end{example}
%----------------------------------------------
% Main Result
%----------------------------------------------
\section{Main Results}
\label{section:Main_Result}
We now state our main results, the achievable rates for the quantum action-dependent channel $\mcN_{SA \to B}\circ T_{G\to SS_0}$ for causal and non-causal CSI.
We consider the model described in Section~\ref{section:ADCoding}, where Alice prepares a quantum action $G$ that affects the channel environment.
Her action is transmitted through
an action map $T_{G \to SS_0}$, producing
Alice's CSI $S_0$ and the inaccessible channel environment $S$.
% \textcolor{red}{In addition, reminder the reader what action dependence means.} 
Given access to CSI, Alice encodes her message and prepares the transmission $A$ through the communication channel, $\mathcal{N}_{SA\to B}$. 
%
% which extends the result 
% in \cite{weissman2010} to the quantum setting.
We define the achievable non-causal rate as
\begin{subequations}
\begin{align}
    \rqad_{\text{n-c}} \left( \mcN\circ T \right) = \max_{\substack{p_{UV} \,,\; \ket{\sigma_G^u} \,,\; \rho_{VSA}^{u}  \\ \Tr_{A} [\rho_{VUAS}] = \Tr_{S_0} [\sigma_{VUSS_0}]}} \left[\, I(VU;B)_{\rho} - I(V;S|U)_{\rho}\,\right] \label{eq: R-nonc}
\end{align} 
where $\mcN \circ T$ is a short notation for the action-dependent channel $\mcN_{SA \to B}\circ T_{G\to SS_0}$, and the subscript  `n-c' indicates non-causal CSI at the encoder.
The optimization is over a classical auxiliary pair $(V,U)\sim p_{VU}$, a pure state collection $\{\ket{\sigma_G^u}\}$, 
% and an encoding channel $\mathcal{F}_{S_0\to A}^v$, such that
and a classical-quantum $\rho_{VUSA}$, for %$\Tr_{VA}[\rho_{VAS}^{u}] = \Tr_{S_0} [\sigma_{SS_0}^u]}$
\begin{align}
    \ket{\sigma_{SS_0}^u} &= T_{G\to SS_0}\ket{\sigma_G^u},
    \label{Equation:phi_SS0}
    \\
    \sigma_{VUSS_0}&=\sum_{v} \sum_u p_{VU}(v,u) \ketbra{v}_V\otimes\ketbra{u}_U\otimes \ketbra{\sigma_{SS_0}^u},
    \\
    % \rho_{SA}^{v,u} &= \mathrm{id}_{S}\otimes \mathcal{F}_{S_0\to A}^v(\ketbra{\sigma_{SS_0}^u}),
    % \\
    % \rho_{VSA}^{u} &= \sum_{v} p_{V|U}(v|u)\ketbra{v}{v}_V\otimes \rho_{SA}^{v,u}
    % \\
    % \rho_{VUSA} &= \sum_{v,u} p_{V,U}(v,u)\ketbra{v}{v}_V\otimes \ketbra{u}{u}_U \otimes \rho_{SA}^{v,u}
    % \\
    \intertext{such that}
    \Tr_{A}[\rho_{VUAS}] &= \Tr_{S_0} [\sigma_{VUSS_0}] \label{eq:marginal_constraint}
    % \\
    %  \rho_{VUSA} &= \sum_{u,v}  p_{V,U}(v,u)  \ketbra{v}{v}_V \otimes \ketbra{u}{u}_U \otimes \rho_{SA}^{v,u}, \label{eq:rho_VUSA} 
\end{align}
\end{subequations}
hence  %$\rho_V^u=\sum_{v\in\mathcal{V}} p_{V|U}(v|u)\ketbra{v}$, and 
$\rho_{VUB}=  \mathrm{id}_{VU}\otimes\mathcal{N}_{SA\to B}(\rho_{VUSA})$.
The alphabet sizes of the auxiliary variables can be restricted to $|\mathcal{U}| \leq |\mathcal{H}_S|^2 |\mathcal{H}_{S_0}|^2$ and $|\mathcal{V}| \leq |\mathcal{H}_{S_0}|^2 |\mathcal{H}_S|^4 |\mathcal{H}_A|^2$.

The formula above can be interpreted as follows.
The variable $U$ is an index over a collection of actions that Alice may choose.
Intuitively,
Alice may send Bob a compressed description $V$ of her CSI. The term $I(V;S|U)_\rho$ then captures the communication cost for a given action (see Subsection~\ref{Subsection:Environment_Dependence}).
For each realization $V=v$,
% Alice may use a different encoding channel $\mathcal{F}^v_{S_0 \to A}$, which represents Alice's
% transmission strategy. In other words, $\mathcal{F}^v_{S_0\to A}$ represents the encoding of both the message and the compressed description of her CSI.
Alice may prepare a different joint input state $\rho_{SA}^{v,u}$, which represents her
transmission strategy. In other words, the choice of $\rho_{SA}^{v,u}$ encodes both the message and a compressed description of her CSI.
 %
% To be precise,
% the auxiliary variable $V$ determines the encoding  
% operation applied to the CSI system $S_0$.
%  In the quantum setting, this encoding is modeled
%  by a set of CPTP maps $\{\mathcal{F}^v_{S_0 \to A}\}_v$.
Since $S_0$ is entangled with the channel environment $S$, the selection of $V$ through Alice's interaction
with $S_0$ induces a correlation between~$V$ and~$S$.
% The achievable rate $I(VU;B)_\rho - I(V;S|U)_\rho$
% reflects the trade-off between message encoding and reliable compression: the first term captures the information
% that Bob can extract from the channel output,
% while the penalty $I(V;S|U)_\rho$ represents the rate cost Alice pays 
% for ``sending corrrelation" with the channel environment $S$. %into her transmission strategy $V$,  given the action $U$. 

% Intuitively, \textcolor{blue}{Explain the interpretation and intuition behind the random variables and the encoding channel. 
% Then, explain the intuition behind the formula $I(VU;B)_{\rho} - I(V;S|U)_{\rho}$. Sth like ``Alice sends a compressed representation of the CSI to Bob, at the cost of $I(V;S|U)$...}

\begin{theorem} [Action dependence and non-causal CSI] \label{Theorem:non-causal}
Consider a quantum action-dependent channel $\mathcal{N}_{SA\to B}\circ T_{G\to SS_0}$. 
The rate $R=\rqad_{\text{n-c}} \left( \mcN\circ T \right)$ is achievable with quantum action dependence and  non-causal CSI, i.e., 
    \begin{align}
        \capc_{\text{n-c}} \left( \mcN \circ T \right) \geq  \rqad_{\text{n-c}} \left( \mcN\circ T \right).
    \end{align}
    Furthermore, the capacity of the quantum action-dependent channel $\mathcal{N}_{SA\to B}$ with non-causal CSI satisfies
     \begin{align}
        \capc_{\text{n-c}} \left( \mcN \circ T \right) = \lim_{k\to\infty} \frac{1}{k} \rqad_{\text{n-c}} \left( \mcN^{\otimes k}\circ T^{\otimes k} \right).
    \end{align}
\end{theorem}
The proof of Theorem~\ref{Theorem:non-causal} is given in Section \ref{section:proof of main theorem}, based on
quantum one-shot information-theoretic methods.

\begin{remark}
As opposed to the classical analysis~\cite{weissman2010}, our techniques establish non-asymptotic performance bounds by directly analyzing the error probability for a finite number of channel uses, rather than relying on asymptotic arguments. The introduction of action dependence makes our analysis more challenging than in previous environment-dependent channel models \cite{Anshu2020, Pereg2022}. In previous settings, the channel environment and side information are set by an external source. Here, however, the transmitter's action \textit{induces} the shared entangled state. Consequently, our analysis must account for this additional layer of control. The formula of our capacity bound thus includes optimization over not only the input state, but also a quantum ensemble for the sender's action.
\end{remark}

  \begin{remark}
    Previous work has considered side information %$S_0$, 
    when the quantum state $\sigma_{S S_0}$ is fixed and dictated by the model \cite{Anshu2020,zivarifard2024covert, Pereg2022} (see Subsection~\ref{Subsection:Environment_Dependence}).
    This fixed state represents the entanglement between
    % In this scenario, the entangled state $\sigma_{S S_0}$ spans 
    the environment       
    subsystem $S$ and the side-information subsystem $S_0$, which is accessible to Alice. By utilizing the 
    side information $S_0$, Alice's encoder can generate entanglement between the channel input and its environment $S$. In this sense, 
    we can think of Alice as being entangled with the channel, and this is a key feature of quantum side information at the transmitter.
    In our model, however, 
    %The difference compared to the previous models is that here, 
    the state $\sigma^u_{S S_0}$ depends on the action encoding $u$ chosen by Alice.
    This adds a degree of freedom that allows Alice to influence
    the channel environment $S$ by selecting different actions.
    Our model is analogous to the classical action-dependent channel model in \cite{weissman2010}, where the channel parameter is a noisy version of Alice's action.
  \end{remark}

% \begin{remark}
% \textcolor{blue}{Add explanation on our rate formula vs. Anshu et al. formula. That is, how our formula reduces to theirs}.
% \end{remark}

\begin{remark}
As a special case, we recover a previous result by Anshu et al. \cite{Anshu2020}.
Specifically,
if the action system $G$ is degenerate, i.e., $|\mcH_G|=1$,  Alice cannot influence the channel environment through the 
choice of action, and the action isometry $T_{G \to SS_0}$ produces a fixed state 
$\ket{\sigma_{SS_0}} = T_{G \to SS_0}\ket{0}_G$ that is independent of Alice's encoding strategy. 
The channel environment $S$ and the side-information subsystem $S_0$ are then determined entirely by this fixed state, 
% and Alice's only remaining degree of freedom is the encoding channel $\mcF^v_{S_0 \to A}$, through which she can still
% exploit her access to $S_0$ in order to adapt the channel input $A$ to the environment.
and Alice's only remaining degree of freedom is the choice of joint input state $\rho_{SA}^v$, through which she can still
exploit the correlation between $S_0$ and the channel environment $S$ to adapt the channel input $A$. The model thus reduces to the 
environment-dependent setting \cite{Anshu2020}, where the environment state  $\sigma_{SS_0}$ is 
fixed, and Alice has access to %encodes with knowledge of 
the correlated side information $S_0$. In this case,
our lower bound reduces to the formula from \cite{Anshu2020} without action dependence (see Subsection~\ref{Subsection:Environment_Dependence}), as
the action variable $U$ 
becomes trivial, the auxiliary $V$ assumes the same role as in \eqref{eq:R_0}, and the achievable rate in \eqref{eq: R-nonc} reduces to 
$\mathsf{R}_{\text{n-c}}(\mcN \circ \sigma)$ as defined 
in \eqref{eq:R_0}. Conversely, our action-dependent model 
generalizes the environment-dependent setting of Anshu et al. \cite{Anshu2020}, as we grant Alice 
 the capability to shape the environment state. 
% joint state $\ket{\sigma^u_{SS_0}} = T_{G \to SS_0}\ket{\sigma^u_G}$ 
% through her choice of action encoding, thereby influencing 
% both the channel environment $S$ that governs the channel 
% $\mcN_{SA \to B}$ and the side information $S_0$ available 
% to her encoder, at the cost of optimizing over the action 
% distribution $p_U$ and the pure state collection 
% $\{\ket{\sigma^u_G}\}$ jointly with the encoding 
% channel $\mcF^v_{S_0 \to A}$.
\end{remark}

We now establish an achievable rate for the causal setting, where Alice only has access to past and present CSI. That is, the transmission at time $i$ may only depend on the side information sequence $S_0^i = {S_0}_1,\ldots , {S_0}_i$	
observed up to that point.
Define
% We define the achievable causal rate as
\begin{subequations}
\label{Equation:Causal_Lower_Bound}
\begin{align}
    \rqad_{\text{caus}} \left( \mcN \circ T \right) = \max_{p_U, \ket{\sigma_G^u}, \mcF_{S_0 \to A}^u} I(U; B)_{\rho}
\label{Equation:Action_Causal_Rate}
\end{align}
where `caus' stands for causal CSI.
The optimization is over a classical auxiliary variable $U \sim p_U$, a collection of action states $\{\ket{\sigma_G^u}\}$, and a collection of maps $\mathcal{F}_{S_0 \rightarrow A}^u$, such that 
    \begin{align}
    % \sigma_{SS_0}^u &= \mathcal{T}_{G\to SS_0}(\sigma_G^u),
    \ket{\sigma_{SS_0}^u} &= T_{G\to SS_0}\ket{\sigma_G^u},
    \label{Equation:phi_SS0_causal}
    \\
    \rho_{SA}^{u} &= \mathrm{id}_{S}\otimes \mathcal{F}_{S_0\to A}^u(\ketbra{\sigma_{SS_0}^u}),
    \\
     \rho_{USA} &= \sum_{u} p_{U}(u)   \ketbra{u}{u}_U \otimes \rho_{SA}^{u} \label{eq:rho_VUSA_Causal} 
    \end{align}
    \end{subequations}
hence  %$\rho_V^u=\sum_{v\in\mathcal{V}} p_{V|U}(v|u)\ketbra{v}$, and 
$\rho_{UB}=  \mathrm{id}_{U}\otimes\mathcal{N}_{SA\to B}(\rho_{USA})$.

\begin{theorem}[Action dependence and causal CSI]
\label{Theorem:causal}
Let $\mathcal{N}_{SA\to B}\circ T_{G\to SS_0}$ be a quantum action-dependent channel.
The rate $R=\rqad_{\text{caus}} \left( \mcN\circ T \right)$ is achievable with quantum action dependence and causal CSI at the transmitter.
% , i.e., .
Hence, the capacity satisfies
\begin{align}
    \capc_{\text{caus}} \left( \mcN \circ T \right) \geq \rqad_{\text{caus}} \left( \mcN \circ T \right).
\end{align}
Furthermore, the capacity of the quantum action-dependent channel $\mathcal{N}_{SA\to B}$ with causal CSI satisfies
     \begin{align}
        \capc_{\text{caus}} \left( \mcN \circ T \right) = \lim_{k\to\infty} \frac{1}{k} \rqad_{\text{caus}} \left( \mcN^{\otimes k}\circ T^{\otimes k} \right).
        \end{align}
\end{theorem}
The proof of Theorem~\ref{Theorem:causal} is given in Section \ref{section:proof of causal theorem}.
 Unlike the non-causal case, the rate is not penalized by $I(V;S|U)_{\rho}$, as the transmission strategy is chosen independently of the induced side information.

\begin{remark}\label{remark:virtual_channel}
As in Shannon's proof for causal CSI, the communication scheme can be interpreted as coding for a virtual channel  \cite{Shannon1958}, \cite[Sec.~3.1]{Keshet2008} (see also discussion in \cite{pereg2019entanglement}).
In our case, the virtual channel is a classical-quantum channel 
$\mathcal{M}:\mathcal{U}\to \mathscr{D}(\mathcal{H}_B)$, defined by 
\begin{align}
\mathcal{M}(u)=\mathcal{N}_{SA\to B}(\rho_{SA}^u)
\end{align}
for $u\in\mathcal{U}$.
Our lower bound for causal CSI is in fact the same as the Holevo information of this virtual channel, i.e.,
 $\rqad_{\text{caus}} \left( \mcN \circ T \right)=\chi(\mathcal{M})$ (cf. \eqref{eq:holevo_info} and \eqref{Equation:Action_Causal_Rate}).
\end{remark}

% As a consequence of Theorem~\ref{Theorem:causal}, we characterize the capacity of the environment-dependent channel without action dependence.
% That is, we consider a channel $\mathcal{N}_{SA\to B}$ governed by a fixed environment  state, $\ket{\sigma_{SS_0}}$.
% To obtain the corollary below, we set the action dimension as$\abs{\mathcal{H}_G}=1$.
% Consider

% % We define the achievable causal rate as
% \begin{align}
%     \rqad^{\text{caus}} \left( \mcN \circ \sigma \right) = \max_{p_U , \mcF_{S_0 \to A}^u} I(U; B)_{\omega}.
% \label{Equation:0_Action_Causal_Rate}
% \end{align}
% %
% The optimization is over a classical auxiliary variable $U \sim p_U$ and  maps $\mathcal{F}_{S_0 \rightarrow A}^u$, such that 
%     \begin{align}
%     \omega_{SA}^{u} &= \mathrm{id}_{S}\otimes \mathcal{F}_{S_0\to A}^u(\ketbra{\sigma_{SS_0}}),
%     \\
%      \omega_{USA} &= \sum_{u} p_{U}(u)   \ketbra{u}{u}_U \otimes \omega_{SA}^{u}. \label{eq:0_rho_VUSA_Causal} 
%     \end{align}

% \begin{corollary}[Without action dependence,  causal CSI]
% \label{Corollary:causal}
% The capacity of the quantum environment-dependent  channel $\mathcal{N}_{SA \to B}$ with causal 
%     CSI and \emph{without} action dependence satisfies
%     \begin{align}
%         C_{0}^{\text{caus}}(\mathcal{N} \circ \sigma)
%         =
%         \lim_{k \to \infty} \frac{1}{k}\,
%         \mathsf{R}_{0}^{\text{caus}}\!\left(\mathcal{N}^{\otimes k} \circ \sigma^{\otimes k}\right).
%     \end{align}

% \end{corollary}

\begin{remark}
Shannon \cite{Shannon1958} established that for a classical channel 
$P_{Y|X S_0}(y|x,s)$,
 with causal CSI and no action dependence, the capacity can be expressed as
\begin{align}
C_{\text{caus}}(P_{Y|XS_0})=\max_{p_F} I(F;Y)
\end{align}
such that  $X=F(S_0)$. 
In this formulation, $F:\mathcal{S}_0\to \mathcal{X}$ is an auxiliary function, known as a \emph{Shannon strategy} (see \cite[Remark 6]{Pereg2022}). The optimization is taken over all probability distributions on such strategies. In Shannon’s coding construction, the strategy assigns each parameter realization $S_{0,i}=s$ to an input symbol $X_i=F(s)$.
 Analogously, the collection $\{\mathcal{F}^u_{S_0\to A}\}_{u\in\mathcal{U}}$
 may be interpreted as a family of quantum Shannon strategies, each  maps from the CSI system $S_0$
 to the channel input system $A$.
\end{remark}

\section{%Example: 
Memory With Depolarization and  Selective Rewrite}
\label{section:examples}
We demonstrate our results  
in Theorems~\ref{Theorem:non-causal} and~\ref{Theorem:causal}
through different scenarios of a qubit memory with depolarization noise and selective rewrite. 
The model can be viewed as a quantum analog of the 
classical memory with selective rewrite \cite[Sec.~V]{weissman2010}.
%
% \subsection{Depolarization with Selective Rewrite}
% \label{subsection:example_causal}
%
We consider %the physical scenario is 
a quantum memory block where the writer (Alice) writes a qubit state into storage, where the memory itself may introduce a depolarization error into the stored qubit. Alice can sample the error index and if needed, rewrite the faulty entry before the reader (Bob) reads from memory.
Denote the qubit depolarizing channel as 
\begin{align}
\mathcal{D}_p(\rho)&=(1-p)\rho + {\tfrac{p}{3}}\Big[\mathsf{X}\rho\mathsf{X}  + \mathsf{Y}\rho\mathsf{Y}  + \mathsf{Z}\rho\mathsf{Z}\Big]
\\
&=\left(1-\frac{4p}{3}\right)\rho+\left(\frac{4p}{3}\right)\frac{\mbone}{2},
\end{align}
for $p \in \left[0,\frac{3}{4}\right]$ (see \cite[Sec. 4.7.4]{wilde2017}). We now define an action dependent channel $\mathcal{N}_{SA\to B}\circ T_{G\to SS_0}$.

Consider an action map such that the CSI register $S_0$ stores a flag that indicates which Pauli error has occurred, $\mbone $, $\mathsf X$, $\mathsf Y$, or $\mathsf Z$. Specifically,
\begin{align}
    T_{G\to SS_0} = \sqrt{1-p}\,\mbone \otimes \ket{0}_{S_0} + \sqrt{\tfrac{p}{3}}\Big[\mathsf{X}  \otimes \ket{1}_{S_0} + \mathsf{Y} \otimes \ket{2}_{S_0} + \mathsf{Z} \otimes \ket{3}_{S_0}\Bigr].
\end{align}
Note that this isometry is associated with the Stinespring dilation of the depolarizing channel, i.e., $\mathcal{D}_p(\rho)=\trace_{S_0}\left[ T\rho T^\dagger \right]$ for $T\equiv T_{G\to SS_0}$. 
 By measuring the CSI system $S_0$, Alice can learn whether her stored qubit was corrupted.
When actions are encoded using the computational basis, $\{\ket{0}_G,\ket{1}_G\}$,
both $\mathsf X$ and $\mathsf Y$ introduce a bit flip, while $\mbone$ and $\mathsf Z$ do not.
Hence, a measurement outcome of $s_0\in\{1,2\}$ indicates a bit flip, and $s_0\in\{0,3\}$ signals that a bit flip has not occurred.

Furthermore, consider a qutrit input space $\mathcal{H}_A$, with an orthonormal basis ${\{\ket\perp_A,\ket0_A,\ket1_A\}}$, where $\ket\perp$ is referred to as the ``idle" state.
Let $\Pi_{01}$ denote the projector onto the qubit subspace, i.e.,
$\Pi_{01}=\ketbra0_A+\ketbra1_A$.

The communication channel $\mcN_{SA \to B}$ acts as a selective memory:
\begin{itemize}
\item
If the input state is $\ket\perp$ (idle), then the environment qubit $S$ is forwarded to the output.

\item
If the input state is $\ket0$ or $\ket1$ (rewrite), the input qubit $A$ is transmitted with depolarization noise.
\end{itemize}
% This two-case behavior is illustrated in Figure~\ref{fig:qmemory}.
The channel is thus specified by 
\begin{align}
    \mcN_{SA \to B}(\rho_{SA}) =  \Tr_A \left[ \left( \id_S \otimes \ketbra{\perp}_A \right)  \rho_{SA} \right] +
      \mathcal{D}_p  \left(  \Pi_{01} \cdot  \Tr_S \rho_{SA}\cdot   \Pi_{01} \right),
\end{align}
where $\Pi_{01}$ projects onto the qubit space.

\begin{figure}[t]
    \centering
    \includegraphics[scale=0.13]{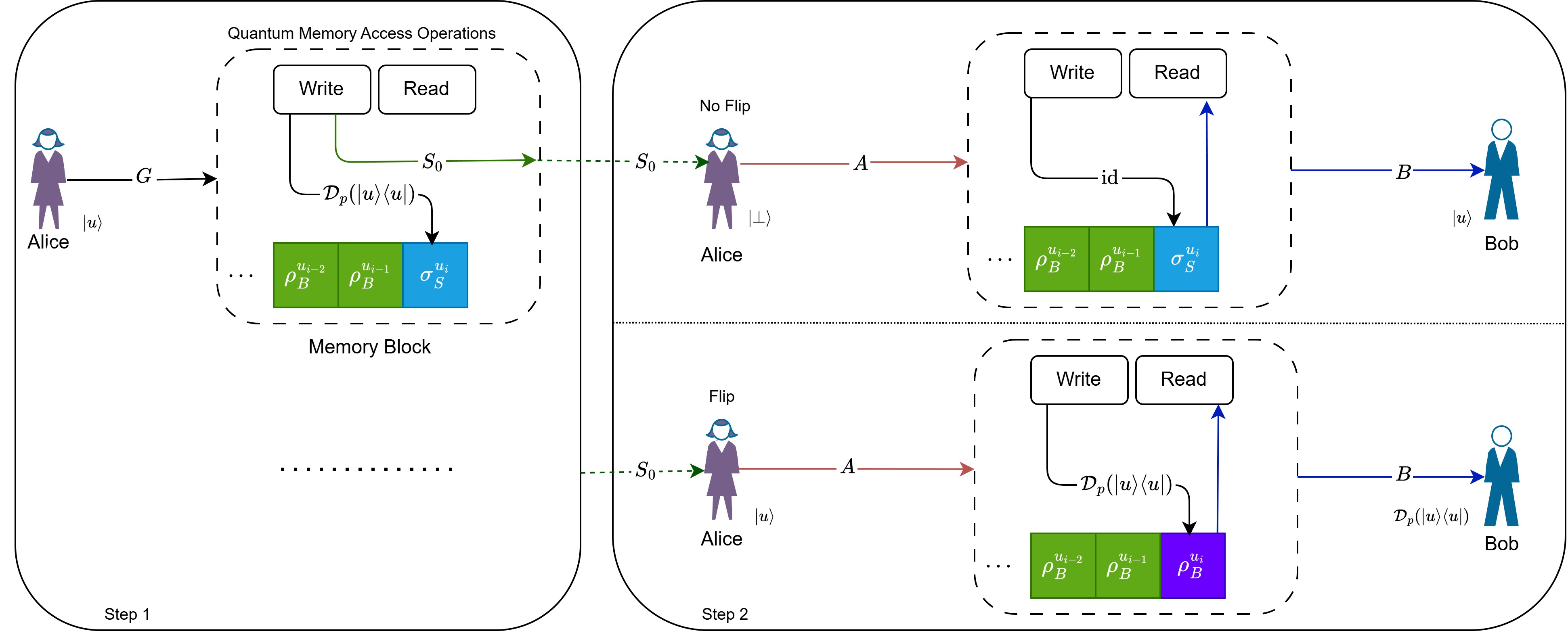}
    \caption{Illustration of the selective rewrite strategy for the depolarizing memory channel.
    \emph{Left:} Alice's first step: she encodes her classical bit $u$ as $\ket{u}$ and writes it into memory, where it is subject to depolarization noise $\mathcal{D}_p$ as described in the example.
    \emph{Top right:} No bit-flip detected ($s_0\in\{0,3\}$), hence Alice sends the idle state $\ket{\perp}$, forwarding the already-correct memory content to Bob.
    \emph{Bottom right:} Bit-flip detected ($s_0\in\{1,2\}$), hence Alice rewrites $\ket{u}$, correcting the error before Bob reads from memory.}
    \label{fig:qmemory}
\end{figure}

In the two-stage rewrite setting, the encoder first writes to the memory (first pass) and observes a noisy version of the resulting output, which reveals where bit flips have occurred. In the second pass, it can rewrite selected positions based on this knowledge.

%
% Specifically, it may choose to rewrite only a subset of the flipped bits, using the selection of which positions to rewrite as an additional means of conveying information. 
% This introduces a tradeoff: rewriting more bits improves reliability by reducing noise, while rewriting fewer preserves variability that can be used to encode extra information. 
% The gain arises because, with full knowledge of the first-pass outcome, the rewrite decisions themselves become part of the communication strategy, rather than serving only as error correction.
%
% In the non-causal setting, the encoder has full knowledge of the first-pass output sequence before choosing the rewrite strategy, which allows it to do more than simply correct all observed errors. 
% Instead of deterministically rewriting every flipped bit, the encoder can selectively rewrite only a subset of them, using the pattern of rewrites itself as an additional degree of freedom to convey information. 
% This creates a tradeoff: rewriting more positions reduces noise and improves reliability, but rewriting fewer preserves randomness that can be exploited to encode extra bits through the choice of which locations are rewritten. 
The optimal strategy balances these effects, leveraging the rewrite operation not merely as an error-correction mechanism but as part of the signaling scheme, a benefit that is specific to the non-causal case where such global, sequence-level coordination is possible.

We use the capacity notation in Table~\ref{Table:Capacity}. Detailed derivations for the scenarios below are provided in Appendix~\ref{appendix:example_derivations}.
\subsection{Without CSI}
First, consider the case where Alice has no side information.
In the absence of CSI, the model reduces to
depolarization without rewrite, as the writer (Alice) would not gain from using the idle state.
In other words, we effectively have
 a depolarizing channel $\mathcal{D}_p$ without environment dependence, $\tilde{\mathcal N}_{A\to B}=
\mathcal{D}_p$. 
Thus,
the capacity without CSI is given by \cite{King2003}  \cite[Sec. 20.4.4]{wilde2017}:
\begin{align}
    \capc(\tilde{\mathcal N})=\chi(\tilde{\mathcal{N}})= 1 - h_2(\delta) 
    \label{Equation:Rewrtie_Without_CSI}
\end{align}
where $h_2(x)=-(1-x)\log(1-x)-x\log x$ is the binary entropy function, 
and
$\delta = \frac{2p}{3}$.

\subsection{Causal CSI}
\label{subsection:example_causal}
We now consider the quantum action-dependent channel with causal
CSI at the transmitter, as in Theorem~\ref{Theorem:causal}.
Consider the rate formula in \eqref{Equation:Causal_Lower_Bound}.
In the selective rewrite case, 
one may think of $U\in\{0,1\}$ as the classical bit Alice wants to store.
We set the action ensemble as $\{ \sigma_G^u=\ketbra u_G\}$, 
% to be $\{\ket{0}_G, \ket{1}_G\}$, 
with $U \sim \text{Bernoulli}\!\left(\frac{1}{2}\right)$. 
%
% \begin{align}
%     \mathcal{D}_p(\rho) = (1-p)\rho + \frac{p}{3}\left(\mathsf{X}\rho\mathsf{X} + \mathsf{Y}\rho\mathsf{Y} + \mathsf{Z}\rho\mathsf{Z}\right).
% \end{align}
%
% The action channel is then
% \begin{align}
%     \mcT_{G \to SS_0}(\sigma_G^u) = \sum_{k \in \{0,1\}} \bra{k}\mathcal{D}_p(\sigma^u_G)\ket{k} \cdot \ketbra{k}{k}_{S} \otimes \ketbra{k}{k}_{S_0}. \label{causal_example:action_channel}
% \end{align}
 % Since the depolarizing channel acts on a computational-basis state $\ketbra{u}{u}$, and noting that $\mathsf{X}\ketbra{u}{u}\mathsf{X} = \ketbra{\bar{u}}{\bar{u}}$, $\mathsf{Y}\ketbra{u}{u}\mathsf{Y} = \ketbra{\bar{u}}{\bar{u}}$, and $\mathsf{Z}\ketbra{u}{u}\mathsf{Z} = \ketbra{u}{u}$, we have
% \begin{align}
%     \mcN_p(\ketbra{u}{u}) &= (1-p)\ketbra{u}{u} + \frac{p}{3}\left(\ketbra{\bar{u}}{\bar{u}} + \ketbra{\bar{u}}{\bar{u}} + \ketbra{u}{u}\right) \nonumber \\
%     &= (1 - \delta)\ketbra{u}{u} + \delta\ketbra{\bar{u}}{\bar{u}},
% \end{align}
% \begin{align}
%     \mcN_p(\sigma_G^u) &= (1 - \delta)\ketbra{u}{u} + \delta\ketbra{\bar{u}}{\bar{u}},
% \end{align}
% where $\bar{u} = 1 - u$ denotes the bit-flip of $u$ and $\delta \coloneqq \frac{2p}{3}$ is the effective flip probability.
The action isometry $T_{G \to SS_0}$ thus produces the bipartite state below (derivation in Appendix~\ref{appendix:example_derivations}):
\begin{align}
    \ket{\sigma_{SS_0}^u} &=T_{G \to SS_0}\ket{\sigma_G^u}
    \nonumber \\
%    &= \ket{u}_{S} \otimes \left(\sqrt{1-p}\ket{0}_{S_0} + (-1)^u \sqrt{\frac{p}{3}}\ket{3}_{S_0}\right) + \sqrt{\frac{p}{3}}\ket{\bar{u}}_S \otimes \left(\ket{1}_{S_0} + i(-1)^u\ket{2}_{S_0} \right)
%        \nonumber \\
    &= \sqrt{1-\frac{2p}{3}}\ket{u}_{S} \otimes \ket{\theta_{03}^u}_{S_0} + \sqrt{\frac{2p}{3}}\ket{\bar{u}}_S \otimes \ket{\theta_{12}^u}_{S_0}
    \label{eq:sigma_SS0_causal_example}
\end{align}
where $\bar{u}=1-u$ denotes the flipped bit, $\ket{\theta_{03}^u}=\sqrt{\frac{1}{1-2p/3}}\left(\sqrt{1-p}\ket{0} + (-1)^u\sqrt{\frac{p}{3}}\ket{3}\right)$
and $\ket{\theta_{12}^u}=\frac{1}{\sqrt2}\left(\ket{1} +i(-1)^u\ket{2} \right) $.
Note that the states $\ket{\theta_{03}^u}$ and $\ket{\theta_{12}^u}$ are orthogonal.
% Since $\ket{\theta_0}$ and $\ket{\theta_1}$ are orthogonal, the reduced CSI state
% \begin{align}
%     \sigma_{S_0}^u &= \Tr_{S} \left[\ketbra{\sigma_{SS_0}^u}\right]
%   = \left( 1-\frac{2p}{3}\right)\ketbra{\theta_0}_{S_0}+\frac{p}{3}\ketbra{\theta_1}_{S_0}
% \end{align}
% does not depend on $u\in\{0,1\}$.

 We specify a partial isometry $F^u_{S_0 \to A}$ for the transmission encoder: 
% Measure $S_0$ in the computational basis, 
% in the computational basis
% and map the outcome $s_0$ to a qutrit system from $ \{\ket{\perp}_A, 0, 1\}$ according to the following rule. 
\begin{align}
F^u_{S_0 \to A}=\ketbra{\perp}{\theta_{03}^u}+\ketbra{u}{\theta_{12}^u}.
\end{align}
In particular,
if the CSI state is $\ket{\theta_{03}^u}_{S_0}$
% $s_0 \in \{0,3\}$ 
(the CSI matches the action), 
transmit the idle state $\ket{\perp}_A$. 
If the CSI state is $\ket{\theta_{12}^u}_{S_0}$
% $s_0 \in \{1,2\}$ 
(a bit-flip has occurred), 
rewrite the action state $\ket{u}_A$.
%
% The map can thus be expressed by
% \begin{align}
%     \mcF^u_{S_0 \to A}(\rho) &= \Tr\left[\Pi_{03} \rho \Pi_{03} \right] \ketbra{\perp}_A  + \Tr\left[\Pi_{12} \sigma^u_{S_0} \Pi_{12} \right] \ketbra{u}_A 
%     \intertext{where}
%     \Pi_{ij}&=\ketbra{i}+\ketbra{j}
% \end{align}
% for $i,j\in\{0,1,2,3\}$.
Figure~\ref{fig:qmemory} illustrates this encoding strategy.
% of the selective rewrite scheme.
% On the left, Alice writes her intended qubit directly into memory (first write).
% On the right, Alice's action depends on the CSI: if no bit-flip occurred, she sends the idle state $\ket{\perp}$ (top right), and if a bit-flip is detected, she rewrites $\ket{u}$ to correct the error (bottom right).
% The resulting channel input is thus
% \begin{align}
%     \ket{\omega_{SA}^u} &= \mbone_S \otimes F^u_{S_0 \to A}\ket{\sigma_{SS_0}^{u}}
%     \\
%     &= \sqrt{1-\frac{2p}{3}}\ket{u}_{S} \otimes \ket{\perp}_{A} + \sqrt{\frac{2p}{3}}\ket{\bar{u}}_S \otimes \ket{u}_{A}
% \end{align}
% and the reduced state is
% $\omega_A^u = \Tr_S(\omega_{SA}^u)=\left({1-\frac{2p}{3}}\right)\ketbra{\perp}+\frac{2p}{3}\ketbra{u}$.
%
% The resulting output state is
% \begin{align}
%     \mcN_{SA \to B}(\ketbra{\omega_{SA}^u}) &= (\mbone\otimes \bra{\perp}) \ketbra{\omega_{SA}^u} (\mbone\otimes \ket{\perp})   +
%       \mathcal{D}_p  \left(  \Pi_{01} \cdot   \omega_{A}^u \cdot   \Pi_{01} \right)
%       \\
%        &=\left({1-\frac{2p}{3}}\right)\ketbra{u}+ \frac{2p}{3}\mathcal{D}_p  \left(  \ketbra{{u}} \right)
%        \\
%        &=\left({1-\frac{2p}{3}}\right)\ketbra{u}+ \frac{2p}{3}  \left[  \left({1-\frac{2p}{3}}
%        \right)\ketbra{{u}}
%        +\frac{2p}{3}\ketbra{\bar{u}} \right]
%        \\
%        &=({1-\delta^2})\ketbra{{u}}+ \delta^2\ketbra{\bar{u}} ,
% \end{align}
% where $\delta = \frac{2p}{3}$.

As derived in Appendix~\ref{appendix:example_derivations}, the output state is $(1-\delta^2)\ketbra{u} + \delta^2\ketbra{\bar{u}}$ where $\delta = \frac{2p}{3}$.
This yields the following achievable rate for the selective-rewrite channel $\mathcal{N}_{SA\to B}\circ T_{G\to SS_0}$ with action dependence and causal 
CSI, 
%
% We compute the achievable rate $I(U;B)_\omega$.
% Bob's output state is the classical binary state
% $\omega_B^u = (1-\delta^2)\ketbra{u}{u} + \delta^2\ketbra{\bar{u}}{\bar{u}}$,
% with $H(B|U)_\omega = h_2(\delta^2)$ and $H(B)_\omega = 1$, yielding
\begin{align}
    % \capc^{\text{n-c}}(\mcN \circ T)\geq
    \capc_{\text{caus}}(\mcN \circ T)&\geq\rqad_{\text{caus}}(\mcN \circ T)
    \nonumber\\
    &\geq I(U;B)_\rho
    \nonumber\\
    &= 1 - h_2(\delta^2) 
    \nonumber\\
    &= 1 - h_2\!\left(\frac{4p^2}{9}\right).
    \label{Equation:Rewrtie_Causal_CSI}
\end{align}
The effective crossover probability $\delta^2$ reflects two consecutive depolarization errors, first a bit-flip in the action channel, followed by an error upon rewrite. As can be seen in Figure~\ref{fig:noiseless_comparison}, the achievable rate with CSI is significantly larger than 
the bound without CSI (cf. \eqref{Equation:Rewrtie_Without_CSI} and \eqref{Equation:Rewrtie_Causal_CSI}). 
Notice that even when the channel is completely depolarizing, for $p=\frac34$, Alice can use the CSI and the rewrite strategy to send information, hence 
$\capc_{\text{caus}}(\mcN \circ T)\geq 1 - h_2\!\left(\frac{1}{4}\right)>0$ with CSI, while $\capc(\tilde \mcN)=0$ without CSI. 
%, and is significantly smaller than $\delta$ for the standard depolarizing channel without CSI.

\begin{remark}[Without action dependence]\label{remark:0C_comparison}
We note that the action dependence is critical. 
%
% In the environment-dependent channel setting, without action dependence, the environment is fixed by nature. 
For example, suppose % we take the environment to be uniformly distributed:
$\ket{\sigma_{SS_0}} = \frac{1}{\sqrt2}(\ket{0}\otimes \ket{0}+\ket{1}\otimes \ket{1})$.
Then, using the idle state $\ket\perp$ only inserts stronger depolarization, and then the strategy above does not improve the performance, 
as in the case without CSI. See the blue line in Figure~\ref{fig:noiseless_comparison}.
%
% % Here, unlike the action dependent setting where $S_0$ carries the full Pauli index, here the side information $S_0 \in \{0,1\}$ is a 
% % classical copy of the environment bit $S$, so Alice can perfectly identify 
% % if the environment matches her intended message or not.
% Alice applies the same selective rewrite strategy $\mathcal{F}_{S_0\to A}^u$. Since $U$ and $S_0$ are uncorrelated,  
% % the environment matches with probability $\frac{1}{4}$, and Alice sends idle;
%  the output is sent through $\mcD_p$ with probability $\frac{3}{4}$. Bob's output state is
% \begin{align}
%     \omega_B^u = \frac{1}{4}\ketbra{u}{u} + \frac{3}{4}\!\left((1-\delta)\ketbra{u}{u}+\delta\ketbra{\bar{u}}{\bar{u}}\right) = \left(1 - \tfrac{3\delta}{4}\right)\ketbra{u}{u} + \tfrac{3\delta}{4}\ketbra{\bar{u}}{\bar{u}},
% \end{align}
% with achievable rate $\msR_{0}^{\text{caus}}(\mcN \circ \sigma) = 1 - h_2\!\left(\frac{3\delta}{4}\right) = 1 - h_2\!\left(\frac{p}{2}\right)$.
% We conclude that
% % Since $\delta^2 < \frac{3\delta}{4} < \delta$ for $\delta \in \left(0,\frac{1}{3}\right)$, we have
% \begin{align}
%     \rqad^{\text{caus}}(\mcN \circ T)  \;>\; \msR_{0}^{\text{caus}}(\mcN \circ \sigma)  \;>\; \chi(\mathcal{D}_p),
% \end{align}
%  as $\chi(\mathcal{D}_p) = 1 - h_2\!\left(\frac{2p}{3}\right)$ is the capacity of the standard depolarizing channel without CSI. 
% This example demonstrates that action dependence provides a strict advantage over CSI alone, which in turn outperforms the standard depolarizing channel without CSI.
\end{remark}

\begin{figure}[t]
    \centering
    \includegraphics[scale=0.5]{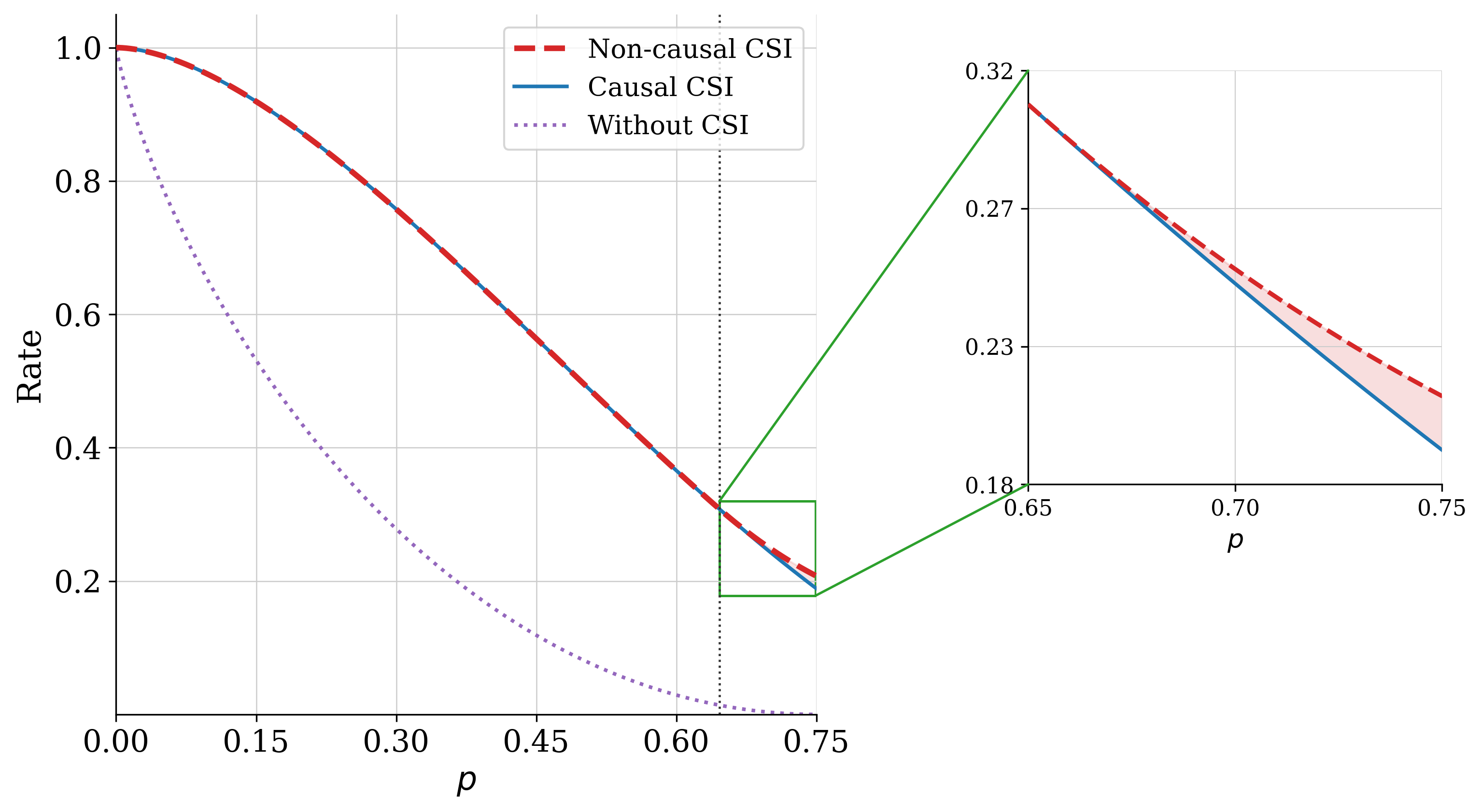}
    \caption{Achievable rates for the selective rewrite channel with action dependence,
    as a function of the depolarization parameter $p \in [0,\frac{3}{4}]$.
    The dashed red line shows the non-causal achievable rate $\rqad_{\text{n-c}}(\mcN \circ T)$,
    the solid blue line shows the causal achievable rate $\rqad_{\text{caus}}(\mcN \circ T)$,
    and the dotted purple line shows the Holevo capacity without CSI~\eqref{Equation:Rewrtie_Without_CSI}.
    The inset shows a zoomed view near $p = \frac{3}{4}$, highlighting the gap between the non-causal and causal rates.
    }
    \label{fig:noiseless_comparison}
\end{figure}

\subsection{Non-Causal CSI}
\label{subsection:example_noncausal}
In the non-causal case, the encoder sees the entire first-pass output sequence before deciding how to rewrite, which enables a more refined strategy than simply correcting all detected errors.
Instead of rewriting every flipped bit, the encoder can selectively rewrite only a subset of them, using the pattern of rewrites itself as an additional degree of freedom to convey information.
This creates a tradeoff: rewriting more positions reduces noise and improves reliability, but rewriting fewer preserves randomness that can be exploited to encode extra bits of information through the choice of which locations are rewritten.

For the case of non-causal CSI at the transmitter, as in Theorem~\ref{Theorem:non-causal}, we
consider the rate formula in~\eqref{eq: R-nonc}.
We use the same channel model, action ensemble $\{\sigma^u_G = \ketbra{u}_G\}$ with
$U \sim \mathrm{Bernoulli}(1/2)$, and bipartite state $\ket{\sigma^u_{SS_0}}$ as
in~\eqref{eq:sigma_SS0_causal_example}.
As before, Alice performs a projective measurement on $S_0$ with projectors
$\{\ketbra{\theta^u_{03}}, \ketbra{\theta^u_{12}}\}$, obtaining the outcome associated
with $\ket{\theta^u_{12}}$ (bit-flip) with probability $\delta$, and $\ket{\theta^u_{03}}$
(no bit-flip) with probability $1-\delta$.

We introduce a parameter $\alpha \in [0,1]$ and set $V \in \{0,1\}$ as follows.
Based on the measurement outcome, she sets $V=1$ with probability $\alpha$ if a bit-flip is detected,
and $V=0$ otherwise. The resulting conditional probabilities are
\begin{align}
    p_{V|U}(0|u) = 1 - \delta\alpha, \qquad p_{V|U}(1|u) = \delta\alpha.
\end{align}
%
% \begin{align}
% \Pr(V=v|12)=
% \begin{cases}
% \alpha& v=1
% \\
% 1-\alpha& v=0
% \end{cases}
% \\
% \Pr(V=v|03)=
% \begin{cases}
% 1& v=0
% \\
% 0& v=1
% \end{cases}
% \end{align}
% Equivalently,
% \begin{align}
%     %p_{V|U}(|u) = 1 - \delta\alpha, \qquad p_{V|U}(1|u) = \delta\alpha.
% \Pr(V=v|U=u, \ket{\bar u}_S)=
% \begin{cases}
% \alpha& v=1
% \\
% 1-\alpha& v=0
% \end{cases}
% \\
% \Pr(V=v|U=u, \ket{ u}_S)=
% \begin{cases}
% 1& v=0
% \\
% 0& v=1
% \end{cases}
% \end{align}
% Hence,
% \begin{align}
%     %p_{V|U}(|u) = 1 - \delta\alpha, \qquad p_{V|U}(1|u) = \delta\alpha.
% p_{V|U}(v|u)=\delta\cdot
% \begin{cases}
% \alpha& v=1
% \\
% 1-\alpha& v=0
% \end{cases}
% +(1-\delta)\cdot
% \begin{cases}
% 0& v=1
% \\
% 1& v=0
% \end{cases}
% \end{align}
Intuitively, $\alpha$ is the fraction of bit flips that are \emph{not} rewritten.

The states $\rho^{u}_{VSA}$ are defined as follows.
If no bit-flip occurred, then
 Alice sends the idle state $\ket{\perp}_A$ and sets $V=0$.
 Otherwise, if a bit-flip was detected, then she randomly chooses whether to perform a rewrite or not. Specifically, she sends the idle state $\ket{\perp}_A$ with probability $\alpha$
 ($V=1$), and sends a rewrite $\ket{u}_A$ with probability $1-\alpha$
 ($V=0$).
Overall, we have
\begin{align}
    \rho_{SVA}^{u} &= (1-\delta)   \ketbra{u}_S\otimes\ketbra{0}_V \otimes\ketbra{\perp}_A + 
    \delta  \ketbra{\bar u}_S\otimes (\alpha \ketbra{1}_V \otimes \ketbra{\perp}_A + (1-\alpha) \ketbra{0}_V \otimes \ketbra{u}_A)
\end{align}
Equivalently,
\begin{align}
   \rho_{VSA}^{u} &=\sum_{v\in\{0,1\}} p_{V|U}(v|u) \ketbra{v}_V\otimes \rho_{SA}^{v,u}
\end{align}
where
\begin{align}
    \rho^{0,u}_{SA}
    &= \frac{(1-\delta)\ketbra{u}_S\otimes\ketbra{\perp}_A
          + \delta(1-\alpha)\ketbra{\bar{u}}_S\otimes\ketbra{u}_A}
         {1-\delta\alpha}.
    \label{eq:rho_0u_nc}
\intertext{and}
% For $V=1$, a bit-flip was detected but Alice chooses not to rewrite, transmitting idle 
% $\ket{\perp}_A$ regardless:
% \begin{align}
    \rho^{1,u}_{SA} &= \ketbra{\bar{u}}_S\otimes\ketbra{\perp}_A.
    \label{eq:rho_1u_nc}
\end{align}

Optimizing over $\alpha \in [0,1]$, we obtain the following achievable rate with action dependence and non-causal CSI:
\begin{align}
    \capc_{\text{n-c}}(\mcN \circ T)
    &\geq \rqad_{\text{n-c}}(\mcN \circ T) \nonumber \\
    &\geq \max_{0 \leq \alpha \leq 1}
    \left\{
        1 - h_2(\delta)
        + (1-\delta\alpha)\!\left[
            h_2\!\left(\frac{\delta(1-\alpha)}{1-\delta\alpha}\right)
            - h_2\!\left(\frac{\delta^2(1-\alpha)}{1-\delta\alpha}\right)
        \right]
    \right\}
    \label{eq:nc_rate_example}
\end{align}
as derived in Appendix~\ref{appendix:example_derivations}.
At $\alpha = 0$, all bit flips are rewritten. The rate then reduces to
$1 - h_2(\delta^2)$, recovering the causal lower bound in~\eqref{Equation:Rewrtie_Causal_CSI}.
For $\alpha >0$, the rate in~\eqref{eq:nc_rate_example} can be larger, 
% analogously to Weissman's classical result~\cite[eq.~(52)]{weissman2010}, 
demonstrating 
an improvement
% that non-causal CSI improves 
upon our lower bound for the causal CSI case.
As illustrated in Figure~\ref{fig:noiseless_comparison}, the gap between the causal and non-causal bounds increases with the depolarization parameter~$p$.
At the boundary point $p = \frac{3}{4}$, where $\delta = \frac{1}{2}$,
% one can verify analytically that 
the optimal rewrite parameter is $\alpha^* = \frac{1}{3}$.
The non-causal achievable rate then evaluates to
$
    \rqad_{\text{n-c}}(\mcN \circ T)%\Big|_{p=3/4}
    \geq \frac{5}{6}\!\left[h_2\!\left(\frac{2}{5}\right) - h_2\!\left(\frac{1}{5}\right)\right] 
    \approx 0.208 ,
$
while the causal bound~\eqref{Equation:Rewrtie_Causal_CSI} at the same point gives $1 - h_2\!\left(\frac{1}{4}\right) \approx 0.189$, a gap of approximately $0.019$ bits per channel use.

\section{One-Shot Code Construction} 
\label{section:Coding}
In this section, we introduce a coding scheme for the one-shot setting, of $n=1$, over the quantum action-dependent channel with non-causal side information (see Figure~\ref{fig:environment-dependent-channel}). 
% by using the one-shot techniques as in \cite{Anshu2020, zivarifard2024covert}. 
Let $\mcN_{SA \to B}%:\, \mathscr{D}(\mcH_S \otimes \mcH_A) \to \mathscr{D}(\mcH_B)
$ be a quantum action-dependent channel. 
Let $T_{G \to SS_0}%:\, \mathscr{D}(\mcH_G) \to \mathscr{D}(\mcH_S \otimes \mcH_{S_0})
$ be the Stinespring dilation action channel $\mcT_{G \to S_0}$, that generates Alice's side information subsystem $S_0$ and the environment subsystem $S$.
Consider the quantum states defined in Theorem~\ref{Theorem:non-causal}:
 $\ket{\sigma_{SS_0}^u}$ is the pure state produced by  $T_{G \to SS_0}$, hence, %, which is shared between Alice and the main channel $\mcN_{SA \to B}$.
 $\rho_{VUSA}=\sum_v\sum_u p_{VU}(v,u)\ketbra{v}_V\otimes \ketbra{u}_U\otimes\rho_{SA}^{v,u}$ is the channel input, and the corresponding output is $\rho_{VUB}=  \mathrm{id}_{UV}\otimes\mathcal{N}_{SA\to B}(\rho_{UVSA})$. % be defined as in Theorem~\ref{Theorem:non-causal}, we define the following states:
% \begin{align}
%   \rho_{VUB} &=\id_V \otimes \id_U \otimes \mcN_{SA \to B}(\rho_{VUSA}), \label{eq:rho_VUB} \\
%   \rho_{B} &= \Tr_{VU}(\rho_{VUB}) =  \mcN_{SA \to B}(\rho_{SA}). 
%   % \sigma_{S S_0} &= \mcT_{G \to SS_0}(\rho_G), \label{eq:sigma_SS0} 
% \end{align}

We now describe the one-shot coding scheme. 

\subsection{Codebook}
\label{subsection:noncausal_codebook}
% Codebook Generation
Let $R,\;R_{S}>0$ denote the coding rates corresponding to the information and action encodings.
First, the action-encoding codebook % that is being used for action encoding: 
$\mcC_U := \{u(m)\}_{m\in \{1,\ldots,2^{R}\}}$ is sampled from a random set of i.i.d. codebooks, 
% $\mbfC_U$
 distributed according to $p_{U}$.
Then, let $\{\mcC_V(m)\}_{m \in \{1,\ldots,2^R\}}$ be $2^R$ subcodebooks, such that $\mcC_V(m) := \{v(m,1), v(m,2), \ldots, v(m,2^{R_S})\}$, and $\{v(m,\ell)\}_{\ell \in\{1, \ldots, 2^{R_S} \}}$ are drawn independently according to $p_{V|U}(\cdot |u(m))$. 
Both are revealed to Alice and Bob. The overall codebook is  $\mcC := \mcC_U \cup \{\mcC_V(m)\}$.
% Throughout the paper, 
We use the notation $\mcC$ for a deterministic codebook and $\mathbf{C}$ for a random codebook. % that is generated according to $p_{U}p_{V|U}$. 

\subsection{Encoder}
Our encoding scheme consists of two parts, action encoding and message encoding:
% \begin{enumerate}
  \subsubsection{Action Encoding}
  % For each value $u$, let  $\ket{\sigma_{G K_0}^{u}}$ be  a purification of the state $\sigma_{G}^{u}$, with $K_0$ as a reference.
 For each value $u$, let  $\ket{\sigma_{G}^{u}}$ be a pure state on $\mcH_G$.
  Given a message $m$, Alice prepares $\ket{\sigma_{G}^{u(m)}}$, and transmits $G$ through the action channel $\mcT_{G \to S_0}$.

% Alice prepares the state %As a result, we have the following purification of the state $\frac{1}{\sqrt{R_S} } \sum_{i} \rho_{S}^{v(m,\ell)}$: 
% %
% \begin{align}
%   \ket{\sigma_{G K_0 L}} = \frac{1}{\sqrt{2^{R_S}} } \sum_{\ell=1}^{2^{R_S}} \ket{\sigma_{G K_0}^{u(\ell)}} \otimes \ket{\ell}.
% \end{align}
%    
  % For each message $m\in\mcM$ Alice applies an action encoder map $\mcJ$ and prepares an input state $\varphi_{G^n}^{(m)} \in \mathscr{D}(\mcH_G^{\otimes n})$:
  % \begin{align}
  %   \varphi_{G}^{(m)} =   \sigma_{G}^{u(m)}, \quad u(m) \sim P_U.
  % \end{align}
 % Here, $\varphi_{G}^{u_i(m)}$ is a pure, and $u_i(m)$ is an action which is chosen according to an i.i.d. distribution $P_U$.  
 Consider a Stinespring dilation of 
  the action channel, with an isometry $T_{G \to S S_0}$,  where $S$ is the channel environment and $S_0$ is Alice's side information. 
  Upon the action encoding above,  this channel acts on $\ket{\sigma_{G}^{u(m)}}$ to produce the joint state $\ket{\sigma^{u(m)}_{S S_0}}$:
  \begin{align}
    \ket{\sigma_{S S_0}^{u(m)}} &=   T_{G\to S S_0 } \ket{\sigma_{G}^{u(m)}},
   %  \\
   % &=\mcT_{G\to S S_0} \bigl( \sigma_{G}^{u(m)}\bigr)
   % \\
   % &=\sigma_{SS_0}^{u(m)}
  \end{align}
  such that
  \begin{align}
      \sigma_S^{u(m)} = \Tr_{S_0}\ketbra{\sigma_{S S_0}^{u(m)}}{\sigma_{S S_0}^{u(m)}}.
  \end{align}

  \subsubsection{Message Encoding}
  % Before we describe the encoding, we introduce the following notation.
  %  Suppose that Alice would like to encode  $m$ and $\ell$. 
  % %  , where 
  % % $m\in \{1,\ldots,2^R\}$ and $\ell\in \{1,\ldots,2^{R_S}\}$.
  % By using the side-information subsystem $S_0$, %for each message $m\in\mcM$ 
  % Alice implements the encoding map $\mcE^{v(m,\ell)}_{S_0 \to A}$ on the side-information subsystem $S_0$. 
  % Consider a Stinespring dilation %of the  encoding map $\mcE^{{m}}_{S_0 \to A}$ 
  % with an isometry $E_{S_0\to  A T }^{{v(m,\ell)}}$. This produces the 
  % channel input state %$\rho_{S A}^{u(m),v(m,\ell)}$:
  % Consider $\ket{\rho_{SAT}^{v(m,\ell), u(m)}}$ to be the purification of $\rho_{SA}^{v(m,\ell),u(m)}$, thus Alice's state is
  Suppose Alice would like to encode message $m$ and subcodebook index $\ell$.
  For each index pair $(v(m,\ell), u(m))$,  let $\ket{\rho_{SAT}^{v(m,\ell),u(m)}}$ be a purification of $\rho_{SA}^{v(m,\ell),u(m)}$, i.e.,
  \begin{align}
      \rho_{SA}^{v(m,\ell),u(m)} = \Tr_T \ketbra{\rho_{SAT}^{v(m,\ell),u(m)}}
  \end{align}
  where $T$ is an auxiliary purifying register.
  
  Alice prepares the superposition state
\begin{align}
  \ket{\phi_{SATL}^{m}} = \frac{1}{\sqrt{2^{R_S}} } \sum_{\ell=1}^{2^{R_S}} \ket{\rho_{SAT}^{v(m,\ell),u(m)}} \otimes \ket{\ell},
  \label{Equation:Lsuperposition}
\end{align}
hence
\begin{align}
    % \phi_S^{m}&= 
    \Tr_{TL} \ketbra{\phi_{SATL}^{m}}{\phi_{SATL}^{m}}  
     = \frac{1}{2^{R_S}} \sum_{\ell = 1}^{2^{R_S}} \rho_{SA}^{v(m,\ell),u(m)}.
\end{align}
% when $\rho_{S}^{v(m,i)} = \Tr_A[\rho_{SA}^{v(m,i)}]$.
% \begin{remark}
%     The isometry $E^{m}_{S_0\to A K_2}$ represents any CPTP encoding map, including those that internally perform a measurement on $S_0$. In this case, part of the auxiliary register $K_2$ can be taken as a classical outcome register $V$, while the remaining subsystem plays the role of the environment. 
% \end{remark}
The corresponding output state is
\begin{align}
\hat{\Theta}_B(m)=\frac{1}{2^{R_S}} \sum_{\ell = 1}^{2^{R_S}} \rho_{B}^{v(m,\ell),u(m)}.
\end{align}

According to {Uhlmann’s Theorem} \cite{tomamichel2015quantum}, for every pair of %states $\rho_A$ and $ \sigma_A$, there exist 
    purifications
     $\ket{\psi}_{AB}$ and $\ket{\phi}_{AC}$ of $\rho_A$ and $\sigma_A$, respectively, there exists an isometry $W_{C \to B}$ such that
     $F(\rho_A, \sigma_A) = F(\ketbra{\psi}{\psi}_{AB},  W(\ketbra{\phi}{\phi}_{AC}) W^{\dag})$.
Then, in our case, %From Uhlmann's theorem 
it follows that there exists a set of isometries,
\begin{align}
  \{W_{S_0 \rightarrow ATL}^{m}\}_{m \in \{1,\ldots,2^{R}\}} \in \mathscr{L}\left(\mcH_{S_0} \to \mcH_A \otimes \mcH_{T}\otimes \mcH_{L} \right)
\end{align} 
% that map $\ket{\sigma_{S S_0}^{u(m)}}$ into $ \ket{\phi_{SATL}^{m}} $. %, or equivalently,
% from $\ket{\sigma_{S S_0}^{u(m)}}$ to $\ket{\rho^{v(m,\ell),u(m)}_{SAT}}$.
such that 
$F(\phi_{S}^{m}, \sigma_{S }^{u(m)}) = F\left(\ketbra{\phi_{SATL}^{m}},  W_{S_0 \rightarrow ATL}^{m}\ketbra{\sigma_{S S_0}^{u(m)}} (W_{S_0 \rightarrow ATL}^{m})^{\dag}\right)$.
Similarly, the purified distance satisfies %using the short notation $\tilde W^{m} = \mbone_S \otimes W_{S_0 \rightarrow ATL}^{m} $,
\begin{align}
  P\left( \ketbra{\phi_{SATL}^{m}}{\phi_{SATL}^{m}} ,  W^{m} (\sigma^{u(m)}_{SS_0})  W^{m \dag}  \right) 
  % &= P\left(\frac{1}{2^{R_S}} \sum_{\ell \in \msL} \rho_{S}^{u(\ell),v(m,\ell)}\otimes \ketbra{\ell}, \frac{1}{2^{R_S}} \sum_{\ell \in \msL} \sigma_{S}^{u(\ell)}\otimes \ketbra{\ell} \right)
  %   \\
  &= P\left(\frac{1}{2^{R_S}} \sum_{\ell =1 }^{2^{R_S}} \rho_{S}^{v(m,\ell),u(m)},  \sigma_{S}^{u(m)} \right)
    \label{eq:enc_purd}. 
\end{align}

The message encoder is defined such that
given a message $m$, Alice applies the isometry $W_{S_0 \rightarrow ATL}^{m}$ on $S_0$, %$\ket{\sigma_{S S_0 }^{u(m)}}$ 
% mapping $\ket{\sigma_{S S_0}^{u(m)}}$ into $ \ket{\phi_{SATL}^{m}} $,
and transmits $A$ over the action-dependent channel.

\begin{remark}  
Alice never needs to physically
``choose'' an index $\ell$. Rather, she prepares the 
coherent state $\ket{\phi^m_{SATL}}$ in \eqref{Equation:Lsuperposition}, and then applies an isometry $W^m_{S_0\to ATL}$.
Thus, the encoding procedure is fully specified and requires no random selection of a subcodeword.
Tracing out $L$ yields  the average over the sub-codebook $\{v(m,\ell)\}_{\ell}$. 
This is not an accidental consequence of the construction, but rather the intended purpose. %Alice never needs to 
% ``choose'' $\ell$ physically: 
The register $L$ serves as a purification, allowing us to invoke Uhlmann's theorem and replace the conceptual description ``Alice randomizes over $2^{R_S}$
 subcodewords" by a single, physically implementable encoding isometry acting on the pure state
 $\sigma_{SS_0}^{u(m)}$.
 \end{remark}

% \end{enumerate}
\subsection{Decoding}
  \label{Subsection:Decoding}
% \begin{enumerate}   
  % \item \textbf{Channel Transmission:} The memoryless quantum channel 
  % \begin{align}
  % \mcN^{\otimes n}_{SA \to B}:\, D(\mcH_{S}^{\otimes n} \otimes \mcH_{A}^{\otimes n})\to D(\mcH_{B}^{\otimes n})
  % \end{align}
  % acts on the input state $\rho_{S^n A^n} \in D(\mcH_{S}^{\otimes n} \otimes \mcH_{A}^{\otimes n})$ yielding the channel output 
  % \begin{align}
  % \rho_{B^n}^{(m)} =\mcN^{\otimes n}_{SA \to B} \bigl(\rho_{S^n A^n}^{(m)}\bigr).
  % \end{align}
  
%   \item \textbf{Decoding:} At the receiver (Bob), a POVM $\{D_m\}_{m\in \mcM}$ is employed to infer the transmitted message 
%   from $\rho_{B^n}^{(m)}$. The code is said to have an average error probability not exceeding $\varepsilon$ if
%   \begin{align}
%   \bar{p}_e = 1 -\frac{1}{M}\sum_{m\in \mcM} \Tr\Bigl[D_m\, \rho_{B^n}^{(m)}\Bigr] \leq \varepsilon. \label{eq:avg_error_defenition}
%   \end{align}
% \end{enumerate}

% A rate
% \begin{align}
% R = \frac{1}{n}\log M
% \end{align}
% is achievable if there exists a sequence of
%  $(n, R, \varepsilon)$ codes with $\varepsilon \to 0$ as $n \to \infty.$ The capacity of the action-dependent quantum channel is defined as the supremum of all achievable rates.

% \subsection{Decoding}
We would like to devise a POVM measurement $\{D_m\}$ on the system $B$ that
distinguishes between the states $\{\rho_{B}^{(m)}\}_{m\in \{1,\ldots,2^{nR}\}}$ with high probability. Let $\mcE_1$ 
be the pinching map associated with $\rho_{VU} \otimes \rho_B $ such that:
  $\rho_{VU} \otimes \rho_B = \sum_{\lambda} \lambda \Pi_{\lambda}$
is the spectral decomposition of $\rho_{VU} \otimes \rho_B$. The pinching map $\mcE_1$ is defined as:
$\mcE_1(\rho) = \sum_{\lambda} \Pi_{\lambda} \rho \Pi_{\lambda},$
where $\{\Pi_{\lambda}\}$ are the orthogonal projectors onto the eigenspace of $\rho_{VU} \otimes \rho_B$.
We define $\nu_1$ as the number of distinct nonnegative eigenvalues of $\rho_{VU} \otimes \rho_B$.
Now ${\mcE_1(\rho_{VUB})}$ is block-diagonal in the eigenbasis of $\rho_{VU} \otimes \rho_B$, thus it commutes with $\rho_{VU} \otimes \rho_B$.
The pinching of $\rho_{VUB}$ with respect to $\rho_{VU} \otimes \rho_B$ is defined as:
%
% \pagebreak
\begin{align}
  \mcE_1(\rho_{VUB}) = &\sum_{\lambda = 1}^{\nu_1} \Pi_{\lambda} \left( \sum_{v,u} p_{VU}(v,u) \ketbra{v}{v}_V \otimes \ketbra{u}{u}_U \otimes \rho_{B}^{v,u} \right)  \Pi_{\lambda} 
\end{align}

For two Hermitian matrices $A$ and $B$, we define the projection $\{A \geq B\}$  as 
$\sum_{\lambda \geq 0} P_{\lambda} $ , where the spectral decomposition of $A-B$ is given as
$\sum_{\lambda } \lambda P_{\lambda} $. In this notation, $P_{\lambda}$ is the projection to the eigenspace corresponding
to the eigenvalue $\lambda$. Then, let
\begin{align}
    \Pi_{VUB} = \{\mcE_1 \left(\rho_{VUB} \right) \geq 2^{R + R_S} \rho_{VU} \otimes \rho_B\}.
    \label{eq:Pi_VUB_def}
\end{align}
For every $m \in \{1,\ldots,2^{R}\}, \ell\ \in \{1,\ldots,2^{R_S}\}$, we define:
\begin{align}
    \gamma(m,\ell) = 
    &\Tr_{VU} \left[\Pi_{VUB} \left(\ketbra{v(m,\ell)}{v(m,\ell)} \otimes \ketbra{u(m)}{u(m)} \otimes \mbone_B \right)\right] .
    \label{eq:gamma_def}
\end{align}
Our set of POVM is then normalized as
\begin{align}
    \beta(m,\ell) = \left( \sum_{m',\ell'} \gamma(m',\ell') \right)^{-\frac{1}{2}} \gamma(m,\ell)  \left( \sum_{m',\ell'} \gamma(m',\ell') \right)^{-\frac{1}{2}}.
\end{align}

\subsection{Error Bound}
We are now ready to state our one-shot result. %regarding the one-shot average error probability.
\begin{proposition} [One-shot error probability] \label{Theorem:1}
Let $\alpha \in \left(0,\frac{1}{2} \right)$.
    Then, the average error probability is bounded by:
    \begin{align}
       \mathbb{E}_{\mathbf{C}}[ \bar{p}_e^{(1)}] &\leq   12 \cdot \nu_{1}^{\alpha} 2^{\alpha \left[R + R_S - \tilde{D}_{1-\alpha}\left(\rho_{VUB} \| \rho_{VU} \otimes \rho_B \right)\right]} 
        +\frac{2}{\alpha}  \frac{\nu_2^\alpha}{2^{\alpha R_S}}  2^{\alpha \tilde{D}_{1+\alpha}(\rho_{VUS} \| \rho_{V-U-S})}
        \intertext{with}
        &\rho_{V-U-S} = \sum_{u} p_{U}(u) \rho_{V}^u \otimes \ketbra{u}{u} \otimes  \sigma_{S}^u
    \end{align}    
        % &+\frac{2}{\alpha \ln 2} \left( \frac{\nu_2}{2^{R}} \right)^{\alpha} 2^{\alpha \tilde{D}_{1+\alpha}(\rho_{US} \| \rho_{U} \otimes \rho_{S})} 
        % \nonumber\\
        % & + \frac{2}{\alpha \ln 2} \left( \frac{\nu_3}{2^{R + R_S}} \right)^{\alpha} 2^{\alpha \tilde{D}_{1+\alpha}(\rho_{VUS} \| \rho_{VU} \otimes \rho_{S})} ,
    % \end{align}
where  $\nu_1$ is the number of distinct eigenvalues of $\rho_{VU} \otimes \rho_B$, and  $\nu_2$ is the maximum number of distinct eigenvalues of $\{\sigma_S^{u(m)}\}_{\forall m \in \{1,\ldots,2^{R}\}}$. 
% and $\mcE_S$ is the pinching map with respect to $\rho_S$.
\end{proposition}
The proof of Proposition~\ref{Theorem:1} is given in Appendix~\ref{appendix:Th1}.  The outline is given below.

\subsection{Proof Outline}
\label{Subsection:Outline}
First, we show that 
    \begin{align}
        \mathbb{E}_{\mbfC}[\bar{p}_e^{(1)}]  
        &{\leq} 2\mbE_{\mbfC} \left\{\Tr\left[\left(\sum_{m' \neq 1, \ell} \beta(m',\ell)\right)  \hat \Theta_B(1)\right] \right\}
        + 2\mbE_{\mbfC} \left\{P\left(\Theta_B(1),  \hat \Theta_B(1) \right)^2 \right\}
        \label{Eq:Error_Outline}
    \end{align}
where
 $\Theta_B(m)$ Bob's output state 
    given that the message $m$ was sent, 
    and
$
         \hat \Theta_B(m) = \frac{1}{2^{R_S}}\sum_{\ell = 1}^{2^{R_S}}  \rho_{B}^{v(m,\ell),u(m)}
$
is the average output state,  when averaged over the subcodebook. 

Intuitively, bounding the second error term requires showing the codebook average state $\Theta_B(1)$ is close to the probabilistic average state $\hat \Theta_B(1)$.
To this end, we use the subcodebook property below: 
% \begin{lemma}% [based on {\cite[Lemma 7]{Anshu2020}}] 
% \label{lem:1}
%     Let $\rho_{VUS} = \Tr\left[\rho_{VUSA}\right]$ be a classical-quantum state, and $\alpha \in (0,\frac{1}{2})$. Furthermore, let $\mcC_m = \{v(m,1), \ldots v(2^{m,R_S}), u(m)\}$ be a collection of random 
%     variables, such that 
%     the sequence $v(m,1),\ldots,v(m,2^{R_S})$ is conditionally i.i.d.
%     $\sim p_{V|U}(\cdot| u(m)) $, for every given $u(m)$.
%     % we have the following probability distributions:
%     % % for every $\ell \in \{1,\ldots, 2^{R_S}\}
%     % %= \msL
%     % % ,\;
%     % $v(m,\ell) \sim p_{V|U}(\cdot| u(m)), \; \left(v(m,\ell), v(m, \ell') \right) \sim p_{V|U} (\cdot| u(m))\cdot p_{V|U}(\cdot| u(m))$, 
%     % for all $\ell \in \{1,\ldots, 2^{R_S}\}$.
%     % % Let $C = \{U(1),\ldots, U(2^{R_S}), V(m,1), \dots V(m,R_s)\}$ be a set of random variables generated i.i.d., where $U \sim p_U$ and $V(\cdot|\ell) \sim P_{V|}(\cdot|U(\ell))$. 
%     We consider the following state:
%     \begin{align}
%         \tau_{S|\mcC_m} \triangleq \frac{1}{2^{R_S}} \sum_{\ell \in \msL} \rho_{S}^{v(m,\ell),u(m)}. \label{eq:tau_S_Cm}
%     \end{align}
%     % and 
%     % \begin{align*}
%     %     \hat \tau_{S} \triangleq \frac{1}{2^{R_S}} \sum_{\ell \in \msL} \rho_{S}^{u(\ell)}.
%     % \end{align*}
%     %
%     Then, %for $\alpha \in (0,\frac{1}{2})$,
%     there exists a constant $\nu_2 \geq 0$ such that:
    \begin{align}
        \mathbb{E}_{\mbfC}\left\{\tilde{D}_{1+\alpha}\left(\frac{1}{2^{R_S}} \sum_{\ell =1}^{2^{R_S}} \rho_{S}^{v(m,\ell),u(m)}\| \sigma_S^{u(m)} \right)\right\} \leq 
        \frac{1}{\alpha \ln 2}  \frac{\nu_2^{\alpha}}{2^{\alpha R_S}}  2^{\alpha \tilde{D}_{1+\alpha}(\rho_{VUS} \| \rho_{V-U-S})} 
    \end{align}
   where %$\rho_{S}^u \triangleq \sum_v p_{V|U}(v|u) \rho_{S|u,v}$, and  
    $\rho_{V-U-S} = \sum_{u} p_{U}(u) \rho_{V}^u \otimes \ketbra{u}{u} \otimes  \sigma_{S}^u$ and  $\alpha \in (0,\frac{1}{2})$. %, and $\nu_2$ is the maximum number of distinct eigenvalues of the states $\{\rho_{S|u}\}_u$.
% \end{lemma}
This property is shown in Appendix~\ref{appendix:proof_of_lem1}.

As for the first error term on the right-hand side of \eqref{Eq:Error_Outline},  we use the Hayashi-Nagaoka inequality in order to show to bound this error term by 
    \begin{align}
      \mbE_{\mbfC} \left\{\Tr\left[\left(\sum_{m' \neq 1, \ell} \beta(m',\ell)\right)  \hat \Theta_B(1)\right] \right\}
      &\leq  2\mbE_{\mbfC} \left\{\Tr[\left( \mbone - \gamma(1,1) \right)  \rho_B^{v(1,1), u(1)} ]\right\}
      \nonumber\\ &\phantom{\leq}
      + 4 \sum_{(m',\ell) \neq (1,1)} \mbE_{\mbfC} \left\{\Tr[\left( \gamma(m',\ell) \right)  \rho_B^{v(1,1), u(1)} ]\right\}.
    \end{align}
The bound on the first term on the right-hand side of \eqref{Eq:Error_Outline} is then obtained as a consequence of the projector properties in Appendix~\ref{appendix:proof_of_lem2}:
 % \begin{lemma} \label{lem:2}
 %      For every $\alpha \in \left(0,\frac{1}{2}\right)$, and $R, R_S > 0$ we have:
      \begin{align}
        \Tr[\left( \mbone - \Pi_{VUB}\right) \rho_{VUB}] 
        \leq &\nu_{1}^{\alpha} 2^{\alpha \left[R + R_S - \tilde{D}_{1-\alpha}\left(\rho_{VUB} \| \rho_{VU} \otimes \rho_B \right)\right]}, 
        \label{eq:first_term_lemma_2_def}  \\ 
         2^{R + R_S} \Tr \left[\Pi_{VUB}\left(\rho_{VU} \otimes \rho_B \right)  \right] 
        \leq & \nu_{1}^{\alpha} 2^{\alpha \left[R + R_S- \tilde{D}_{1-\alpha}\left(\rho_{VUB} \| \rho_{VU} \otimes \rho_B \right)\right]}.  \label{eq:second_term_lemma_2_def}
      \end{align}
    % \end{lemma}
    % The Proof of Lemma \ref{lem:2} is given in 
    The details are given in Appendix~\ref{appendix:Th1}.
    This concludes the outline of the error analysis. 
%

%----------------------------------------------
    % Proof of Theorem 2
%----------------------------------------------
\section{Proof of Theorem \ref{Theorem:non-causal}}
\label{section:proof of main theorem}
% The proof consists of two parts: achievability and converse. The cardinality bounds on the auxiliary variables are established in Appendix~\ref{appendix:cardinality}.

\subsection{Achievability Proof}
We consider the average error probability as the number of channel uses $n$ grows to infinity.
Let $\varepsilon>0$.
Let $d_{VUB}$ be the dimension of $\mcH_{V} \otimes \mcH_{U} \otimes \mcH_{B}$, and $d_S$ be the dimension of $\mcH_{S}$. 
As shown in \cite[Lemma~3.9]{Hayashi2017QIT}, we can bound $\nu_1$ and $\nu_2$ as follows:
\begin{align}
  \nu_1 &\leq \left(n+1\right)^{d_{VUB} - 1}, \nonumber\\ 
  \nu_2 &\leq \left(n+1\right)^{d_{S} - 1}.
\end{align}
Based on Proposition \ref{Theorem:1} for $n$ uses of the channel, there exists a (deterministic) codebook $\mcC$ such that:
\begin{align}
  \bar{p}_e^{(n)} 
  &\leq  12 \left(n+1\right)^{\alpha(d_{VUB} - 1)} 2^{\alpha \left[n(R + R_S) - \tilde{D}_{1-\alpha}\left(\rho_{VUB}^{\otimes n} \| \rho_{VU}^{\otimes n} \otimes \rho_B^{\otimes n} \right)\right]}  
  \nonumber\\ 
    &+\frac{2}{\alpha} (n+1)^{\alpha(d_S - 1)}  2^{\alpha \left[ -nR_S + \tilde{D}_{1+\alpha}(\rho_{VUS}^{\otimes n} \| \rho_{V-U-S}^{\otimes n}) \right]}
    \nonumber \\
    &= 2^{-n \left[-\alpha(R + R_S) -\frac{\log(12)}{n} - \frac{\alpha(d_{VUB} -1)}{n}\log(n+1) + \frac{\alpha}{n} \tilde{D}_{1-\alpha}\left(\rho_{VUB}^{\otimes n} \| \rho_{VU}^{\otimes n} \otimes \rho_B^{\otimes n} \right)\right]} 
    \nonumber \\ 
    &+ 2^{-n \left[\alpha R_S -\frac{1}{n}\log(\frac{2}{\alpha}) -\frac{\alpha(d_S -1)}{n}\log(n+1)  - \frac{\alpha}{n}\tilde{D}_{1+\alpha}(\rho_{VUS}^{\otimes n} \| \rho_{V-U-S}^{\otimes n}) \right]}
\end{align}
Hence, the error probability tends to zero as $n\to \infty$, provided that 
\begin{align}
  \alpha(R + R_S) &<
  \frac{\alpha}{n} \tilde{D}_{1-\alpha}\left(\rho_{VUB}^{\otimes n} \| \rho_{VU}^{\otimes n} \otimes \rho_B^{\otimes n} \right)
  - \frac{\alpha(d_{VUB} - 1)}{n}\log \left(n+1\right) - \frac{\log(12)}{n},
  \label{eq:R+Rs_Bound}
\end{align}
 and 
 \begin{align}
   \alpha R_S &>
     \frac{\alpha}{n} \tilde{D}_{1+\alpha}\left(\rho_{VUS}^{\otimes n} \| \rho_{V-U-S}^{\otimes n}\right) 
     + \frac{\alpha(d_{S} - 1)}{n}\log \left(n+1\right) + \frac{1}{n}\log(\frac{2}{\alpha}). 
    % \\
    %  R + R_S &\geq
    %  \frac{1}{n} \tilde{D}_{1+\alpha}\left(\rho_{VUS}^{\otimes n} \| \rho_{VU}^{\otimes n} \otimes \rho_{S}^{\otimes n}\right) 
    %  \nonumber\\ & \phantom{\geq}
    %  + \frac{\alpha(d_{S} - 1)}{n}\log \left(n+1\right).
     \label{eq:Rs_Bound}
\end{align}
 % We achieve $\lim_{n\to\infty} \bar{P}_e^{(n)} = 0.$
% Thus, as

In the limit of $n \to \infty$ and $\alpha \to 0$, the bounds from \eqref{eq:R+Rs_Bound} 
and \eqref{eq:Rs_Bound} can be simplified. The last two terms %proportional to $\frac{\log(n+1)}{n}$ and $\frac{1}{n}$ 
vanish as $n \to \infty$. We then apply the additivity of the sandwiched R\'enyi divergence, $\tilde{D}_{\alpha}(\rho^{\otimes n} \| \sigma^{\otimes n}) = n \tilde{D}_{\alpha}(\rho \| \sigma)$, and its convergence to the quantum divergence as $\alpha \to 0$ \cite{tomamichel2015quantum, MuellerLennert2013}:
\begin{align}
    % R+R_S &< 
    \lim_{\alpha \to 0} \left[ \lim_{n \to \infty} \frac{1}{n}\tilde{D}_{1-\alpha}(\rho_{VUB}^{\otimes n} \| \rho_{VU}^{\otimes n} \otimes \rho_B^{\otimes n}) \right] 
            &= \lim_{\alpha \to 0} \left[ \tilde{D}_{1-\alpha}(\rho_{VUB} \| \rho_{VU} \otimes \rho_B) \right]
            \nonumber \\
            &=D(\rho_{VUB} \| \rho_{VU} \otimes \rho_B)
            \nonumber\\
            &= I(VU;B)_{\rho}, 
            \intertext{and similarly,}
     % R_S &> 
     \lim_{\alpha \to 0} \left[ \lim_{n \to \infty} \frac{1}{n}\tilde{D}_{1+\alpha}(\rho_{VUS}^{\otimes n} \| \rho_{V-U-S}^{\otimes n}) \right] 
     % &= \lim_{\alpha \to 0} 
     % \tilde{D}_{1+\alpha}(\rho_{VUS} \| \rho_{V-U-S})
     % \nonumber \\ 
     % &= D(\rho_{VUS} \| \rho_{V-U-S}) 
     &= I(V;S|U)_{\rho}.  
\end{align}
We deduce that 
the error probability satisfies $\bar{p}_e^{(n)}\leq \varepsilon$  provided that 
% To formalize the strict inequalities, we introduce an arbitrarily small $\delta > 0$, which leads directly to the achievable rate region
%
% In the limit of $n \to \infty$ and $\alpha\to 0$, the bounds reduce to \cite{tomamichel2015quantum, MuellerLennert2013} %\textcolor{blue}{add reference to Tomamichel's book}
 % and $\alpha \to 0$, the achievable rate of the channel $R_0$ satisfies:
\begin{align}
  R + R_S &< I(VU;B)_{\rho} - \frac{\delta}{3}, 
  \\
  R_S &> I(V;S|U)_{\rho} + \frac{\delta}{3}
\end{align}
for some $\delta>0$ and for sufficiently large $n$.
% Set
% % Let $R' = I(U;S)_{\rho} + \epsilon_1$, and set $R_S = I(VU;S)_{\rho} - R' + \epsilon_2 $, for $\epsilon_2 \geq \epsilon_1 >0$. Then, we have
% %
% \begin{align}
%   R_S   &= \max \left[  I(U;S)_{\rho} \,,\; I(V;S|U)_{\rho} \right] + \epsilon_1 
%   % \nonumber\\
%   % &= I(V;S|U)_{\rho} + \epsilon_2 -\epsilon_1
% \end{align}
% Plugging this into \eqref{Equation:Sum_Rate_Bound}, we deduce that the error probability tends to zero for
% \begin{align}
%    R &= %\lim_{n \to \infty} \frac{1}{n} D\left(\rho_{VUB}^{\otimes n} \| \rho_{VU}^{\otimes n} \otimes \rho_B^{\otimes n}\right) \\ 
%   I(VU;B)_{\rho} - R_S-\epsilon_2
%   \nonumber\\
%   &= I(VU;B)_{\rho} - I(V;S|U)_{\rho} -\epsilon_2 +\epsilon_1
% \end{align}
%
% Choose
% \begin{align}
%   R_S    &= I(V;S|U)_{\rho} +\alpha
% \end{align}
% Plugging this into \eqref{Equation:Sum_Rate_Bound}, 
Hence, a transmission rate $R$ such that
% We deduce that the error probability tends to zero for
\begin{align}
   R
  &< I(VU;B)_{\rho} - I(V;S|U)_{\rho} -\delta,
\end{align}
is achievable.

% It remains to establish the converse, showing that the rate $I(VU;B)_{\rho} - I(V;S|U)_{\rho}$ is also an upper bound on the capacity.

To show achievability for the regularized rate $\frac{1}{k} \rqad_{\text{n-c}} \left( \mcN^{\otimes k}\circ T^{\otimes k} \right)$, we employ the coding scheme above to the tensor-product action-dependent channel $\mcN_{SA\to B}^{\otimes k}\circ T_{G\to SS_0}^{\otimes k}$.
This completes the achievability proof.

 \subsection{Converse Proof}
We establish the converse part for the regularized capacity formula in Theorem~\ref{Theorem:non-causal}. 
The proof is straightforward and it relies on Fano's inequality and the Holevo bound.
% together with the observation that the $n$-letter state induced by an arbitrary 
% code is a feasible point of the optimization defining 
% $\rqad^{\text{n-c}} ( \mcN^{\otimes n}\circ T^{\otimes n} )$.
%
Consider a sequence of $(2^{nR},n,\varepsilon_n)$ codes such that $\varepsilon_n$ tends to zero as $n\to \infty$, 
and let $M$ be a uniformly distributed message from $\{1,\ldots,2^{nR}\}$.
As described in Section~\ref{section:ADCoding},
the action encoder maps the message $m$ into a pure 
action state $%\sigma_{G^n}^{(m)}=
\ket{\sigma_{G^n}^{m}}
$. 
As the action state is sent through the action channel $T^{\otimes n}_{G\to SS_0}$, the resulting   
 side-information state is %$\sigma_{S^n S_0^n}^{(m)}=\ketbra{\sigma_{S^n S_0^n}^{m}}$ with 
 \begin{align}
 \ket{\sigma_{S^n S_0^n}^{m}}=T^{\otimes n}_{G\to SS_0}\ket{\sigma_{G^n}^{m}}.
 \end{align}
 The corresponding classical-quantum state is
\begin{align}
     \sigma_{M S^n S_0^n} = \frac{1}{2^{nR}}\sum_{m} \ketbra{m}{m}_M\otimes \ketbra{\sigma_{S^n S_0^n}^{m}}.
    \label{eq:conv_rhoM}
\end{align}
 
 Next, Alice encodes the message $m$ and the quantum  CSI $S_0^n$ into
a channel input state $\rho_{MSA^n}$. Note that since the encoder does not act on $S^n$, we have $\sigma_{MS^n}=\rho_{MS^n}$. % $\rho_{S^n A^n}^{(m)}=\mathrm{id}_{S^n}\otimes \mcE^{(m)}_{S_0^n\to A^n}(\sigma_{S^n S_0^n}^{(m)})$,
The channel output is thus $\rho_{MB^n}=\mathrm{id}_M\otimes\mcN^{\otimes n}_{SA\to B}(\rho_{MS^n A^n})$. 
% $\rho_{B^n}^{(m)}=\mcN^{\otimes n}_{SA\to B}(\rho_{S^n A^n}^{(m)})$. We collect these into the message-indexed classical--quantum states
% \begin{align}
%     \rho_{M S^n A^n} = \frac{1}{2^{nR}}\sum_{m} \ketbra{m}{m}_M\otimes \rho_{S^n A^n}^{(m)},
%     \nonumber \\
%     \sigma_{M S^n S_0^n} = \frac{1}{2^{nR}}\sum_{m} \ketbra{m}{m}_M\otimes \sigma_{S^n S_0^n}^{(m)},
%     \label{eq:conv_rhoM}
% \end{align}
% with $\rho_{M B^n}=\mathrm{id}_{M}\otimes \mcN^{\otimes n}_{SA\to B}(\rho_{M S^n A^n})$. 
Bob applies a decoding POVM to obtain an estimate $\hat M$, with $\Pr(\hat M\neq M)\leq \varepsilon_n$.

By Fano's inequality %~\cite[Theorem 10.7.3]{wilde2017} 
and
%we have
% \begin{align}
%     H(M\mid \hat M)\leq  n\varepsilon_n', 
% \end{align} 
% with $\varepsilon_n'\eqqcolon \frac{1}{n}+\varepsilon_n R$. %\to 0$. 
% Since $\hat M$ is the outcome of a measurement on $B^n$, the Holevo bound and 
the data-processing 
inequality,  % yield
\begin{align}
    nR %= H(M) = I(M;\hat M) + H(M\mid \hat M)
    \leq I(M;B^n)_\rho + n\varepsilon_n'.
    \label{eq:conv_holevo}
\end{align}
where $\varepsilon_n'\eqqcolon \frac{1}{n}+\varepsilon_n R$
(see \cite[Th. 10.7.3 and 11.9.4]{wilde2017}).
As we set the auxiliary variables as $U=V=M$, we have 
$I(V;S^n\mid U)_\rho=0$. Hence, we may write
% We now bring this bound into the form of the regularized rate formula. Since the message is determined once we condition on it, the conditional mutual information $I(M;S^n\mid M)_\rho=0$ vanishes trivially, and we may subtract it from~\eqref{eq:conv_holevo} without affecting the bound, obtaining
\begin{align}
    R &\leq 
    % I(M;B^n)_\rho - I(M;S^n\mid M)_\rho + n\varepsilon_n'
    % \label{eq:conv_inM}\\
    % &= 
    \frac{1}{n}[ I(VU;B^n)_\rho - I(V;S^n\mid U)_\rho ]+ \varepsilon_n'.
    \label{eq:conv_inUV}
    \\
    &\leq \frac{1}{n}\, \rqad_{\text{n-c}} \left( \mcN^{\otimes n}\circ T^{\otimes n} \right).
    \label{eq:conv_capacity}
\end{align}
This completes the proof of Theorem~\ref{Theorem:non-causal}.
\qed

\section{Proof of Theorem \ref{Theorem:causal}}
\label{section:proof of causal theorem}

We prove achievability by constructing a coding scheme based on quantum Shannon strategies and applying the Quantum Packing Lemma.
First, we describe the coding scheme for the causal setting.
\paragraph*{Codebook Generation}
Let $R > 0$ denote the information rate. 
% In contrast to Section~\ref{subsection:noncausal_codebook}, here, the codebook construction removes the subcodebook 
% binning required for the non-causal case.
The action-encoding codebook $\mcC_U = \{u^n(m)\}_{m \in \{1,\ldots, 2^{nR}\}}$ consists of $2^{nR}$ independent action sequences, sampled according to $p_U$.
In addition, we generate $2^{nR}$ independent sequences of strategy indices $\mcC_V = \{v^n(m)\}_{m \in \{1,\ldots, 2^{nR}\}}$, sampled according to $p_V$, 
where each $v_i$ indexes a causal encoding map $\mathcal{F}_{S_0 \rightarrow A}^{v_i}$.
Both codebooks are revealed to Alice and Bob. We define the combined codeword $\tilde{u}^n(m) = (u^n(m), v^n(m))$, where $\tilde{u}_i(m) = (u_i(m), v_i(m))$. 
Since $U$ and $V$ are independent, we have $p_{\tilde{U}} = p_U \cdot p_V$. Henceforth, this combined codeword $\tilde{u}$ serves as the index for both the action state $\sigma_G^{\tilde{u}}$ and 
the encoding map $\mcF_{S_0 \rightarrow A}^{\tilde{u}}$.

%Our encoding scheme consists of two stages, action encoding and transmission encoding.

\paragraph*{Encoder}
For a given message $m$, at each time $i$, Alice prepares the action state $\ket{\sigma_{G}^{\tilde{u}_i(m)}}$, which induces the entangled state $\ket{\sigma_{SS_0}^{\tilde{u}_i(m)}}$ 
via the Stinespring dilation of the action isometry $T_{G \rightarrow S S_0}$.
She then applies the encoding map 
$\mathcal{F}_{S_0 \rightarrow A}^{\tilde{u}_i(m)}$ on ${S_{0}}_i$.
The channel input is then
% \begin{align}
% \rho_{S_1 A_1}^{(m)}&=\mathrm{id}_{S_1}\otimes \mathcal{F}_{S_0 \rightarrow A}^{\tilde{u}_1(m)} (\sigma_{SS_0}^{\tilde{u}_1(m)}) \\
% \rho_{S^2 A^2}^{(m)}&=\mathrm{id}_{S^2}\otimes \mathcal{F}_{S_0 \rightarrow A}^{\tilde{u}_1(m)} \otimes \mathcal{F}_{S_0 \rightarrow A}^{\tilde{u}_2(m)} (\sigma_{SS_0}^{\tilde{u}_1(m)}\otimes \sigma_{SS_0}^{\tilde{u}_2(m)})
% \\
% &\cdots
% \\
% \rho_{S^n A^n}^{(m)}&=\mathrm{id}_{S^n}\otimes \bigotimes_{i=1}^n \mathcal{F}_{S_0 \rightarrow A}^{\tilde{u}_i(m)} \left(\bigotimes_{i=1}^n \sigma_{SS_0}^{\tilde{u}_i(m)} \right)
% \end{align}
% That is, 
\begin{align}
\rho_{S^n A^n}^{(m)}&= \bigotimes_{i=1}^n
\mathrm{id}_{S}\otimes
\mathcal{F}_{S_0 \rightarrow A}^{\tilde{u}_i(m)} \left( \sigma_{SS_0}^{\tilde{u}_i(m)} \right)
= \bigotimes_{i=1}^n
\rho_{SA}^{\tilde{u}_i(m)} 
\end{align}
Thus, the output is 
\begin{align}
\rho_{B^n}^{(m)}=\bigotimes_{i=1}^n
\rho_{B}^{\tilde{u}_i(m)} .
\end{align}

These maps satisfy the causal consistency requirement, where the choice of map is governed by $\tilde{u}^n(m)$ and the input $A_i$ depends only on the causal side-information ${S_0}_1, \dots, {S_0}_i$. Specifically, the induced joint state satisfies the consistency requirements defined in \eqref{eq:causal-encoding}.

%We would like to design a POVM measurement $\{D_m\}$ on the system $B^n$ that distinguishes between the output states $\{\rho_{B^n}^{(v^n(m), u^n(m))}\}_{m\in\msM}$ with high probability.

\paragraph*{Decoding}
Bob performs a decoding measurement to distinguish between the $2^{nR}$ possible output states $\rho_{B^n}^{(m)}$.
We apply the Quantum Packing Lemma \cite[Lemma~16.3.1]{wilde2017} to construct the decoder. Let $\rho_B^{\otimes n} = \rho_{B^n}$ denote the expected output state, where $\rho_B = \sum_{\tilde u} p_{\tilde U}(\tilde u) \rho_B^{\tilde u}$.
For each message $m  \in \{1,\ldots,2^{nR}\}$, the codeword state is $\rho_B^{(m)} = \rho_B^{\tilde u^n(m)}$.
We define:
\begin{itemize}
    \item Code subspace projector $\Pi^\delta$: the projector onto the $\delta$-typical subspace $T_\delta^{(n)}(\rho_B)$ of $\rho_B^{\otimes n}$, satisfying $\Tr\{\Pi^\delta\} \leq 2^{n(H(B)_\rho + \delta)}$.
    \item Codeword subspace projectors $\{\Pi^\delta_m\}_{m \in \{1,\ldots,2^{nR}\}}$: for each $m$, the projector onto the conditionally typical subspace $T_\delta^{(n)}(\rho_{B^n}^{(m)})$ of $\rho_{B^n}^{(m)}$ given $\tilde{u}^n(m)$, satisfying $\Tr\{\Pi^\delta_m\} \leq 2^{n(H(B|\tilde{U})_\rho + \delta)}$.
\end{itemize}
By standard properties of typical subspaces, the Packing Lemma conditions are satisfied with high probability:
\begin{align}
    \Tr\{\Pi^\delta \rho_{B^n}^{(m)}\} &\geq 1 - \varepsilon_n, \\
    \Tr\{\Pi^\delta_m \rho_{B^n}^{(m)}\} &\geq 1 - \varepsilon_n, \\
    \Tr\{\Pi^\delta_m\} &\leq d \coloneqq 2^{n(H(B|\tilde{U})_\rho + \delta)}, \\
    \Pi^\delta \rho_B^{\otimes n} \Pi^\delta &\leq 2^{-n(H(B)_\rho - \delta)} \Pi^\delta.
\end{align}
Setting $D = 2^{n(H(B)_\rho - \delta)}$, we have $2^{nR} \cdot \frac{d}{D} = 2^{n(R + H(B|\tilde{U})_\rho - H(B)_\rho + 2\delta)} = 2^{-n(I(\tilde{U};B)_\rho - R - 2\delta)}$.
Thus,
by the Quantum Packing Lemma, there exists a POVM $\{D_m\}_{m   \in \{1,\ldots,2^{nR}\}}$ such that the average probability of correct decoding satisfies:
\begin{align}
    \mbE_{\mbfC}\left\{\frac{1}{2^{nR}} \sum_{m \in \{1,\ldots,2^{nR}\}} \Tr\{D_m \rho_{B^n}^{(\tilde{u}^n(m))}\}\right\} \geq 1 - 2(\varepsilon_n + 2\sqrt{\varepsilon_n}) - 4 \cdot 2^{-n(I(\tilde{U};B)_\rho - R - 2\delta)},
\end{align}
where $\varepsilon_n \to 0$ as $n \to \infty$.
Hence, the error probability vanishes as $n \rightarrow \infty$, for any $R < I(\tilde{U}; B)_\rho-2\delta $.

The regularized characterization follows by the same arguments as in 
Section~\ref{section:proof of main theorem}. The details are omitted. 
\qed

%------------------------------
% Summary
%------------------------------
\section{Summary and Discussion}
\label{section:summary}
%  We study communication over a quantum action-dependent channel in which the encoder first performs an 
%  action $G$ that influences the channel environment $S$, and subsequently encodes a message into a transmission $A$ sent to the receiver. 
%  This two-stage interaction models practical scenarios such as memories with defects and magnetic recording with rewriting 
%  \cite{kittichokechai2012multi}. In the quantum setting, measurements performed on the 
%  transmission system may induce state collapse in the channel environment, introducing additional 
%  coupling between encoding and channel behavior. 
%  We derive achievable rates for both causal and non-causal side information at the encoder. 
%  As a case study, we analyze a memory model with depolarization and selective rewriting, 
%  illustrating the impact of action-dependent control on communication performance.    

\subsection{Summary}
We study communication over a quantum action-dependent channel, in which the encoder first performs an action $G$,
that influences the channel environment $S$, and subsequently encodes a message into a transmission $A$ sent to the receiver.   
This two-stage interaction models practical scenarios such as memories with defects and magnetic recording with rewriting 
\cite{kittichokechai2012multi}. In the quantum setting, measurements performed on the 
transmission system may induce state collapse in the channel environment, introducing additional 
coupling between encoding and channel behavior.

Our main results are achievable rates for the quantum action-dependent channel under two CSI availability 
assumptions at the encoder. In Theorem~\ref{Theorem:non-causal}, we establish that the rate $\rqad_{\text{n-c}}(\mcN\circ T) = \max [ I(VU;B)_{\rho} - I(V;S|U)_{\rho}]$ is achievable with non-causal CSI, 
where the optimization is over the action ensemble $\{\ket{\sigma_G^u}\}$, a classical auxiliary pair $(V,U)$, and joint input states $\rho_{SA}^{v,u}$ subject
to a partial trace constraint.
For this non-causal setting, we further establish a matching converse, showing that the regularized version of this rate equals the capacity.
In Theorem~\ref{Theorem:causal}, we establish that the rate $\rqad_{\text{caus}}(\mcN\circ T) = \max\, I(U;B)_{\rho}$ is achievable with causal CSI, where Alice applies a quantum Shannon strategy $\{\mathcal{F}^u_{S_0\to A}\}$,
a family of encoding maps each adapted to the action index $u$. Unlike in the non-causal case, no penalty term appears in the causal rate formula, since Alice's transmission strategy is chosen 
independently of the specific induced side-information sequence.
Both results expand previously established frameworks. In the non-causal setting, taking a degenerate action system
% ($|\mathcal{H}_G|=1$) 
removes Alice's influence over the environment state, 
and our rate reduces to the one in  \cite{Dupuis2009, Dupuis2010t, pereg2019entanglement, Pereg2022, pereg2022classical, pereg2021quantum, Anshu2020, zivarifard2024covert}, i.e., the quantum environment-dependent channel without action dependence, 
$\mathsf{R}_{\text{n-c}}(\mathcal{N}\circ\sigma)$ as in \eqref{eq:R_0}. In the causal setting, the achievable rate equals the Holevo information $\chi(\mathcal{M})$ of a 
virtual channel $\mathcal{M}:u\mapsto \mathcal{N}_{SA\to B}(\rho_{SA}^u)$, paralleling Shannon's classical observation for causal CSI \cite{Shannon1958}.

To illustrate these results, we analyze  qubit memory  with depolarization noise and selective rewrite. 
The action map is a %isometry was identified with the 
Stinespring dilation of the 
depolarizing channel, where the side-information system $S_0$ stores a flag indicating which Pauli error occurred. 
Without CSI, the memory capacity is the same as the Holevo information of the depolarizing channel. 
With causal CSI, Alice can detect bit flips before transmission, improving the achievable rate. 
In the non-causal setting, Alice observes the full error pattern before encoding, enabling a selective-rewrite strategy that uses the rewrite pattern itself as an additional signaling 
degree of freedom, yielding  further rate improvement.

\subsection{Open Problems}
% Several directions remain open. First, establishing matching upper bounds and characterizing the capacity of the quantum action-dependent channel is an open problem.
% In the classical setting, Weissman's framework yields a single-letter capacity formula \cite{weissman2010}. The quantum analog is complicated by the general non-additivity
% of the Holevo information \cite{Hastings2009}, and the capacity of a general quantum action-dependent channel may require regularization, as in \eqref{eq:regularized_capacity}.
% Identifying channel classes for which a single-letter characterization holds, such as entanglement-breaking or classical-quantum channels, is of particular interest.
Several directions remain open. First, A single-letter capacity formula for a general quantum channel is a major open problem in quantum Shannon theory.  % while our matching converse characterizes the non-causal capacity in the regularized (multi-letter) sense, obtaining a single-letter characterization remains an open problem.
In the classical setting, Weissman derived a single-letter capacity formula \cite{weissman2010}. Whereas, our capacity characterization has a regualarized multi-letter form (see \eqref{eq:regularized_capacity}).
In the quantum setting, however, the Holevo information is super-additive and does not provide a single-letter characterization of the capacity \cite{Hastings2009}, even in the basic setting of a point-to-point quantum channel,
either with or without  environment dependence \cite{Anshu2020} (see discussion in \cite[Sec. III]{Pereg2022}).
It therefore follows that our capacity formula
% The quantum analog is complicated by the general non-additivity
% of the Holevo information \cite{Hastings2009}, which is why the capacity of a general quantum action-dependent channel 
requires regularization.
%
% In fact, a single-letter converse is open even for the environment-dependent channel without action dependence (and without entanglement assistance) \cite{Anshu2020, Pereg2022}, 
% which our model recovers as a special case. 
%
% A matching converse for the causal setting is likewise left for future work.
%
Identifying channel classes for which a single-letter characterization holds, such as entanglement-breaking or classical-quantum channels, is of particular interest.

% Second, the present model assumes the action channel admits an isometric extension $T_{G\to SS_0}$, so that the rates are achievable with pure input states $\ket{\sigma_G^u}$.
% A natural generalization replaces this with a general CPTP action map $\mcT_{G\to SS_0}$, yielding a mixed state $\sigma_{SS_0}^u$; in this case, the structure of the achievable rate and the interplay 
% between the action-induced correlations and the side information remain to be characterized. 

% Third, 
In this work, we address the transmission of classical information via a quantum action-dependent channel. %Extending the framework to 
Studying quantum information transmission and characterizing the quantum capacity %of the action-dependent channel, 
is a natural future direction. 
% Third, extensions to multi-user scenarios are natural extensions. In a multiple-access channel with action dependence, multiple transmitters may each take actions that jointly influence the channel environment, 
% raising new questions about the coordination of actions and the tradeoff between individual and cooperative strategies. 
% In a broadcast channel \cite{steinberg2012degraded}, the action enables the transmitter to tailor the environment state to the different receivers simultaneously. 
% Fourth, incorporating an action cost constraint \cite{kittichokechai2015coding} and studying the probing capacity \cite{asnani2011probing}, 
% which quantifies how much information about the induced environment can be learned through the action, are further directions. 
Finally, in the environment-dependent model, the additional resource of entanglement assistance between Alice and Bob increases the channel's capacity \cite{Dupuis2009,pereg2019entanglement}.
Moreover, recent work by Yao and Jafar \cite{yao2026quantum} has shown that entanglement assistance can strictly improve the capacity of classical channels with causal state information. 
Exploring whether entanglement between Alice and Bob provides a similar benefit in the quantum action-dependent setting is an intriguing direction.

\subsection{Practical Applications}

Quantum action dependence naturally arises in a broad range of physically relevant settings where
controlled interventions influence the effective communication environment.
Such scenarios are expected to become increasingly important in adaptive quantum technologies,
where learning, control, and information transmission occur simultaneously.

One important example is \emph{joint communication and sensing} \cite{Wang2022ITW} with adaptive probing.
In these systems, the transmitter does not merely send information through a fixed physical medium,
but actively interacts with the environment in order to learn and manipulate the channel while communication takes place.
The probing signals used for sensing inevitably affect the underlying quantum system,
thereby modifying the effective channel experienced by subsequent transmissions.
As a result, the communication channel itself becomes action dependent.

This operational structure fits naturally within our framework.
The action $G$ represents the probing strategy selected by the transmitter,
while the environment state $S$ models the channel response generated by the probing interaction.
The acquired side information then enables the transmitter to adapt its encoding strategy to the learned channel conditions.
From this perspective, adaptive sensing and communication are fundamentally coupled processes:
the transmitter must balance the resources devoted to probing the environment against those devoted to reliable information transmission.

This viewpoint differs fundamentally from conventional sensing formulations \cite{liu2024quantum},
where measurement backaction and probing-induced disturbances are typically regarded as unavoidable impairments.
In contrast, the action-dependent framework treats these interactions as a controllable resource.
The probing actions are deliberately designed not only to estimate the channel,
but also to shape the effective communication environment prior to transmission.
This reveals a fundamental tradeoff between channel learning, environment manipulation,
and communication performance.
More broadly, it suggests that adaptive probing can serve simultaneously as a sensing mechanism,
a channel-control mechanism, and an information-bearing signal.

Another important application is \emph{quantum error mitigation} (QEM)
\cite{endo2018practical,cai2023quantum},
which has emerged as a central approach for improving the performance of noisy intermediate-scale quantum (NISQ) devices.
Because fully fault-tolerant quantum computation remains technologically demanding,
current quantum processors operate in regimes where noise and decoherence significantly degrade computational accuracy.
Quantum error mitigation seeks to suppress these effects without requiring the substantial overheads associated with full quantum error correction.

In general, QEM methods operate by characterizing, learning, or reshaping the effective noise process affecting a quantum device,
and then exploiting the acquired information to compensate for the resulting errors.
Representative techniques include noise extrapolation, probabilistic error cancellation,
virtual distillation, and adaptive noise learning \cite{endo2018practical,cai2023quantum}.
A common feature underlying these methods is the use of controlled interventions
to probe and partially control the effective noise channel before the final information-processing task is performed.

This structure aligns naturally with the action-dependent framework studied in this work.
The noise-characterization protocol plays the role of the action stage,
where the chosen action influences the effective environment experienced by the subsequent transmission or computation.
The resulting environment state $S$ represents the learned or induced noise realization,
while the acquired side information enables noise-aware encoding and compensation strategies.
From this perspective, QEM may be viewed as an adaptive communication protocol in which
probing actions are used not only to estimate noise,
but also to actively shape the effective channel prior to information transmission.

More broadly, action-dependent quantum systems are expected to play an increasingly important role
in the development of adaptive quantum communication and computation architectures,
where learning, control, and information transmission are fundamentally intertwined.

\begin{appendices}

%--------------------------------------
% proof of lemma 1
%--------------------------------------
\section{Subcodebook Average}
\label{appendix:proof_of_lem1}

We begin our one-shot analysis with the subcodebook average result below (see proof outline in Subsection~\ref{Subsection:Outline}).
\begin{lemma} [see {\cite[Lemma 7]{Anshu2020}}] 
\label{lem:1}
    Let $\rho_{VUS} = \Tr_A\left[\rho_{VUSA}\right]$ be a classical-quantum state, and $\alpha \in (0,\frac{1}{2})$. Furthermore, let $\mcC_m = \{v(m,1), \ldots v({m,2^{R_S}}), u(m)\}$ be a collection of random variables such that the sequence $v(m,1),\ldots,v(m,2^{R_S})$ is conditionally i.i.d.
    $\sim p_{V|U}(\cdot| u(m)) $, for every given $u(m)$.
    % we have the following probability distributions:
    % % for every $\ell \in \{1,\ldots, 2^{R_S}\}
    % %= \msL
    % % ,\;
    % $v(m,\ell) \sim p_{V|U}(\cdot| u(m)), \; \left(v(m,\ell), v(m, \ell') \right) \sim p_{V|U} (\cdot| u(m))\cdot p_{V|U}(\cdot| u(m))$, 
    % for all $\ell \in \{1,\ldots, 2^{R_S}\}$.
    % % Let $C = \{U(1),\ldots, U(2^{R_S}), V(m,1), \dots V(m,R_s)\}$ be a set of random variables generated i.i.d., where $U \sim p_U$ and $V(\cdot|\ell) \sim P_{V|}(\cdot|U(\ell))$. 
    We consider the following state:
    \begin{align}
        \tau_{S|\mcC_m} \triangleq \frac{1}{2^{R_S}} \sum_{\ell = 1}^{2^{R_S}} \rho_{S}^{v(m,\ell),u(m)}. \label{eq:tau_S_Cm}
    \end{align}
    % and 
    % \begin{align*}
    %     \hat \tau_{S} \triangleq \frac{1}{2^{R_S}} \sum_{\ell \in \msL} \rho_{S}^{u(\ell)}.
    % \end{align*}
    %
    Then, %for $\alpha \in (0,\frac{1}{2})$,
    there exists a constant $\nu_2 \geq 0$ such that:
    \begin{align}
        \mathbb{E}_{\mbfC}\left\{\tilde{D}_{1+\alpha}(\tau_{S|\mcC_m} \| \sigma_S^{u(m)})\right\} \leq 
        \frac{1}{\alpha \ln 2}  \frac{\nu_2^{\alpha}}{2^{\alpha R_S}}  2^{\alpha \tilde{D}_{1+\alpha}(\rho_{VUS} \| \rho_{V-U-S})} 
    \end{align}
    for all $m\in \{1,\ldots, 2^R\}$,
    where %$\rho_{S}^u \triangleq \sum_v p_{V|U}(v|u) \rho_{S|u,v}$, and  
    $\rho_{V-U-S} = \sum_{u} p_{U}(u) \rho_{V}^u \otimes \ketbra{u}{u} \otimes  \sigma_{S}^u$, and $\nu_2$ is the maximum number of distinct eigenvalues of the states $\{\sigma_{S|u}\}_u$.
\end{lemma}

The lemma can be obtained as a consequence of \cite[Lemma 7]{Anshu2020}. For completeness, we provide the full proof below. 
 % \begin{lemma} \label{lem:2}
 %      For every $\alpha \in \left(0,\frac{1}{2}\right)$, and $R, R_S > 0$ we have:
 %      \begin{align}
 %        \Tr[\left( \mbone - \Pi_{VUB}\right) \rho_{VUB}] 
 %        \leq &\nu_{1}^{\alpha} 2^{\alpha \left[R + R_S - \tilde{D}_{1-\alpha}\left(\rho_{VUB} \| \rho_{VU} \otimes \rho_B \right)\right]}, 
 %        \label{eq:first_term_lemma_2}  \\ 
 %         2^{R + R_S} \Tr \left[\Pi_{VUB}\left(\rho_{VU} \otimes \rho_B \right)  \right] 
 %        \leq & \nu_{1}^{\alpha} 2^{\alpha \left[R + R_S- \tilde{D}_{1-\alpha}\left(\rho_{VUB} \| \rho_{VU} \otimes \rho_B \right)\right]}.  \label{eq:second_term_lemma_2}
 %      \end{align}
 %    \end{lemma}
\begin{proof}
For $\alpha \in (0,1]$ and due to concavity of $\log(\cdot)$, we obtain:
\begin{align}
    \mbE_{\mbfC}\left\{\tilde{D}_{1+\alpha}(\tau_{S|\mcC_m} \| \sigma_S^{u(m)})\right\}
    &= \frac{1}{\alpha}\mbE_{\mbfC}\left\{\alpha \tilde{D}_{1+\alpha}(\tau_{S|\mcC_m} \| \sigma_S^{u(m)})\right\}
    \nonumber\\
    &= \frac{1}{\alpha}\mbE_{\mbfC}\left\{2^{\log_2\left(\alpha \tilde{D}_{1+\alpha}(\tau_{S|\mcC_m} \| \sigma_S^{u(m)})\right)}\right\}
    \nonumber \\
    &\leq \frac{1}{\alpha} \log_2 \left(\mbE_{\mbfC} \left\{ 2^{\alpha \tilde{D}_{1+\alpha}(\tau_{S|\mcC_m} \| \sigma_S^{u(m)})} \right\} \right).
    \label{eq:lem1_proof_part1}
\end{align}

Then, by substituting the definition of $\tilde{D}_{1+\alpha}(\cdot\|\cdot)$: \footnote{To avoid confusion, from now on, we are using the notation $\sigma_{S|u(m)}$ and $\rho_{S|v(m,\ell),u(m)}$ to represent the states $\sigma_S^{u(m)}$ and $\rho_S^{v(m,\ell),u(m)}$, respectively.} 
\begin{align}
    \mbE_{\mbfC} &\left\{ 2^{\alpha \tilde{D}_{1+\alpha}(\tau_{S|\mcC_m} \| \sigma_{S|u(m)})}\right\}
    \nonumber \\
    &= \mbE_{\mbfC} \left\{ \Tr\left[ \left({\sigma_{S|u(m)}}^{\frac{-\alpha}{2(1+\alpha)}} \tau_{S|\mcC_m} {\sigma_{S|u(m)}}^{\frac{-\alpha}{2(1+\alpha)}} \right)^{1+\alpha} \right] \right\}
    \nonumber \\
    &= \mbE_{\mbfC} \left\{ \Tr\left[ \left({\sigma_{S|u(m)}}^{\frac{-\alpha}{2(1+\alpha)}} \left(\frac{1}{2^{R_S}} \sum_{\ell = 1}^{2^{R_S}} \rho_{S|v(m,\ell),u(m)} \right) {\sigma_{S|u(m)}}^{\frac{-\alpha}{2(1+\alpha)}} \right)^{1+\alpha} \right] \right\}
    \nonumber \\
    &= \mbE_{\mbfC} \left\{ \Tr \left[\left({\sigma_{S|u(m)}}^{\frac{-\alpha}{2(1+\alpha)}} \left(\frac{1}{2^{R_S}} \sum_{\ell = 1}^{2^{R_S}} \rho_{S|v(m,\ell),u(m)} \right) {\sigma_{S|u(m)}}^{\frac{-\alpha}{2(1+\alpha)}} \right) \right.\right.
    \nonumber \\
    & \quad\times\left. \left.\left({\sigma_{S|u(m)}}^{\frac{-\alpha}{2(1+\alpha)}} \left(\frac{1}{2^{R_S}} \sum_{\ell' = 1}^{2^{R_S}} \rho_{S|v(m,\ell'),u(m)} \right) {\sigma_{S|u(m)}}^{\frac{-\alpha}{2(1+\alpha)}} \right)^{\alpha} \right] \right\}
    % \nonumber \\
    %  &= \frac{1}{2^{R_S}} \sum_{\ell \in \msL} \mbE_{\mbfC} \Tr \left[\left({\rho_{S|u(m)}}^{\frac{-\alpha}{2(1+\alpha)}} \left( \rho_{S|v(m,\ell),u(m)} \right) {\rho_{S|u(m)}}^{\frac{-\alpha}{2(1+\alpha)}} \right) \right.
    % \nonumber \\
    % & \quad\times \left.\left({\rho_{S|u(m)}}^{\frac{-\alpha}{2(1+\alpha)}} \left(\frac{1}{2^{R_S}} \sum_{\ell' \in \msL} \rho_{S|v(m,\ell'),u(m)} \right) {\rho_{S|u(m)}}^{\frac{-\alpha}{2(1+\alpha)}} \right)^{\alpha} \right]
    \nonumber \\
     &= \frac{1}{2^{R_S}} \sum_{\ell = 1}^{2^{R_S}} \mbE_{\mbfC} \left\{ \Tr \left[\left({\sigma_{S|u(m)}}^{\frac{-\alpha}{2(1+\alpha)}} \left( \rho_{S|v(m,\ell),u(m)} \right) {\sigma_{S|u(m)}}^{\frac{-\alpha}{2(1+\alpha)}} \right) \right. \right.
    \nonumber \\
    & \quad\times \left.\left.\left({\sigma_{S|u(m)}}^{\frac{-\alpha}{2(1+\alpha)}} \frac{1}{2^{R_S}}\left( \rho_{S|v(m,\ell),u(m)} + \sum_{\ell' \neq \ell} \rho_{S|v(m,\ell'),u(m)} \right) {\sigma_{S|u(m)}}^{\frac{-\alpha}{2(1+\alpha)}} \right)^{\alpha} \right] \right\}
    \nonumber \\
     &\leq \frac{1}{2^{R_S}} \sum_{\ell = 1}^{2^{R_S}} \mbE_{\mbfC} \left\{ \Tr \left[\left({\sigma_{S|u(m)}}^{\frac{-\alpha}{2(1+\alpha)}} \left( \rho_{S|v(m,\ell),u(m)} \right) {\sigma_{S|u(m)}}^{\frac{-\alpha}{2(1+\alpha)}} \right) \right.\right.
    \nonumber \\
    & \quad\times \left.\left.\left({\sigma_{S|u(m)}}^{\frac{-\alpha}{2(1+\alpha)}} \frac{1}{2^{R_S}}\left( \rho_{S|v(m,\ell),u(m)} + \sum_{\ell' = 1}^{2^{R_S}} \rho_{S|v(m,\ell'),u(m)} \right) {\sigma_{S|u(m)}}^{\frac{-\alpha}{2(1+\alpha)}} \right)^{\alpha} \right]\right\},
\end{align}
where the second equality  is obtained by substituting $\tau_{S|\mcC_m}$ as in \eqref{eq:tau_S_Cm}.
Recall that the sequence $v(m,\ell')$, $\ell'\in\{1,\ldots,2^{R_S}\}$, is conditionally i.i.d.
    $\sim p_{V|U}(\cdot| u(m)) $, for every given $u(m)$.
Therefore,
\begin{align}
    \mbE_{\mbfC} &\left\{ 2^{\alpha \tilde{D}_{1+\alpha}(\tau_{S|\mcC_m} \| \sigma_{S|u(m)})}\right\}
    \nonumber \\
    &\leq \frac{1}{2^{R_S}} \sum_{\ell = 1}^{2^{R_S}} \mbE_{\mbfC} \left\{ \Tr \left[\left({\sigma_{S|u(m)}}^{\frac{-\alpha}{2(1+\alpha)}} \left( \rho_{S|v(m,\ell),u(m)} \right) {\sigma_{S|u(m)}}^{\frac{-\alpha}{2(1+\alpha)}} \right) \right.\right.
    \nonumber \\
    & \quad\times \left.\left.\left({\sigma_{S|u(m)}}^{\frac{-\alpha}{2(1+\alpha)}} \frac{1}{2^{R_S}}\left( \rho_{S|v(m,\ell),u(m)} 
    + 2^{R_S} \mbE_{V|u(m)} \left[\rho_{S|V,u(m)} \right] \right) {\sigma_{S|u(m)}}^{\frac{-\alpha}{2(1+\alpha)}} \right)^{\alpha} \right]\right\}
\end{align}
We note that the inner expectation term satisfies
$\mbE_{V|u(m)} \left[\rho_{S|V,u(m)} \right]= \sigma_{S|u(m)}$.
Taking the expectation over $u(m)$ and $v(m,\ell)$ as well, we obtain
\begin{align}
    \mbE_{\mbfC} &\left\{ 2^{\alpha \tilde{D}_{1+\alpha}(\tau_{S|\mcC_m} \| \sigma_{S|u(m)})}\right\}
    \nonumber \\
    &\leq
    %  &\stackrel{(c)}{=} \frac{1}{2^{R_S}} \sum_{\ell \in \msL} \mbE_{\mbfC} \Tr \left[\left({\rho_{S|u(m)}}^{\frac{-\alpha}{2(1+\alpha)}} \left( \rho_{S|v(m,\ell),u(m)} \right) {\rho_{S|u(m)}}^{\frac{-\alpha}{2(1+\alpha)}} \right) \right.
    % \nonumber \\
    % & \quad\times \left.\left({\rho_{S|u(m)}}^{\frac{-\alpha}{2(1+\alpha)}} \frac{1}{2^{R_S}}\left( \rho_{S|v(m,\ell),u(m)} 
    % + 2^{R_S} \rho_{S|u(m)}  \right) {\rho_{S|u(m)}}^{\frac{-\alpha}{2(1+\alpha)}} \right)^{\alpha} \right]
     % \nonumber \\
    % \frac{1}{2^{R_S}} \sum_{\ell \in \msL} \mbE_{\mbfC} \Tr \left[\left({\rho_{S|u(m)}}^{\frac{-\alpha}{2(1+\alpha)}} \left( \rho_{S|v(m,\ell),u(m)} \right) {\rho_{S|u(m)}}^{\frac{-\alpha}{2(1+\alpha)}} \right) \right.
    % \nonumber \\
    % & \quad\times \left.\left({\rho_{S|u(m)}}^{\frac{-\alpha}{2(1+\alpha)}} \frac{1}{2^{R_S}}\left( \rho_{S|v(m,\ell),u(m)} 
    % + 2^{R_S} \rho_{S|u(m)}  \right) {\rho_{S|u(m)}}^{\frac{-\alpha}{2(1+\alpha)}} \right)^{\alpha} \right]
    %  \nonumber \\
    \mbE_{V, U} \left\{ \Tr \left[\left({\sigma_{S|U}}^{\frac{-\alpha}{2(1+\alpha)}} \left( \rho_{S|V,U} \right) {\sigma_{S|U}}^{\frac{-\alpha}{2(1+\alpha)}} \right) 
    \left({\sigma_{S|U}}^{\frac{-\alpha}{2(1+\alpha)}} \frac{1}{2^{R_S}}\left( \rho_{S|V,U} 
    + 2^{R_S} \sigma_{S|U}  \right) {\sigma_{S|U}}^{\frac{-\alpha}{2(1+\alpha)}} \right)^{\alpha} \right]\right\}.
\end{align}

Next, we introduce the pinching map $\mcE_2$ with respect to $\sigma_{S|u}$. Let $\nu_2$ denote the maximum number of distinct eigenvalues of $\sigma_{S|u}$. Then, by the pinching inequality in \eqref{eq:pinching_inequality}, 
\begin{align}
    \mbE_{\mbfC} &\left\{ 2^{\alpha \tilde{D}_{1+\alpha}(\tau_{S|\mcC_m} \| \sigma_{S|u(m)})}\right\}
    \nonumber \\
    &
{\leq} \mbE_{V, U} \left\{ \Tr \left[\left({\sigma_{S|U}}^{\frac{-\alpha}{2(1+\alpha)}} \left( \rho_{S|V,U} \right) {\sigma_{S|U}}^{\frac{-\alpha}{2(1+\alpha)}} \right) \left({\sigma_{S|U}}^{\frac{-\alpha}{2(1+\alpha)}} \frac{1}{2^{R_S}}\left( \nu_2 \mcE_2(\rho_{S|V,U}) 
    + 2^{R_S} \sigma_{S|U}  \right) {\sigma_{S|U}}^{\frac{-\alpha}{2(1+\alpha)}} \right)^{\alpha} \right] \right\}
    \nonumber \\
     &{\leq} \mbE_{V, U} \left\{ \Tr \left[\left({\sigma_{S|U}}^{\frac{-\alpha}{2(1+\alpha)}} \left( \rho_{S|V,U} \right) {\sigma_{S|U}}^{\frac{-\alpha}{2(1+\alpha)}} \right) 
     \left({\sigma_{S|U}}^{\frac{-\alpha^2}{2(1+\alpha)}} \frac{1}{2^{\alpha R_S}}\left( \nu_2^{\alpha} \mcE_2^{\alpha}(\rho_{S|V,U}) 
    + 2^{\alpha R_S} \sigma_{S|U}^{\alpha}  \right) {\sigma_{S|U}}^{\frac{-\alpha^2}{2(1+\alpha)}} \right) \right] \right\}
\end{align}
where the second inequality holds since the operators $ \mbE_2(\rho_{S|V,U})$ and $\sigma_{S|U}$ commute.
By trace cyclicity, we can write the last expression as 
\begin{align}
    & \mbE_{V, U} \left\{\Tr \left[{\sigma_{S|U}}^{\alpha -\alpha}\rho_{S|V,U}\right] \right\}
     + \mbE_{V, U} \left\{\Tr \left[{\sigma_{S|U}}^{\frac{-\alpha}{2(1+\alpha)}} \rho_{S|V,U} \; {\sigma_{S|U}}^{\frac{-(\alpha +\alpha^2)}{2(1+\alpha)}} \frac{\nu_2^{\alpha}}{2^{\alpha R_S}}\left(  \mcE_2^{\alpha}(\rho_{S|V,U})  \right) {\sigma_{S|U}}^{\frac{-\alpha^2}{2(1+\alpha)}} \right] \right\}
    \nonumber \\
   &= 1 + \frac{\nu_2^{\alpha}}{2^{\alpha R_S}}\mbE_{V, U} \left\{\Tr \left[ \rho_{S|V,U} \;\mcE_2^{\alpha}(\rho_{S|V,U}) \; {\sigma_{S|U}}^{-\alpha} \right]\right\}
   \nonumber \\
   &\stackrel{(a)}{=} 1 + \frac{\nu_2^{\alpha}}{2^{\alpha R_S}}\mbE_{V, U} \left\{\ \Tr \left[ \rho_{S|V,U} \;\mcE_2\left(\mcE_2^{\alpha}(\rho_{S|V,U}) \; {\sigma_{S|U}}^{-\alpha}\right) \right]\right\}
   \nonumber \\
   &\stackrel{(b)}{=} 1 + \frac{\nu_2^{\alpha}}{2^{\alpha R_S}}\mbE_{V, U} \left\{\ \Tr \left[\mcE_2^{1+\alpha}(\rho_{S|V,U}) \; {\sigma_{S|U}}^{-\alpha} \right]\right\}
   \nonumber \\
   &\stackrel{(c)}{=} 1 + \frac{\nu_2^{\alpha}}{2^{\alpha R_S}}\mbE_{V, U}  \left\{ 2^{\alpha \tilde D_{1+\alpha}\left( \mcE_2(\rho_{S|V,U}) \| \mcE_2(\sigma_{S|U}) \right)} \right\}
    \nonumber \\
   &\stackrel{(d)}{\leq} 1 + \frac{\nu_2^{\alpha}}{2^{\alpha R_S}}\mbE_{V, U}  \left\{ 2^{\alpha \tilde D_{1+\alpha}\left( \rho_{S|V,U} \| \sigma_{S|U} \right)} \right\}
   \nonumber \\
   &\stackrel{(e)}{=} 1 + \frac{\nu_2^{\alpha}}{2^{\alpha R_S}} 2^{\alpha \tilde D_{1+\alpha}\left( \rho_{VUS} \| \rho_{V-U-S} \right)},
   \label{eq:lem1_proof_part2} 
   \end{align}

where
\begin{itemize}
    \item[$(a)$] holds since the state is invariant to another application of the pinching map,
    \item[$(b)$] since
    \begin{align*}
        \Tr[\mcE_{\sigma}(\rho) \omega] &= \Tr\left[\sum_{\lambda} \Pi_{\lambda} \rho \Pi_{\lambda} \omega \right]
        = \sum_{\lambda} \Tr\left[\Pi_{\lambda} \rho \Pi_{\lambda} \omega \right]
        = \sum_{\lambda} \Tr\left[\rho \Pi_{\lambda} \omega \Pi_{\lambda}\right]
        = \Tr\left[\rho \mcE_{\sigma}(\omega)\right],
    \end{align*}
    \item[$(c)$]  is obtained by substituting the definition of $\tilde D_{1+\alpha}(\cdot\|\cdot)$, and since $\mcE_2(\sigma_{S|u}) = \sigma_{S|u}$.
    \item[$(d)$]  is obtained by the data processing inequality of $\tilde D_{1+\alpha}(\cdot\|\cdot)$, for $\alpha \in (-\frac{1}{2},0) \; \cup \; (0, \infty) $  \cite{tomamichel2015quantum,MuellerLennert2013}.
    \item[$(e)$]  is obtained by taking the expectation over $V,U$:
    \begin{align}
        \mbE_{V,U} \left\{ 2^{\alpha \tilde D_{1+\alpha}\left( \rho_{S|V,U} \| \sigma_{S|U} \right)} \right\} 
        &= \mbE_{V,U} \left\{ \Tr\left[ \left({\sigma_{S|u}}^{\frac{-\alpha}{2(1+\alpha)}} \rho_{S|V,U} {\sigma_{S|U}}^{\frac{-\alpha}{2(1+\alpha)}} \right)^{1+\alpha} \right]  \right\}
        \nonumber\\
        &= \Tr\left[ \sum_{v,u} p_{V,U}(v,u) \ketbra{v}{v} \otimes \ketbra{u}{u} \otimes \left({\sigma_{S|u}}^{\frac{-\alpha}{2(1+\alpha)}} \rho_{S|v,u} \; {\sigma_{S|u}}^{\frac{-\alpha}{2(1+\alpha)}} \right)^{1+\alpha} \right]
        \nonumber\\
        &= \Tr\left[ \left({\rho_{V-U-S}}^{\frac{-\alpha}{2(1+\alpha)}} \rho_{VUS} \; {\rho_{V-U-S}}^{\frac{-\alpha}{2(1+\alpha)}} \right)^{1+\alpha} \right]
        \nonumber\\
        &= 2^{\alpha \tilde D_{1+\alpha}\left( \rho_{VUS} \| \rho_{V-U-S} \right)}.
    \end{align}
\end{itemize}

% Therefore, from \eqref{eq:lem1_proof_part1} and \eqref{eq:lem1_proof_part2},
We conclude that for $\alpha \in (0,\frac{1}{2})$:
\begin{align}
    \mbE_{\mbfC}  \left\{\tilde{D}_{1+\alpha}(\tau_{S|\mcC_m} \| \sigma_{S|u(m)})\right\}
    & \leq \frac{1}{\alpha} \log_2 \left( 1 + \frac{\nu_2^{\alpha}}{2^{\alpha R_S}} 2^{\alpha \tilde D_{1+\alpha}\left( \rho_{VUS} \| \rho_{V-U-S} \right)} \right) 
    \nonumber \\
    & \leq \frac{1}{\alpha \ln 2}  \frac{\nu_2^{\alpha}}{2^{\alpha R_S}}  2^{\alpha \tilde{D}_{1+\alpha}(\rho_{VUS} \| \rho_{V-U-S})},
    \label{eq:lem1_final}
\end{align}
where the last line follows from $\log_2(1+x) \leq \frac{x}{\ln 2}$ for $x\in (-1,\infty)$.
This completes the proof of the subcodebook property in Lemma~\ref{lem:1}.
\end{proof}

%----------------------------------------------
% Example
%----------------------------------------------
% \section{Example}
% Consider the example of where the Alice can influence the error probability of a depolarizing channel. 
% he action encoder maps a message $m$ to a classical action $g \in \mathcal{G}$, which represents a chosen operational mode (e.g., power level).

% The first channel, $\mathcal{T}_{G \to SS_0}$, is a classical channel that models the noisy outcome of the action. 
% It takes action $g$ and produces a classical state $s \in \{0, 1, 2, 3\}$ according to a conditional probability 
% distribution $P_{S|G}(s|g)$. This state $s$ is provided as side information $S_0=s$ to the main encoder and as 
% the environment state $S=s$ to the main channel. We assume $P_{S|G}(s=0|g) = 1-p$ and $P_{S|G}(s \in \{1,2,3\}|g) = p/3$, where $p$ is an error probability dependent on the action $g$.

% The main channel $\mathcal{N}_{SA \to B}$ is a Pauli channel that applies the Pauli operator $\sigma_s$ corresponding to the received state $s$:
% \begin{align*}
% \mathcal{N}_{SA \to B}(\rho_A) = \sigma_s \rho_A \sigma_s^\dagger,
% \end{align*}
% where $\sigma_0 = \mathbb{I}, \sigma_1 = X, \sigma_2 = Y, \sigma_3 = Z$. The overall channel $\mathcal{N}_{G \to B}$ is then given by:
% \begin{align*}
% \mathcal{N}_{G \to B}(\rho_A) = \sum_{s} P_{S|G}(s|g) \sigma_s \rho_A \sigma_s^\dagger.
% \end{align*}

%-----------------------------------
%proof of lemma 3
%-----------------------------------
\section{Projector Properties}%{Proof of Lemma \ref{lem:2}} 
\label{appendix:proof_of_lem2}

Our one-shot error analysis makes use of the projector properties below (see proof outline in Subsection~\ref{Subsection:Outline}). 
 \begin{lemma} \label{lem:2}
      For every $\alpha \in \left(0,\frac{1}{2}\right)$, and $R, R_S > 0$ we have:
      \begin{align}
        \Tr[\left( \mbone - \Pi_{VUB}\right) \rho_{VUB}] 
        \leq &\nu_{1}^{\alpha} 2^{\alpha \left[R + R_S - \tilde{D}_{1-\alpha}\left(\rho_{VUB} \| \rho_{VU} \otimes \rho_B \right)\right]}, 
        \label{eq:first_term_lemma_2}  \\ 
         2^{R + R_S} \Tr \left[\Pi_{VUB}\left(\rho_{VU} \otimes \rho_B \right)  \right] 
        \leq & \nu_{1}^{\alpha} 2^{\alpha \left[R + R_S- \tilde{D}_{1-\alpha}\left(\rho_{VUB} \| \rho_{VU} \otimes \rho_B \right)\right]}.  \label{eq:second_term_lemma_2}
      \end{align}
    \end{lemma}

\begin{proof}
We start by showing the upper bound in \eqref{eq:first_term_lemma_2}:
\begin{align}
    \Tr[\left( \mbone - \Pi_{VUB}\right) \rho_{VUB}] 
    &\stackrel{(a)}{=}   \Tr[\mcE_{1}\left(\left( \mbone - \Pi_{VUB}\right) \rho_{VUB}\right)]
    \nonumber \\   
    &\stackrel{(b)}{=}  \Tr[\left( \mbone - \Pi_{VUB}\right) \mcE_{1}(\rho_{VUB})]
    \nonumber \\ 
    &=  \Tr[\left( \mbone - \Pi_{VUB}\right) \mcE_{1}(\rho_{VUB})^{1-\alpha} \mcE_{1}(\rho_{VUB})^{\alpha}]
    \nonumber \\ 
    &\stackrel{(c)}{\leq}   2^{\alpha (R + R_S)} \Tr[\left( \mbone - \Pi_{VUB}\right) \mcE_{1}(\rho_{VUB})^{1-\alpha} (\rho_{VU}\otimes \rho_{B})^{\alpha}] 
    \nonumber \\ 
    &{\leq}   2^{\alpha (R + R_S)} \Tr[\mcE_{1}(\rho_{VUB})^{1-\alpha} (\rho_{VU}\otimes \rho_{B})^{\alpha}] 
    % \nonumber \\ 
    % &=   2^{\alpha (R + R_S)} \Tr[(\rho_{VU}\otimes \rho_{B})^{\frac{\alpha}{2}} \mcE_{1}(\rho_{VUB})^{1-\alpha} (\rho_{VU}\otimes \rho_{B})^{\frac{\alpha}{2}}] 
    % \nonumber \\ 
    % &=   2^{\alpha (R + R_S)} \Tr[\left((\rho_{VU}\otimes \rho_{B})^{\frac{\alpha}{2(1-\alpha)}} \mcE_{1}(\rho_{VUB}) (\rho_{VU}\otimes \rho_{B})^{\frac{\alpha}{2(1-\alpha)}}\right)^{1-\alpha}] 
    \nonumber \\ 
    &=    2^{\alpha (R + R_S)} \Tr \left[\left( \left(\rho_{VU} \otimes \rho_B \right)^{\frac{\alpha}{2(1-\alpha)}} \mcE_1 \left(\rho_{VUB}\right)  \left(\rho_{VU} \otimes \rho_B \right)^{\frac{\alpha}{2(1-\alpha)}} \right) \right. 
    \nonumber \\ 
    &\quad \left. \times \left( \left(\rho_{VU} \otimes \rho_B \right)^{\frac{\alpha}{2(1-\alpha)}} \mcE_1 \left(\rho_{VUB}\right)  \left(\rho_{VU} \otimes \rho_B \right)^{\frac{\alpha}{2(1-\alpha)}} \right)^{-\alpha} \right] 
    \nonumber \\ 
    & \stackrel{(d)}{=}    2^{\alpha (R + R_S)} \Tr \left[\left( \left(\rho_{VU} \otimes \rho_B \right)^{\frac{\alpha}{2(1-\alpha)}} \rho_{VUB}  \left(\rho_{VU} \otimes \rho_B \right)^{\frac{\alpha}{2(1-\alpha)}} \right) \right. 
    \nonumber \\ 
    &\quad \left. \times \left( \left(\rho_{VU} \otimes \rho_B \right)^{\frac{\alpha}{2(1-\alpha)}} \mcE_1 \left(\rho_{VUB}\right)  \left(\rho_{VU} \otimes \rho_B \right)^{\frac{\alpha}{2(1-\alpha)}} \right)^{-\alpha} \right] 
    % \nonumber \\ 
    % &\stackrel{(f)}{\leq}   2^{\alpha (R + R_S)} \nu_{1} ^{\alpha} \Tr \left[\left( \left(\rho_{VU} \otimes \rho_B \right)^{\frac{\alpha}{2(1-\alpha)}} \rho_{VUB}  \left(\rho_{VU} \otimes \rho_B \right)^{\frac{\alpha}{2(1-\alpha)}} \right) \right. 
    % \nonumber \\ 
    % &\quad \left. \times \left( \left(\rho_{VU} \otimes \rho_B \right)^{\frac{\alpha}{2(1-\alpha)}} \rho_{VUB}  \left(\rho_{VU} \otimes \rho_B \right)^{\frac{\alpha}{2(1-\alpha)}} \right)^{-\alpha} \right] 
    % \nonumber \\ 
    % &=   2^{\alpha (R + R_S)} \nu_{1} ^{\alpha} \Tr \left[\left( \left(\rho_{VU} \otimes \rho_B \right)^{\frac{\alpha}{2(1-\alpha)}} \rho_{VUB}  \left(\rho_{VU} \otimes \rho_B \right)^{\frac{\alpha}{2(1-\alpha)}} \right)^{1-\alpha} \right] 
    % \nonumber \\ 
    % &=  \nu_{1}^{\alpha} 2^{\alpha \left[R + R_S - \tilde D_{1-\alpha}\left(\rho_{VUB} \| \rho_{VU} \otimes \rho_B \right)\right]}.
    % \label{eq:first_term_lemma_2_proof}
\end{align}
  where
\begin{itemize}
  \item[$(a)$] holds due to the trace-preserving property of pinching, 
  \item[$(b)$] follows from the definition of $\Pi_{VUB}$ in \eqref{eq:Pi_VUB_def}, %The projector

  \item[$(c)$] from the definition of $\Pi_{VUB}$ and 
  the operator monotonicity of the function $f(x) = x^\alpha$ 
  for $ \alpha \in (0,1]$,    
  % \item[$(d)$] holds since $\mbone - \Pi_{VUB} \leq \mbone$ 
  \item[$(d)$]  holds since the pinching map $\mathcal{E}_1$ is with respect to the product state 
  $\rho_{VU}\otimes \rho_B$.
  
  % follows from the fact that for three states $\rho, \sigma, \omega \; \in \mathscr{D}(\mcH)$, where $\omega$ is block diagonal in the  eigen basis of $\sigma$, and $\mcE_{\sigma}$ is the pinching map with respect to $\sigma$. Then for $\beta > 0$:
  % \begin{align*}
  %     \Tr\left[ \sigma^{\beta}\mcE_{\sigma}(\rho) \sigma^{\beta} \omega \right] &= \Tr\left[\sum_{\lambda} \sigma^{\beta} \Pi_{\lambda}  \rho \Pi_{\lambda} \sigma^{\beta} \omega \right]
  %     \\
  %     &= \sum_{\lambda} \Tr\left[\sigma^{\beta} \Pi_{\lambda}  \rho \Pi_{\lambda} \sigma^{\beta} \omega \right]
  %     \\
  %     &= \sum_{\lambda} \Tr\left[\Pi_{\lambda} \sigma^{\beta} \omega \sigma^{\beta} \Pi_{\lambda}  \rho \right]
  %     \\
  %     &= \Tr \left[\mcE_{\sigma}(\sigma^{\beta} \omega \sigma^{\beta}) \rho \right]
  %     \\
  %     &=  \Tr \left[\sigma^{\beta} \omega \sigma^{\beta} \rho \right]
  %     \\
  %     &=  \Tr \left[\sigma^{\beta} \rho \sigma^{\beta} \omega  \right].
  % \end{align*}
  % The equality then follows by considering the case where the states are: $\sigma = \rho_{VU} \otimes \rho_B, \; \rho = \rho_{VUB}, \; \omega = \left( \left(\rho_{VU} \otimes \rho_B \right)^{\frac{\alpha}{2(1-\alpha)}} \mcE_{1}(\rho_{VUB})  \left(\rho_{VU} \otimes \rho_B \right)^{\frac{\alpha}{2(1-\alpha)}} \right)^{-\alpha}$, and $\beta = \frac{\alpha}{2(1-\alpha)}, \; \mcE_{\sigma} = \mcE_1$ which simplifies the expression. 
\end{itemize}
Based on the pinching inequality \eqref{eq:pinching_inequality}, it follows that 
\begin{align}
    \Tr[\left( \mbone - \Pi_{VUB}\right) \rho_{VUB}]    
    &{\leq}   2^{\alpha (R + R_S)} \nu_{1} ^{\alpha} \Tr \left[\left( \left(\rho_{VU} \otimes \rho_B \right)^{\frac{\alpha}{2(1-\alpha)}} \rho_{VUB}  \left(\rho_{VU} \otimes \rho_B \right)^{\frac{\alpha}{2(1-\alpha)}} \right) \right. 
    \nonumber \\ 
    &\quad \left. \times \left( \left(\rho_{VU} \otimes \rho_B \right)^{\frac{\alpha}{2(1-\alpha)}} \rho_{VUB}  \left(\rho_{VU} \otimes \rho_B \right)^{\frac{\alpha}{2(1-\alpha)}} \right)^{-\alpha} \right] 
    \nonumber \\ 
    % &=   2^{\alpha (R + R_S)} \nu_{1} ^{\alpha} \Tr \left[\left( \left(\rho_{VU} \otimes \rho_B \right)^{\frac{\alpha}{2(1-\alpha)}} \rho_{VUB}  \left(\rho_{VU} \otimes \rho_B \right)^{\frac{\alpha}{2(1-\alpha)}} \right)^{1-\alpha} \right] 
    % \nonumber \\ 
    &=  \nu_{1}^{\alpha} 2^{\alpha \left[R + R_S - \tilde D_{1-\alpha}\left(\rho_{VUB} \| \rho_{VU} \otimes \rho_B \right)\right]}.
    \label{eq:first_term_lemma_2_proof}
\end{align}

The derivation for  \eqref{eq:second_term_lemma_2} follows
 similar steps:
  \begin{align}
  2^{R+R_S} \Tr \left[\Pi_{VUB}\left(\rho_{VU} \otimes \rho_B \right)  \right] 
  &= 2^{R + R_S} \Tr \left[\Pi_{VUB} \left(\rho_{VU} \otimes \rho_B \right)^{1-\alpha}  \left(\rho_{VU} \otimes \rho_B \right)^{\alpha}  \right] 
  \nonumber \\ 
  &{\leq}  2^{R+R_S} 2^{-(R+R_S)(1-\alpha)} \Tr \left[\Pi_{VUB} \mcE_1 \left(\rho_{VUB} \right)^{1-\alpha}  \left(\rho_{VU} \otimes \rho_B \right)^{\alpha}  \right] \label{eq:lem_3_proof_part2_a} 
  \nonumber \\ 
  &{\leq}    2^{\alpha (R+R_S)} \Tr \left[ \mcE_1 \left(\rho_{VUB}\right)^{1-\alpha}  \left(\rho_{VU} \otimes \rho_B \right)^{\alpha}  \right] 
  \nonumber \\ 
  % & =    2^{\alpha (R+R_S)} \Tr \left[\left( \left(\rho_{VU} \otimes \rho_B \right)^{\frac{\alpha}{2(1-\alpha)}} \mcE_1 \left(\rho_{VUB}\right)  \left(\rho_{VU} \otimes \rho_B \right)^{\frac{\alpha}{2(1-\alpha)}} \right)^{1-\alpha} \right]\nonumber 
  % \nonumber \\ 
  &{\leq} \nu_{1}^{\alpha} 2^{\alpha \left[R + R_S - \tilde D_{1-\alpha}\left(\rho_{VUB} \| \rho_{VU} \otimes \rho_B \right)\right]}. 
  \end{align}
%
% where
% \begin{itemize}
%   \item[$(a)$] is  due to the definition of $\Pi_{VUB}$ in \eqref{eq:Pi_VUB_def}, and since $f(x) = x^\alpha$ 
%   is a matrix monotone function $\forall \alpha \in (0,1]$
%   \item[$(b)$] holds since $\Pi_{VUB} \leq \mbone$ 
%   \item[$(c)$] follows from \eqref{eq:first_term_lemma_2_proof}.
% \end{itemize}
\end{proof}

\section{Proof of Proposition \ref{Theorem:1}}
\label{appendix:Th1}
    Let $\Theta_B(m)$ be the state Bob receives:
    \begin{align}
        &\Theta_B(m) = 
        \Tr_{TL}\left[\mcN_{SA \to B}\left(\tilde W^{m} (\sigma^{u(m)}_{SS_0}) \tilde W^{\dag m}   \right) \right]
        \label{Eq:ThetaB}
    \end{align}
    given that the message $m$ was transmitted.
    Furthermore, let $\hat \Theta_B(m)$ be the average state that Bob receives,  when averaged over the subcodebook of $\mcC_V(m)$:
    \begin{align}
        % \hat \Theta_B(m) = \frac{1}{2^{R_S}}\sum_{\ell \in \msL} \Tr_{T K_1 K_0} \left[\mcN_{SA \to B}\left( \rho_{SAT K_1 K_0 }^{v(m,\ell),u(m)} \right)  \otimes\id_{T K_1 K_0}\right].
         \hat \Theta_B(m) = \frac{1}{2^{R_S}}\sum_{\ell = 1}^{2^{R_S}}  \rho_{B}^{v(m,\ell),u(m)}.
         \label{Eq:hatThetaB}
    \end{align}
    % Here, $\hat \Theta_B^{(m)}$ can be interpreted as the average state across the codebook $\mcC_m$, given the message $m$.
    % The motivation behind introducing $\hat{\rho}_B^{(m)}$ is that calculating an upper bound on the average error probability using $\hat{\rho}_B^{(m)}$ is much easier than using $\rho_B^{(m)}$.
    % This will help us to bound the average error probability in a more efficient way.  
    By the symmetry of encoding and decoding, we may assume without loss of generality that Alice sent $m=1$. 
    Consider the pinching-based decoder $\{\beta(m,\ell)\}_{(m,\ell)}$ that has been constructed in Subsection~\ref{Subsection:Decoding}.
    We now bound the average error probability as follows:
    \begin{align}
        \mathbb{E}_{\mbfC}[\bar{p}_e^{(1)}] &= \Pr(\hat{M}\neq 1 \mid M = 1) \nonumber\\
        &\stackrel{(a)}{=} \mbE_{\mbfC} \left\{\Tr\left[\left(\sum_{m' \neq 1, \ell} \beta(m',\ell)\right) \Theta_B(1)\right] \right\}
        \nonumber\\ %\label{eq:theorem1_proof_part1_a}
        &\stackrel{(b)}{\leq} 2\mbE_{\mbfC} \left\{\Tr\left[\left(\sum_{m' \neq 1, \ell} \beta(m',\ell)\right)  \hat \Theta_B(1)\right] \right\}
      \nonumber \\ &\phantom{\stackrel{(b)}{\leq}} 
         + 2\mbE_{\mbfC} \left\{ \left| \sqrt{ \Tr\left[\left(\sum_{m' \neq 1, \ell} \beta(m',\ell)\right) \Theta_B(1) \right]} \right. \right. 
         \left. \left. - \sqrt{\Tr\left[\left(\sum_{m' \neq 1, \ell} \beta(m',\ell)\right)  \hat \Theta_B(1)\right]} \right|^2 \right\} 
        \nonumber\\ %\label{eq:theorem1_proof_part1_b}
        &\stackrel{(c)}{\leq} 2\mbE_{\mbfC} \left\{\Tr\left[\left(\sum_{m' \neq 1, \ell} \beta(m',\ell)\right)  \hat \Theta_B(1)\right] \right\}
        + 2\mbE_{\mbfC} \left\{P\left(\Theta_B(1),  \hat \Theta_B(1) \right)^2 \right\}
        % \\
        % &= 2\sum_{m' \neq 1, \ell} \mbE_{C} \left[\Tr\left[\left( \beta(m',\ell)\right) \hat\Theta_B(1)\right] \right]
        % \nonumber\\ %\label{eq:theorem1_proof_part1_c}
        % &+2\mbE_{C} \left[P\left(\Theta_B(1),  \hat \Theta_B(1) \right)^2\right],  
        % &\leq 4 \mbE_{C} \left[\Tr\left[\left(\mbI - \gamma(1,1)\right) \hat{\rho_{B}}^{(1)}\right] \right] +
        %  8\sum_{m' \neq 1, i} \mbE_{C} \left[\Tr\left[\left(\gamma(m',i)\right) \hat{\rho_{B}}^{(1)}\right] \right] +
        %   2\mbE_{C} \left[P\left(\rho_B^{(1)}, \hat{\rho_B}^{(1)} \right)^2\right] \label{eq:theorem1_proof_part1_d}
        \label{eq:Union_Error}
    \end{align}
    where $(a)$ follows since the error events are disjoint, 
     $(b)$ holds due to the following: Note that
     $(x-y)^2\geq 0$ implies $(x+y)^2\leq 2(x^2+y^2)$.
     Therefore,
     % is obtained by rewriting the left-hand side term as 
     $z = \left(\sqrt{w} + \sqrt{z}  -\sqrt{w} \right)^2\leq 2w+2\abs{\sqrt{z}  -\sqrt{w}}^2$.
     The inequality follows by taking 
     $z=\Tr\left[\left(\sum_{m' \neq 1, \ell} \beta(m',\ell)\right) \Theta_B(1)\right]$
     and $w=\Tr\left[\left(\sum_{m' \neq 1, \ell} \beta(m',\ell)\right)  \hat \Theta_B(1)\right]$.
     %and for non-commutative x and y: 
     % \begin{align*}
     %     0 & \leq (x-y)^2 
     %     \\
     %     0 & \leq x^2 -xy - yx +y^2 
     %     \\ 
     %     0  &\leq 2x^2 +2y^2- x^2 -xy - yx -y^2
     %     \\
     %     x^2 +xy + yx +y^2 & \leq 2(x^2 +y^2)
     %     \\
     %     (x +y)^2 &\leq 2(x^2 +y^2),
     % \end{align*}
     % when the last inequality is applied.
     Then, $(c)$ is obtained by $\sum_{m' \neq 1, \ell} \beta(m',\ell) \leq \mbone$
     and 
     $\left| \sqrt{\Tr[\Delta \sigma]} - \sqrt{\Tr[\Delta \rho]} \right| \leq P(\sigma, \rho)$
      for every pair of quantum states $\rho , \sigma \in \mathscr{D}(\mcH)$ and  $0\leq \Delta \leq \mbone$, based on
         \cite[Fact 7]{Anshu2020}.
        
    Consider the first term on the right-hand side of \eqref{eq:Union_Error}:
    \begin{align}
      \sum_{m' \neq 1, \ell} \mbE_{\mbfC} \left\{\Tr\left[\left( \beta(m',\ell)\right) \hat \Theta_B(1)\right] \right\}
      &=   \frac{1}{2^ {R_S}}\sum_{\ell'}    \sum_{m' \neq 1, \ell} \mbE_{C} \left\{\Tr\left[\left( \beta(m',\ell)\right) \rho_{B}^{v(1,\ell'), u(1) }\right] \right\}
       \nonumber \\ 
      &=   \sum_{m' \neq 1, \ell} \mbE_{\mbfC} \left\{\Tr\left[\left( \beta(m',\ell)\right) \rho_{B}^{v(1,1), u(1)}\right] \right\}
       \nonumber \\ 
      &\leq  \mbE_{\mbfC} \left\{\Tr[\left( \mbone - \beta(1,1) \right)  \rho_B^{v(1,1), u(1)} ]\right\}
       \nonumber \\ 
    &\leq  \mbE_{\mbfC} \left\{\Tr[\left( \mbone - \left( \sum_{m',\ell} \gamma(m',\ell) \right)^{-\frac{1}{2}} \gamma(1,1)  \left( \sum_{m',\ell} \gamma(m',\ell) \right)^{-\frac{1}{2}} \right)  \rho_B^{v(1,1), u(1)} ]\right\}
       \nonumber \\ 
      &\leq  2\mbE_{\mbfC} \left\{\Tr[\left( \mbone - \gamma(1,1) \right)  \rho_B^{v(1,1), u(1)} ]\right\}
      \nonumber\\ &\phantom{\leq}
      + 4 \sum_{(m',\ell) \neq (1,1)} \mbE_{\mbfC} \left\{\Tr[\left( \gamma(m',\ell) \right)  \rho_B^{v(1,1), u(1)} ]\right\},
      \label{eq:theorem1_proof_part2}
    \end{align}
    where the last inequality is based on the Hayashi-Nagaoka operator inequality \cite{HayashiNagaoka2003}:
$\mbone -\left(S+T \right)^{-\frac{1}{2}} S \left(S+T \right)^{-\frac{1}{2}} \leq 2 \left(\mbone-S \right) + 4T$, for every
     $0 \leq S \leq \mbone$ and $ T\geq 0$  on a Hilbert space $\mcH$. 
    % then
    % \begin{align}
    %   \mbone -\left(S+T \right)^{-\frac{1}{2}} S \left(S+T \right)^{-\frac{1}{2}} \leq 2 \left(\mbone-S \right) + 4T.
    % \end{align}
    In our case, we set $S = \gamma(1,1)$ and $T = \sum_{(m',\ell) \neq (1,1)} \gamma(m',\ell)$.

    Our next steps follow similar considerations as in \cite{Anshu2020,zivarifard2024covert}. Specifically, we obtain an upper bound on each term on the right-hand side of \eqref{eq:theorem1_proof_part2} by using the properties established in  Appendix~\ref{appendix:proof_of_lem1} and Appendix~\ref{appendix:proof_of_lem2}. 
    We bound  the first term on the right-hand side of \eqref{eq:theorem1_proof_part2} as follows:
    \begin{align}
      &2\mbE_{\mbfC} \left\{\Tr[\left( \mbone - \gamma(1,1) \right)  \rho_B^{v(1,1), u(1)} ]\right\}
      \nonumber \\ 
      &\stackrel{(a)}{=} 2\mbE_{\mbfC} \left\{\Tr[\left( \mbone - \Tr_{VU}[\Pi_{VUB} (\ketbra{v(1,1)}{v(1,1)} \otimes \ketbra{u(1)}{u(1)} \otimes \mbone_B)]  \right)  \rho_B^{v(1,1), u(1)} ]\right\} 
      \nonumber \\ 
      &= 2\mbE_{\mbfC} \left\{\Tr[\left( \mbone - \Tr_{VU}[\Pi_{VUB} (\ketbra{v(1,1)}{v(1,1)} \otimes \ketbra{u(1)}{u(1)} \otimes \mbone_B)]  \right) \left( \id_{VU} \otimes \rho_B^{v(1,1), u(1)} \right) ]\right\} 
      \nonumber \\ 
      & \stackrel{(b)}{=} 2\mbE_{\mbfC} \left\{\Tr[\left( \mbone - \Pi_{VUB} (\ketbra{v(1,1)}{v(1,1)} \otimes \ketbra{u(1)}{u(1)} \otimes \mbone_B)  \right) \left( \id_{VU}  \otimes \rho_B^{v(1,1), u(1)} \right) ]\right\} 
      \nonumber \\ 
      &= 2 \left[\Tr[ \left( \mbone - \Pi_{VUB} \right)\mbE_{\mbfC} \left\{\ketbra{v(1,1)}{v(1,1)} \otimes  \ketbra{u(1)}{u(1)} \otimes \rho_B^{v(1,1), u(1)} \right\} ]\right]
      \nonumber \\ 
      &= 2 \left[\Tr[ \left( \mbone - \Pi_{VUB}  \right)\left( \sum_{v,u} p_{VU}(v,u) \ketbra{v}{v} \otimes \ketbra{u}{u} \otimes \rho_{B}^{v,u} \right) ]\right]
      \nonumber \\ 
      &\stackrel{(c)}{=} 2 \Tr[\left( \mbone - \Pi_{VUB}\right) \rho_{VUB}] 
      \nonumber \\ 
      &\stackrel{(d)}{\leq} 2 \cdot \nu_{1}^{\alpha} 2^{\alpha \left[R + R_S - \tilde D_{1-\alpha}\left(\rho_{VUB} \| \rho_{VU} \otimes \rho_B \right)\right]}.
      \label{eq:theorem1_proof_part2_First}
    \end{align}
where
    \begin{itemize}
      \item[$(a)$] is obtained by the definition of $\gamma(1,1)$ as in \eqref{eq:gamma_def}, 
      \item[$(b)$] follows from trace linearity, 
      \item[$(c)$] from taking the expectation with respect to the random codebook $\mbfC$, and 
      \item[$(d)$] from Lemma \ref{lem:2}.
    \end{itemize}

    As for the second term in \eqref{eq:theorem1_proof_part2}:
    \begin{align}
    &4 \sum_{(m',\ell) \neq (1,1)} \mbE_{\mbfC} \left\{\Tr[\left( \gamma(m',\ell) \right)  \rho_B^{v(1,1), u(1)} ]\right\} 
    \nonumber \\ 
    &\stackrel{(a)}{=}   4 \sum_{(m',\ell) \neq (1,1)}  \mbE_{\mbfC} \left\{\Tr[\left(\Tr_{VU}[\Pi_{VUB} (\ketbra{v(m',\ell)}{v(m',\ell)} \otimes \ketbra{u(m')}{u(m')} \otimes \mbone_B)]  \right)  \rho_B^{v(1,1), u(1)} ] \right] 
    \nonumber \\ 
    &= 4 \sum_{(m',\ell) \neq (1,1)}  \mbE_{\mbfC} \left\{\Tr[\left(\Pi_{VUB} (\ketbra{v(m',\ell)}{v(m',\ell)} \otimes \ketbra{u(m')}{u(m')}  \otimes \rho_B^{v(1,1), u(1)}   \right) ] \right\}
    \nonumber \\ 
    &\stackrel{(b)}{=} 4 \sum_{(m',\ell) \neq (1,1)}  \Tr \left[\mbE_{\mbfC} \left\{\left(\Pi_{VUB} (\ketbra{v(m',\ell)}{v(m',\ell)} \otimes \ketbra{u(m')}{u(m')}  \otimes \rho_B^{v(1,1), u(1)}   \right) \right\} \right] 
    \nonumber \\ 
    &=   4 \sum_{(m',\ell) \neq (1,1)}  \Tr \; \Pi_{VUB} \left( \sum_{v, v', u, u'} \Pr\left(v(1,1) = v, v(m',\ell)=v', u(1) = u, u(m') = u'\right)  \ketbra{v'}{v'} \otimes \ketbra{u'}{u'} \otimes \rho_{B}^{v,u} \right) 
    \nonumber \\ 
    &=   4\sum_{(m',\ell) \neq (1,1)}  \Tr \left[\Pi_{VUB} \left( \sum_{v, v', u, u'} p_{VU}(v,u)p_{VU}(v',u') \ketbra{v'}{v'} \otimes \ketbra{u'}{u'} \otimes \rho_{B}^{v,u} \right) \right]
    \nonumber \\ 
    &=   4 \sum_{(m',\ell) \neq (1,1)}  \Tr \left[\Pi_{VUB} \left( \sum_{v',u'} p_{VU}(v',u') \ketbra{v'}{v'} \otimes \ketbra{u'}{u'} \right) \otimes \left(\sum_{v,u}  p_{VU}(v,u) \rho_{B}^{v,u} \right) \right]
    \nonumber \\ 
    &\stackrel{(c)}{=}   4 \sum_{(m',\ell) \neq (1,1)} \Tr \left[\Pi_{VUB}\left(\rho_{VU} \otimes \rho_B \right)  \right] 
    \nonumber \\ 
    &\stackrel{(d)}{\leq}   4 \cdot 2^{R+R_S} \Tr \left[\Pi_{VUB}\left(\rho_{VU} \otimes \rho_B \right)  \right] 
    \nonumber \\ 
    &\stackrel{(e)}{\leq} 4 \cdot \nu_{1}^{\alpha} 2^{\alpha \left[R + R_S - \tilde D_{1-\alpha}\left(\rho_{VUB} \| \rho_{VU} \otimes \rho_B \right)\right]} 
    .
    \label{eq:theorem1_proof_part2_Second}
  \end{align}
  where
  \begin{itemize}
    \item[$(a)$] holds by the definition of $\gamma(m,\ell)$  in \eqref{eq:gamma_def},
    \item[$(b)$] by  linearity, 
    \item[$(c)$] follows by taking the expectation with respect to the random codebook, 
    \item[$(d)$] is obtained due to the fact that $v(m,\ell)$ is conditionally independent of $v(m',\ell')$, given $u(m)$, for $(m, \ell)\neq (m', \ell')$,
    \item[$(e)$] follows from Lemma \ref{lem:2}.
  \end{itemize}
  By plugging \eqref{eq:theorem1_proof_part2_First}-\eqref{eq:theorem1_proof_part2_Second} into \eqref{eq:theorem1_proof_part2}, we obtain the following bound on the confusion error term:
  \begin{align}
    \sum_{m' \neq 1, \ell} \mbE_{\mbfC} \left\{\Tr\left[\left( \beta(m',\ell)\right) \hat{\Theta}_{B}{(1)}\right] \right\}
    &\leq 6 \cdot \nu_{1}^{\alpha} 2^{\alpha \left[R + R_S - \tilde D_{1-\alpha}\left(\rho_{VUB} \| \rho_{VU} \otimes \rho_B \right)\right]}.
    \label{eq:prop2_a}
  \end{align} 
  Hence, the expected error probability is bounded by 
      \begin{align}
        \mathbb{E}_{\mbfC}[\bar{p}_e^{(1)}] &\leq  12 \cdot \nu_{1}^{\alpha} 2^{\alpha \left[R + R_S - \tilde D_{1-\alpha}\left(\rho_{VUB} \| \rho_{VU} \otimes \rho_B \right)\right]}
        + 2\mbE_{\mbfC} \left\{P\left(\Theta_B(1),  \hat \Theta_B(1) \right)^2 \right\}
        \label{eq:Union_Error_First_Result}
    \end{align}
    by \eqref{eq:Union_Error}.

  It remains to bound the last term for $\alpha \in (0,\frac{1}{2})$:
  \begin{align}
      &\mbE_{\mbfC} \left\{P\left(\Theta_B{(1)}, \hat{\Theta}_B{(1)} \right)^2\right\} 
      \nonumber \\ 
      &\stackrel{(a)}{=} \mbE_{\mbfC} \left\{P\left( \Tr_{T  L}\left[\mcN_{SA \to B}\left(\tilde W^{(1)} (\sigma^{u(1)}_{SS_0}) \tilde W^{\dag (1)}   \right) \right],
      \frac{1}{2^{R_S}}\sum_{\ell = 1}^{2^{R_S}}  \rho_{B}^{v(1,\ell),u(1)}\right)^2 \right\}  \nonumber
      \nonumber \\ 
      &= \mbE_{\mbfC} \left\{P\left( \Tr_{TL}\left[\mcN_{SA \to B}\left(\tilde W^{(1)} (\sigma^{u(1)}_{SS_0 }) \tilde W^{\dag (1)}   \right) \right],
      \frac{1}{2^{R_S}}\sum_{\ell = 1}^{2^{R_S}} \Tr_{T} \left[ \mcN_{SA \to B} \left(\rho_{SAT}^{v(1,\ell),u(1)} \right) \right]\right)^2 \right\}  
      \nonumber \\ 
      &\stackrel{(b)}{\leq} \mbE_{\mbfC} \left\{P\left( \phi_{SATL}^{(1)} , 
     \tilde W^{(1)} (\sigma^{u(1)}_{SS_0}) \tilde W^{\dag (1)} \right)^2\right\}
      \nonumber \\ 
      &\stackrel{(c)}{=} \mbE_{\mbfC} \left\{ P\left(\frac{1}{2^{R_S}} \sum_{\ell = 1}^{2^{R_S}} \rho_{S}^{v(1,\ell),u(1)},  \sigma_{S}^{u(1)} \right)^2\right\}
      \nonumber \\ 
      &\stackrel{(d)}{=} \mbE_{\mbfC} \left\{ P\left(\tau_{S|\mcC_1}, \sigma_S^{u(1)} \right)^2\right\}
      \nonumber \\ 
      &= \mbE_{\mbfC} \left\{1 - F^2\left(\tau_{S|\mcC_1}, \sigma_S^{u(1)} \right)\right]
      \nonumber \\ 
      & \stackrel{(e)}{\leq} \mbE_{\mbfC} \left\{ 1-2^{-\tilde D_{1+\alpha} \left(\tau_{S|\mcC_1}\|\sigma_S^{u(1)} \right)}\right\}
      \nonumber \\ 
      & \stackrel{(f)}{\leq} \ln2 \mbE_{\mbfC} \left\{ \tilde D_{1+\alpha}\left(\tau_{S|\mcC_1}\|  \sigma_S^{u(1)} \right)\right\}
      \nonumber \\ 
      &\stackrel{(g)}{\leq} \frac{1}{\alpha}  \frac{\nu_2^{\alpha}}{2^{\alpha R_S}}  2^{\alpha \tilde{D}_{1+\alpha}(\rho_{VUS} \| \rho_{V-U-S})},
      \label{eq:prop2_b}
  \end{align}
  where
  \begin{itemize}
    \item[$(a)$] follows by substituting the definition of $\Theta_B(m)$ and $\hat \Theta_B(m)$ in \eqref{Eq:ThetaB} and \eqref{Eq:hatThetaB}, respectively,
    
    \item[$(b)$]  from the fact that monotonicity of the purified distance, with respect to the % does not increase under 
    quantum channel $\Tr_{T}\mathcal{N}_{SA \to B}(\cdot)$, and
    
    \item[$(c)$] from \eqref{eq:enc_purd}.
    \end{itemize}
Furthermore,
    \begin{itemize}
    
    \item[$(d)$] holds as we introduce the notation $\tau_{S| \mathcal{C}_{1}}$ to represent the average state: $\tau_{S| \mathcal{C}_{1}}\equiv \frac{1}{2^{R_{S}}}\sum_{\ell }\rho_{S}^{v(1,\ell),u(1)}$, and since $\sum_{v} p_{V|U}(v|u) \left[ \Tr_{A}{\rho_{SA}^{v,u}} \right] =  \Tr_{S_0}{\sigma_{S S_0}^u}$.
    \item[$(e)$] as $\tilde D_{\alpha}(\cdot \| \cdot)$ is  monotonically increasing in $\alpha$, thus  
    $ F^2(\rho, \sigma) = 2^{-\tilde D_{1/2}(\rho\|\sigma)} \geq 2^{-\tilde D_{1+\alpha}(\rho\|\sigma)}$ for every pair of quantum states $\rho, \sigma \in \mathscr{D}(\mathcal{H})$ and  $\alpha > -\frac{1}{2}$ (see \cite[Corollary 4.3]{tomamichel2015quantum}).
    
    \item[$(f)$] follows from the inequality $1-2^{-x} \le x \ln 2$, and
    
    \item[$(g)$]  by applying Lemma \ref{lem:1}.
\end{itemize}
    % An upper bound for the second term in \eqref{eq:theorem1_proof_part1_c} is obtained by using Lemmas 6 and 7 in \cite{Anshu2020}.
    % \begin{align}
    %   \mbE_{C} &\left[P\left(\rho_B^{(1)}, \hat{\rho}_B^{(1)} \right)^2\right] \\ \nonumber
    %   &\leq \frac{1}{\alpha} \left( \frac{\nu_2}{2^{R_S}}\right)^{\alpha} 2^{\alpha \tilde{D}_{1+\alpha}\left( \rho_{VS|U}^{\otimes n} \| \rho_{V|U}^{\otimes n} \otimes \rho_{S|U}^{\otimes n}\right)}
    % \end{align}
    % where $\nu_2$ is the number of distinct eigenvalues of $\rho_{S|U}$.
    Proposition \ref{Theorem:1} then follows from \eqref{eq:Union_Error_First_Result}-\eqref{eq:prop2_b}.
    \qed

\section{Cardinality Bounds}
\label{appendix:cardinality}
Recall that the achievable rate established in Theorem~\ref{Theorem:non-causal} is given by:
% the regularized formula:
\begin{align}
        \rqad_{\text{n-c}} \left( \mcN\circ T \right) = I(VU;B)_{\rho} - I(V;S|U)_{\rho},
\end{align}
% where the optimization is over the input distribution $p_{VU}$, the action states $\{\sigma_G^u\}$, and encoding channels $\mathcal{F}_{S_0 \to A}^v$.
where the optimization is over the input distribution $p_{VU}$, the pure action states $\{\ket{\sigma_G^u}\}$, and the conditional input states $\{\rho_{SA}^{v,u}\}$, such that $\Tr_{A}[\rho_{VUAS}] = \Tr_{S_0}[\sigma_{VUSS_0}]$.
To bound the alphabet size of the auxiliary random variables $U$ and $V$ required for this optimization, we use the Fenchel-Eggleston-Carathéodory lemma \cite{Eggleston1966} and arguments similar to those in \cite{Pereg2022}.

\begin{align}
    L_U &= |\mathcal{H}_S|^2|\mathcal{H}_{S_0}|^2,  \\
    L_{0} &= |\mathcal{H}_S|^2|\mathcal{H}_A|^2,  \\
    L_{V} &= L_U \cdot L_0 .
\end{align}

First, we consider the outer variable $U$. Fix the distribution $p_{V|U}(v|u)$ and the conditional input states $\{\rho_{SA}^{v,u}\}$. Consider the input ensemble $\{p_{U}(u), \sigma_G^u\}$. A Hermitian matrix of dimension $d$ is specified by $d^2$ real parameters. Since a density matrix is Hermitian and has unit trace, it is specified by $d^2-1$ real parameters.
We define a map $f_0: \mathcal{U} \to \mathbb{R}^{L_U}$ by
\begin{align}
    f_0(u) = \Big( w(\sigma_{SS_0}^u), -H(B|V, U=u)_\rho + I(V;S|U=u)_\rho \Big)
\end{align}
where $w(\cdot)$ denotes the vector representation of the state, $\ket{\sigma_{SS_0}^u} = T_{G \to SS_0}\ket{\sigma_G^u}$.
The map $f_0$ can be extended to probability distributions as follows,
\begin{align}
    F_0: p_U \mapsto \sum_{u \in \mathcal{U}} p_U(u) f_0(u) \,.
\end{align}
According to the Fenchel-Eggleston-Carathéodory lemma \cite{Eggleston1966}, any point in the convex closure of a connected compact set within $\mathbb{R}^d$ belongs to the convex hull of $d$ points in the set. Since the map $F_0$ is linear, it maps the set of distributions on $\mathcal{U}$ to a connected compact set in $\mathbb{R}^{L_U}$. Thus, for every $p_U$, there exists a probability distribution $p_{\tilde{U}}$ on a subset $\tilde{\mathcal{U}} \subseteq \mathcal{U}$ of size $L_U$, such that $F_0(p_{\tilde{U}}) = F_0(p_U)$.
This preserves the average environment state $\sigma_{SS_0} = \sum_u p_U(u) \sigma_{SS_0}^u$, the average channel input $\rho_{SA} = \sum_{v,u} p_U(u)p_{V|U}(v|u) \rho_{SA}^{v,u}$, the output state $\rho_B=\mcN_{SA \to B}(\rho_{SA})$, and the quantity $I(VU; B)_\rho - I(V; S|U)_\rho = H(B)_\rho - H(B|VU)_\rho + I(V; S|U)_\rho$.
We deduce that the alphabet size can be restricted to $|\mathcal{U}| \leq L_U$.

We move to the alphabet size of the inner variable $V$. Fix $u \in \mathcal{U}$ and the corresponding distribution $p_{U}$. Define the map $g_u: \mathcal{V} \to \mathbb{R}^{L_{0}}$ by
\begin{align}
    g_u(v) = \Big( w(\rho_{SA}^{v,u}), -H(B|V=v, U=u)_\rho + H(S|V=v, U=u)_\rho \Big)
\end{align}
Now, the extended map is
\begin{align}
    G_u: p_{V|U}(\cdot|u) \mapsto \sum_{v \in \mathcal{V}_u} p_{V|U}(v|u) g_u(v) \,.
\end{align}

By the Fenchel-Eggleston-Carathéodory lemma \cite{Eggleston1966}, for every $u$, there exists a conditional distribution on a subset $\tilde{\mathcal{V}} \subseteq \mathcal{V}$ of size $L_{0}$, such that $G_u(p_{\tilde{V}|U}) = G_u(p_{V|U})$.
The alphabet cardinality is bounded by
\begin{align}
    |\mathcal{V}| \leq  L_{V},
\end{align}
while preserving the input state $\sigma_{G}^u$, the channel input state $\rho_{SA}^{u}$
, and $I(VU; B)_\rho - I(V; S|U)_\rho$. 

\section{Derivations for Section~\ref{section:examples}}
\label{appendix:example_derivations}
We provide detailed derivations for the depolarizing memory example of Section~\ref{section:examples}.

\subsection{Causal CSI}

The action isometry $T_{G \to SS_0}$ thus produces the bipartite state below:
\begin{align}
    \ket{\sigma_{SS_0}^u} &=T_{G \to SS_0}\ket{\sigma_G^u}
    \nonumber \\
    &= \ket{u}_{S} \otimes \left(\sqrt{1-p}\ket{0}_{S_0} + (-1)^u \sqrt{\frac{p}{3}}\ket{3}_{S_0}\right) + \sqrt{\frac{p}{3}}\ket{\bar{u}}_S \otimes \left(\ket{1}_{S_0} + i(-1)^u\ket{2}_{S_0} \right)
        \nonumber \\
    &= \sqrt{1-\frac{2p}{3}}\ket{u}_{S} \otimes \ket{\theta_{03}^u}_{S_0} + \sqrt{\frac{2p}{3}}\ket{\bar{u}}_S \otimes \ket{\theta_{12}^u}_{S_0}
\end{align}
where $\bar{u}=1-u$ denotes the flipped bit, $\ket{\theta_{03}^u}=\sqrt{\frac{1}{1-2p/3}}\left(\sqrt{1-p}\ket{0} + (-1)^u\sqrt{\frac{p}{3}}\ket{3}\right)$
and $\ket{\theta_{12}^u}=\frac{1}{\sqrt2}\left(\ket{1} +i(-1)^u\ket{2} \right) $.
Note that the states $\ket{\theta_{03}^u}$ and $\ket{\theta_{12}^u}$ are orthogonal.

The resulting channel input is thus
\begin{align}
    \ket{\rho_{SA}^u} &= \mbone_S \otimes F^u_{S_0 \to A}\ket{\sigma_{SS_0}^{u}}
    \\
    &= \sqrt{1-\frac{2p}{3}}\ket{u}_{S} \otimes \ket{\perp}_{A} + \sqrt{\frac{2p}{3}}\ket{\bar{u}}_S \otimes \ket{u}_{A}
\end{align}
and the reduced state is
$\rho_A^u = \Tr_S(\rho_{SA}^u)=\left({1-\frac{2p}{3}}\right)\ketbra{\perp}+\frac{2p}{3}\ketbra{u}$.

The resulting output state is
\begin{align}
    \mcN_{SA \to B}(\ketbra{\rho_{SA}^u}) &= (\mbone\otimes \bra{\perp}) \ketbra{\rho_{SA}^u} (\mbone\otimes \ket{\perp})   +
      \mathcal{D}_p  \left(  \Pi_{01} \cdot   \rho_{A}^u \cdot   \Pi_{01} \right)
      \\
       &=\left({1-\frac{2p}{3}}\right)\ketbra{u}+ \frac{2p}{3}\mathcal{D}_p  \left(  \ketbra{{u}} \right)
       \\
       &=\left({1-\frac{2p}{3}}\right)\ketbra{u}+ \frac{2p}{3}  \left[  \left({1-\frac{2p}{3}}
       \right)\ketbra{{u}}
       +\frac{2p}{3}\ketbra{\bar{u}} \right]
       \\
       &=({1-\delta^2})\ketbra{{u}}+ \delta^2\ketbra{\bar{u}} ,
\end{align}
where $\delta = \frac{2p}{3}$.

\subsection{Non-Causal CSI}

One can verify that the marginal constraint~\eqref{eq:marginal_constraint} is satisfied:
\begin{align}
    \Tr_{VA}\!\left[\rho^u_{VSA}\right]
    &= (1-\delta)\ketbra{u}_S + \delta\ketbra{\bar{u}}_S
     = \Tr_{S_0}\!\left[\ketbra{\sigma^u_{SS_0}}\right].
     \label{Equation:Example_Environement_State}
\end{align}
The resulting output states $\rho^{v,u}_B = \mcN_{SA\to B}(\rho^{v,u}_{SA})$ are given
as follows.
For $V=0$, the idle input forwards $S$ to the output while the rewrite input undergoes
depolarization:
\begin{align}
    \mcN_{SA\to B}(\rho^{0,u}_{SA})
    &= \frac{(1-\delta)\ketbra{u}
           + \delta(1-\alpha)\,\mcD_p(\ketbra{u})}{1-\delta\alpha} \nonumber \\
    &= \left(1-\frac{\delta^2(1-\alpha)}{1-\delta\alpha}\right)\ketbra{u} + \frac{\delta^2(1-\alpha)}{1-\delta\alpha}\ketbra{\bar{u}}.
\end{align}
For $V=1$, Alice transmits $\ket{\perp}_A$, hence the output is simply the environment state, i.e.,
$\mcN_{SA\to B}(\rho^{1,u}_{SA}) = \ketbra{\bar{u}}_B$.

We now evaluate the rate formula~\eqref{eq: R-nonc}.
The first mutual information term is given by
\begin{align}
    I(VU;B)_\rho
    = 1 - (1-\delta\alpha)\,h_2\!\left(\frac{\delta^2(1-\alpha)}{1-\delta\alpha}\right).
    \label{eq:IVU_B_nc}
\end{align}
Next, consider the ``penalty" term, $I(V;S|U)_\rho$.
First, we note that %$\rho^u_S = (1-\delta)\ketbra{u} + \delta\ketbra{\bar{u}}$,  hence
 $H(S|U)_\rho=h_2(\delta)$ by \eqref{Equation:Example_Environement_State}.
Furthermore, the entropies of
the environment states are
$\rho^{0,u}_S $ and $\rho^{1,u}_S $ are
$H(S|V=0,U=u)_\rho= h_2\!\left(\frac{\delta(1-\alpha)}{1-\delta\alpha}\right)$, and
$H(S|V=1,U=u)_\rho = H(\ketbra{\bar{u}})=0$, respectively.
This yields
\begin{align}
    I(V;S|U)_\rho
    = h_2(\delta)
    - (1-\delta\alpha)\,h_2\!\left(\frac{\delta(1-\alpha)}{1-\delta\alpha}\right).
    \label{eq:IVS_U_nc}
\end{align}

\end{appendices}

% Acknowledgment
%-----------------------------------------------
\section*{Acknowledgments}
The authors thank Gerhard Kramer (TUM)
for useful discussions.
Korenberg and Pereg were supported by  Israel Science Foundation (ISF), 
 Grants n. 939/23 and 2691/23,
Deutsche Forschungsgemeinschaft (DFG) through the
 German-Israeli Project Cooperation (DIP) 
 n. 2032991, Ollendorff-Minerva Center of the Technion 
 n. 86160946,  Nevet Program of the 
 Helen Diller Quantum Center at the Technion 
 n. 	2033613, and
  the Planning and Budgeting Committee of the Council for Higher Education of Israel
(VATAT) for the QERNEL Quantum  Computing Research Hub n.
2072651.
Pereg was also supported by Chaya Chair n. 8776026. %, and the Chaya %Career Advancement 

\bibliography{references}
\end{document}

%% file: vertical_fig_QADC.tex
\begin{tikzpicture}[%
    >=latex,
    node distance=3.0cm and 3cm,
    on grid,
    semithick,
    font=\large,
    box/.style={draw, minimum width=2cm, minimum height=1.5cm, align=center},
    largebox/.style={draw, minimum width=2cm, minimum height=1.5cm, align=center}
]

% Main bottom nodes: m, Enc, N, Dec, \hat{m}
\node (m) {$M$};
\node[box, right=6cm of m] (Enc) {Encoder};
\node[box, right=5cm of Enc] (N) {$\mathcal{N}_{SA \to B}^{\otimes n}$};
\node[box, right=5cm of N] (Dec) {Decoder};
\node[right=4cm of Dec] (mhat) {$\hat{M}$};

% Draw bottom horizontal arrows
\draw[->] (m) -- (Enc) node[midway, coordinate] (mEncMid) {};
\draw[->, red] (Enc) -- node[above] {$A^{n}$} (N);
\draw[->, blue] (N) -- node[above] {$B^{n}$} (Dec);
\draw[->] (Dec) -- (mhat);

% The T box is positioned above the Encoder
\node[largebox, above=3.5cm of Enc] (T) {$T_{G \to S S_0}^{\otimes n}$};

% Action Encoder is placed relative to the position of T
\node[box, left = 4.5cm of T] (ActionEnc) {Action Encoder};

% Draw arrow from Action Encoder to T
\draw[->] (ActionEnc) -- node[pos=0.5,above] {$G^{n}$} (T);

% Draw arrow straight down from T to the Encoder
\draw[->] (T.south) -- node[pos=0.6,left] {${S_0}^n$} (Enc.north);

% Draw the bent arrow from the right of T, then down to N
\draw[->, green!40!black] (T.east) -| node[pos=0.85, right] {$S^n$} (N.north);

% The upward arrow from the M signal path remains the same
\coordinate (cross) at (ActionEnc |- m);
\fill (cross) circle (2pt);
\draw[->] (cross) -- (ActionEnc);

\coordinate (Alice_mid) at ($(ActionEnc.south east)!0.5!(Enc.north west) - (0.9,0.9cm)$);
\coordinate (AE_nwest) at (m.south west |- ActionEnc.north west);
\coordinate (AE_neast) at (Alice_mid |- AE_nwest);
\coordinate (E_neast) at (Enc.north east |- Alice_mid);
\coordinate (E_seast) at ($(Enc.south east)$);
\coordinate (M_west) at ($(m.south west)$);

% --- ENCLOSURES START ---
\begin{pgfonlayer}{background}
    % Define the padding amount
    \def\padding{0.4cm}
    
    % Draw the custom path with manual padding added to each coordinate
    \draw[draw=red, line width=1pt, rounded corners, fill=pink, opacity=0.2]
        ($(AE_neast) + (\padding, \padding)$) coordinate (alice_top) -- 
        ($(Alice_mid) + (\padding, \padding)$) -- 
        ($(E_neast) + (\padding, \padding)$) --  
        ($(E_seast) + (\padding, -\padding)$) -| 
        ($(M_west) + (-\padding, \padding)$)  -- 
        ($(AE_nwest) + (-\padding, \padding)$) 
         -- cycle;

    % BOB ENCLOSURE
    \node[draw=blue, line width=1pt, rounded corners, fill=blue!20, opacity=0.2, inner sep=0.5cm,
          fit=(Dec) (mhat), label={[yshift=-2.5cm, color=blue!70!black, font=\large]above:Bob}] (bob_enclosure) {};
\end{pgfonlayer}

% Add the label for Alice above the new, larger shape
\node[above=-5.8cm, color=red!70!black] at (alice_top) {Alice};
% --- ENCLOSURES END ---

\end{tikzpicture}